\documentclass[11pt]{article}
\usepackage[utf8]{inputenc}
\usepackage{style}
\usepackage{framed}
\usepackage{nicefrac}
\usepackage[parfill]{parskip}
\usepackage{setspace}
\usepackage{hyperref}
\usepackage{mdwlist}
\usepackage{mdframed}
\usepackage[ruled,linesnumbered]{algorithm2e}
\usepackage[utf8]{inputenc}
\usepackage{style}
\usepackage{mdframed}
\usepackage{mdwlist}
\usepackage{nicefrac}
\usepackage[parfill]{parskip}
\usepackage{setspace}
\usepackage[colorlinks]{hyperref}
\usepackage[ruled,linesnumbered]{algorithm2e}
\usepackage{amsmath}
\usepackage{amssymb}
\usepackage{complexity}

\newcommand{\Ex}{\E}
\newcommand{\cH}{\mathcal{H}}
\renewcommand{\cP}{{\cal P}}

\allowdisplaybreaks

\usepackage{microtype}

\usepackage{thm-restate}

\usepackage{boxedminipage}

\newcommand{\Paren}[1]{\left(#1 \right )}

\newcommand{\Brac}[1]{\left[#1\right]}

\newcommand{\Abs}[1]{\left\lvert#1\right\rvert}

\newcommand{\Inf}[2]{{\sf Inf}_{#1}\left[#2\right]}

\newcommand{\wh}[1]{\widehat{#1}}

\providecommand{\eps}[0]{\ensuremath{\varepsilon}\xspace}

\newcommand{\Var}{\mathrm{Var}}

\newcommand{\dist}{\mathsf{dist}}

\newcommand{\MaxCut}{\mathsf{MaxCut}}
\providecommand{\Prb}{\mathbb P}
\providecommand{\1}{\mathbf 1}
\providecommand{\R}{\mathbb R}
\renewcommand{\P}{\Pr}
\providecommand{\E}{\mathbb E}

\providecommand{\Set}[1]{\left\{#1\right\}}
\providecommand{\Norm}[1]{\left\lVert#1\right\rVert}
\providecommand{\Iprod}[1]{\left\langle#1\right\rangle}
\providecommand{\sign}{\ensuremath{\mathrm{sign}}}

\providecommand{\cM}{\mathcal{M}}
\providecommand{\cR}{\mathcal{R}}

\renewcommand{\cL}{\mathcal{L}}
\renewcommand{\cS}{\mathcal{S}}

\newcommand{\one}{{\bf 1}}

\newcommand{\Cube}{\{-1,1\}^R}
\newcommand{\ZeroCube}{\{0,1\}^R}
\newcommand{\Cut}{\operatorname{Cut}}
\newcommand{\pos}[1]{\left[#1\right]^+}
\newcommand{\Stab}{\mathbb S}
\newcommand{\bbR}{\mathbb R}
\newcommand{\Influence}{\operatorname{Inf}}
\newcommand{\sgn}{\operatorname{sign}}
\newcommand{\wH}{\widehat{\mathcal{H}}}
\newcommand{\tSDP}{{\sf SDP}_{\mbox{\tiny$\triangle$}}}
\allowdisplaybreaks
\title{Optimal Algorithms and Hardness for Max-Cut \\ in Bounded Degree Graphs}
\author{Rares-Darius Buhai\thanks{EPFL. Supported by the Swiss State Secretariat for Education, Research and
Innovation (SERI) under contract number MB22.00054.} 
\and Suprovat Ghoshal\thanks{Indiana University Bloomington.} 
\and Neng Huang\thanks{TU-Dresden. Work was mostly done when NH was affiliated with the University of Michigan. NH has been funded by the European Union (Project POCOCOP, ERC Synergy Grant 101071674). Views and opinions expressed are however those of the author(s) only and do not necessarily reflect those of the European Union or the European Research Council Executive Agency. Neither the European Union nor the granting authority can be held responsible for them.} 
\and Pravesh Kothari\thanks{Princeton University.} 
\and Euiwoong Lee\thanks{University of Michigan. Supported in part by NSF award CCF-2236669.} 
\and Konstantin Makarychev\thanks{Northwestern University. Supported in part by NSF Award EECS-2216970.} 
\and Yury Makarychev\thanks{Toyota Technological Institute at Chicago. Supported by NSF award ECCS-2216899.}}
\begin{document}

\maketitle
\begin{abstract}

We give a polynomial-time algorithm that achieves an $\alpha_{GW} + \tilde{\Omega}(1/d)$-approximation for Max 2-LIN(2) in weighted instances where each vertex appears in at most $d$ constraints. This improves on the $\alpha_{GW}+\tilde{\Omega}(1/d^2)$-approximation guarantee of~\cite{hsieh2023approximating}. We also show a matching Unique Games hardness under randomized polynomial-time reductions, thereby proving that the gain on $\alpha_{GW}$ is optimal up to $\mathrm{polylog}(d)$ factors. Here, $\alpha_{GW} \approx 0.878$ is the approximation ratio of the Goemans-Williamson algorithm. Our results resolve an open question due to Raghavendra and refute his conjectured optimal approximation ratio of $\alpha_{GW} + \tilde{\Omega}(1/\sqrt{d})$.


Our local improvement scheme can be thought of as using an adaptive (based on the neighbors) shift applied to the Gaussian before using the GW sign rounding. This allows us to implement a \emph{soft} but \emph{simultaneous} flipping decision at each vertex as opposed to hard but sequential flipping decisions in prior works. Our hardness results are based on a new \emph{sparsification} lemma for graphs produced by the reduction from Unique Games~\cite{KKMO07} that preserves the optimal values up to a factor of $1 \pm \epsilon$ while subsampling to an average degree of only $\tilde{O}(1/\epsilon)$.
This is the first example of a CSP in which the optimal advantage on degree-$d$ instances is asymptotically smaller than $O(1/\sqrt{d})$~\cite{Trevisan01}.
A key idea in the proof of our sparsification theorem establishes new bounds on the number of {\em nearly optimal} cuts in the output of the KKMO reduction. This step, in turn, relies on a new, robust \emph{inverse} majority-is-stablest theorem.  



\end{abstract}
\newpage
\tableofcontents
\newpage

\section{Introduction} 
What is an optimal algorithm for Max-Cut?

In the 1990s, Goemans and Williamson (GW)~\cite{MR1412228-Goemans95} introduced their breakthrough approximation algorithm for Max-Cut. Their algorithm rounded the natural semidefinite programming (SDP) relaxation for the problem using the now ubiquitous \emph{hyperplane rounding} to obtain an approximation ratio of $\alpha_{GW} \approx 0.878$. Building on the influential work of Håstad~\cite{Hastad99}, Khot, Kindler, Mossel and O'Donnell~\cite{KKMO07} made the surprising discovery that assuming the Unique Games Conjecture (UGC), the GW algorithm is in fact optimal! Building on this, in a breakthrough work, Raghavendra~\cite{Rag08} proved that this entire story applies word for word to \emph{every} constraint satisfaction problem: a natural SDP relaxation followed by a natural analog of hyperplane rounding is in fact optimal assuming the UGC. These remarkable results made it possible to find an \emph{optimal} algorithm for Max-Cut -- a question that would have been difficult to even pose just a decade before.

\paragraph{Is the GW algorithm really optimal?} The above sequence of works gives a complete and strikingly simple characterization of optimal approximation algorithms for constraint satisfaction problems. But this theory is \emph{coarse}. Once we ask only a slightly more fine-grained question of the optimal approximation guarantee as a function of the \emph{degree} or density of the input graph, the plot thickens. 

Feige, Karpinski and Langberg~\cite{FKL02} showed that the Goemans-Williamson algorithm is \emph{not} optimal. Indeed, a \emph{local search} procedure, when applied to the rounded solution,  improves the approximation ratio~\footnote{The lower-order gain on $\alpha_{GW}$ is an easy consequence of the fact that in the settings where GW is tight, the variance of rounded GW solution is $\sim nd$.} from $\alpha_{GW} + \frac{1}{\sqrt{nd}}$ to $\alpha_{GW} + \Omega(1/d^4)$ where $d$ is the maximum degree of any vertex in the graph. The FKL algorithm relied on a ``vertex-level'' improvement procedure. They showed that because of the variance in the GW rounded solution, there are sufficiently many vertices such that a majority of their neighbors are on the same side of the cut. By switching such vertices (and discarding all of their neighbors from consideration), one can strictly improve the rounded value. Their analysis was improved to obtain an $(\alpha_{GW} + \Omega(1/d^3))$-approximation by Floren~\cite{Floren16}. These results showed the power of local search in improving the GW algorithm. But the actual guarantees remained far from the state-of-the-art hardness-of-approximation result of~\cite{Trevisan01} that only rules out an $(\alpha_{GW} + \omega(1/\sqrt{d}))$-approximation~\footnote{Strictly speaking, Trevisan~\cite{Trevisan01} gave a black-box reduction which can transfer an $\alpha$-factor hardness for a Max-CSP problem to an $(\alpha + O(1/\sqrt{d}))$-factor hardness for the Max-CSP in the bounded-degree setting. This combined with \cite{KKMO07}'s $\alpha_{GW}$-factor hardness implies the $(\alpha_{GW} + O(1/\sqrt{d}))$-factor hardness for Max-Cut in graphs with maximum degree $d$.}. 

\paragraph{The power of local algorithms and Raghavendra's conjecture} Raghavendra (see Problem 10 in ~\cite{boazbarakSoSnotes}) observed that understanding the optimal approximation ratio for Max-Cut as a function of degree $d$ is not only a basic question, but also a testbed for the power of \emph{local search} in approximation algorithm design. Later, in a sequence of lectures~\cite{prasad-talk} at the Simons Institute, he reiterated the goal of finding optimal approximation algorithms for Max-Cut on degree $d$ graphs as an outstanding research direction and conjectured that the optimal approximation ratio for Max-Cut should match the best known hardness result of $\alpha_{GW} + O(1/\sqrt{d})$. This conjecture was also proposed (as Problem 1.13) in~\cite{huang2025weak} along with a certain Markov chain run from the GW rounded solution as a potential algorithm.

Raghavendra's conjecture was motivated in part by the work of~\cite{BMORRSTVWW15}, which shows local algorithms that improve the guarantee of random assignment by an additive $\Omega(1/\sqrt{d})$. As an immediate corollary, we obtain that for every ``approximation-resistant'' CSP such as 3-SAT, one can obtain an $\sim 1/\sqrt{d}$ additive gain for degree-$d$ instances over the optimal inapproximability for general inputs. His path to achieving a similar result for non-approximation-resistant CSPs such as Max-Cut was based on remarkable developments in \emph{local algorithms} inspired by statistical physics~\cite{Montanari19, Montanari-Sellke21}. In particular, these results suggested belief propagation and its variants, which are local algorithms that combine information on larger (but still constant-hop) neighborhoods, as natural avenues of beating the ``vertex-level'' FKL rounding algorithm. 

In response to the conjecture, Hsieh and Kothari~\cite{hsieh2023approximating} found a new (but still ``vertex-level'') local search procedure along with a more ``global'' analysis to obtain an improved approximation ratio of $\alpha_{GW} + \tilde{\Omega}(1/d^2)$. Unlike the prior works, their ``global'' analysis extends even to weighted graphs. Their method naturally suggests an improvement to $\alpha_{GW} + \tilde{\Omega}(1/d^{1.5})$, but this appeared difficult to prove. They reiterated the outstanding question of how to close the gap between the hardness results and the approximation algorithms.     

In summary, we have a wide quantitative gap between the best approximation ratio and the best hardness results for Max-Cut as a function of the degree $d$. Even more importantly, it is not clear if SDP rounding combined with local search could close that gap. Even the right incarnation of local search capable of achieving this remains unclear. 

\paragraph{Our contribution} In this work, we give a surprising resolution of Raghavendra's conjecture. We show that there is a simple, ``vertex-level'' local improvement procedure that achieves an $\alpha_{GW} + \tilde{\Omega}(1/d)$ approximation for Max-Cut. We then show that this is in fact optimal by proving a matching (up to polylog factors) Unique Games hardness! This is in sharp contrast to the case of approximation-resistant CSPs that admit local algorithms that improve the unbounded-degree approximation ratio by an additive $\sim 1/\sqrt{d}$. As a result, we refute Raghavendra's conjecture. 

Our work, perhaps surprisingly, suggests that SDP rounding combined with a ``vertex-level'' local search procedure may be an optimal approximation algorithm for bounded-degree instances of all CSPs. If true, this would be a beautiful Raghavendra-like~\cite{Rag08} characterization of optimal approximation algorithms for bounded-degree CSPs.

\paragraph{Why bounded degree?} Parameterizing CSPs by their density/degree is natural and already extensively studied in related contexts such as random and semirandom constraint satisfaction where understanding the optimal trade-off between density, running time and the quality of approximation has been a long-running research thread~\cite{CojaOghlan09RandomCSP,GoerdtKrivelevich01RandomKSAT,Feige02AverageCaseApprox,AllenODonnellWitmer15RandomCSP,KothariMoriODonnellWitmer17SOSRandomCSP,DeshpandeMontanariODonnellSchrammSen18NAE3SAT,AbascalGuruswamiKothari20SemiRandomCSP,GuruswamiKothariManohar21SmoothedCSP,HsiehKothariMohanty22MooreBound}. 
In the worst-case setting, bounded-degree instances 
have been studied for various CSPs~\cite{Hastad2000BoundedOccurrenceCSP,
BermanKarpinski1999TighterInapproximability,
AlimontiKann2000CubicGraphs,  
BermanKarpinski2002BoundedDegreeMinCSP,
HalperinLivnatZwick2004MaxCutCubicGraphs, 
BazganTuza2008MaxCutSubcubic,
GuruswamiZhou2012BoundedOccurrenceOrderingCSPs, Makarychev2013LocalSearchOrderingCSP,
LeeManurangsi2024BoundedDegreeMax2CSP, minzer2024near}, 
which led to the current best inapproximability for other classical combinatorial optimization problems like Traveling Salesperson Problem~\cite{chlebik2022weighted, nagda2025reinforced} and $k$-Set Packing~\cite{lee2025asymptotically, minzer2026near}.
Furthermore, they have been natural candidates for claims of quantum speedup (including quantum adiabatic algorithms~\cite{FarhiGoldstoneGutmannSipser00Adiabatic,FarhiEtAl01AdiabaticRandom} and the recent surge of interest in Decoded Quantum Interferometry for CSP-like problems such as Max-XORSAT, Max-LINSAT, and Optimal Polynomial Intersection~\cite{JordanEtAl24DQI,AnschuetzGamarnikLu25DQIStructure,KramerSchubertEisert26MaxLINSAT,BuGuLi26WeightedDQI,KramerSchubertEisert26BoundedDegreeMaxLINSAT}) and in fact motivated some research in finding better classical algorithms for such and related instances~\cite{JordanEtAl24DQI,BuGuLi26WeightedDQI,AnschuetzGamarnikLu25DQIStructure,KramerSchubertEisert26BoundedDegreeMaxLINSAT,SchmidhuberEtAl25HDQI}.

We next discuss our results and the key new technical ingredients in more detail. 

\subsection{Our Results}
Our algorithmic results all apply to weighted Max 2-LIN(2) instances, while our hardness result applies to Max-Cut on unweighted bounded-degree simple graphs. Thus, both our results hold in their respective most general natural settings. 

Let us introduce some background to prepare for our results. 
\paragraph{Background:} A Max 2-LIN(2) instance is described by a graph $G=([n],E)$, non-negative weights $w:E \rightarrow \R_+$ and a collection of ``right hand sides'' $b: E \rightarrow \{-1,1\}$. We write $W = \sum_{e \in E} w_e$ for the total weight of the constraints. Given an instance $I = (G,w,b)$ of Max 2-LIN(2), the \emph{value} of an assignment $x$ equals

\begin{equation}
\val(x) = \frac{1}{W} \sum_{e = \{i,j\} \in E}  w_e (\frac{1}{2} + \frac{1}{2} b_e x_i x_j).
\end{equation}
We write $\opt(I) := \max_{x \in \{-1,1\}^n} \val(x)$ to denote the maximum value of any assignment. Observe that when $b_e = -1$ for every $e \in E$, the above is simply the (weighted) Max-Cut problem.

The classical Goemans-Williamson algorithm~\cite{MR1412228-Goemans95} yields an $\alpha_{GW} \approx 0.878$ approximation for Max 2-LIN(2), where $\alpha_{GW}$ is defined precisely below:
\begin{equation*}
    \alpha_{GW} \coloneqq \min_{\rho \in [-1, 1)}\frac{s(\rho)}{c(\rho)}, \qquad c(\rho) = \frac{1- \rho}{2}\,\, \text{and} \,\, s(\rho) = \frac{\arccos \rho}{\pi}.
\end{equation*}
Let $\rho^* \approx -0.689$ denote the unique value of $\rho$ which achieves the minimum in the above ratio, and define $c^* := c(\rho^\ast), s^* := s(\rho^*)$. 

\paragraph{Approximation Algorithm: } Let us first present our algorithmic result. 

\begin{theorem}[Main Result, Algorithm]
\label{thm:main-alg-intro}
There is a polynomial-time algorithm that takes as input a weighted Max 2-LIN(2) instance $I = (G,w,b)$ and outputs an assignment with a value at least $(\alpha_{GW} + \Omega(\frac{1}{d\log d})) \opt(I)$ where $d$ is the maximum degree of any vertex in $G$. 
\end{theorem}

\paragraph{Key Technical Idea:} Let us explain the main idea in our algorithm briefly (we give a more detailed overview in Section~\ref{sec:techniques}). 

As in prior work~\cite{FKL02,Floren16,hsieh2023approximating}, our algorithm starts with Goemans-Williamson rounding of the SDP vectors $v_i$ by choosing a random Gaussian $g \sim \mathcal{N}(0,I)$ and computing   $x_i^0=\mathsf{sign}(z_i) = \operatorname{sign}(\langle g, v_i \rangle)$.
The local algorithm in~\cite{hsieh2023approximating} chooses a fixed narrow slab \(S=\{i:|z_i|<\varepsilon\}\) and flips some $x_i^0$ if it increases the Max-Cut \emph{after} making all flips to all the neighbors of $i$. Intuitively, the Gaussian $g$ must have a small projection in the direction $v_i$ for any $i \in S$, and thus the signs of $z_j$ for $j \sim i$ are approximately independent of $z_i$. We may thus expect a $(1/2 \pm 1/\sqrt{d})$-fraction of them to be on either side of the cut and thus hope to gain a $1/\sqrt{d}$ for each $i \in S$. 

A priori, this is a \emph{sequential} and highly dependent procedure. The gain on a vertex is a function of all \emph{previous} flips on its neighbors. To make the analysis tractable, Hsieh and Kothari's algorithm performs a pessimistic analysis: they assume that whenever a vertex is flipped, all vertices in $S$ that are its neighbors incur a \emph{loss}. Thus, they gain only from those vertices for which the total mass of S in the neighborhood is small relative to the gain at that vertex. This forces them to a small slab width. Indeed, such a strategy has a hard barrier at $\sim 1/d^{1.5}$, and the tightest known analysis~\cite{hsieh2023approximating} only yields a $\sim 1/d^2$ gain. 

Our algorithm beats even this optimistic heuristic by (1) dropping the requirement of a sequential analysis and (2) dropping the pessimistic accounting idea of prior works. Instead, we think of the local improvement as inducing a small ``soft'' bias (as opposed to the hard slab restrictions in prior works) in the rounding threshold for Goemans-Williamson rounding itself and rounding all the vertices \emph{simultaneously} in a single round. 
More precisely, we compute a ``local preference''
$f_i$ and round:
\[
    x_i^\lambda=\operatorname{sign}(z_i+\lambda f_i).
\]
At a high level, $f_i \in [-1,1]$ is a smoothed, weighted average of scores from the neighbors of $i$ (restricted to those on edges with the SDP contribution $\approx (1+\rho_*)/2$).
From the perspective of the algorithm in~\cite {hsieh2023approximating}, our scheme above can be thought of as ``adaptively'' determining the window for flipping $i$ based on the strength of the local preference.

Here is the intuition for why this crosses the $d^{-3/2}$ barrier. A typical local preference has magnitude $|f_i|\asymp d^{-1/2}$. The new rounding flips $x_i$ with probability proportional to $\lambda |f_i|$, and such a change has expected gain proportional to $|f_i|$. Thus, in total, we obtain an expected gain of order $\lambda f_i^2 \asymp \lambda/d$ at \emph{every} vertex! Importantly, no
restriction on an independent set is needed. Intuitively, this is because an edge ``costs'' us only if both endpoints need to
change and is therefore \emph{quadratic} in the small threshold shifts, whereas the useful bump is \emph{linear}. 

\paragraph{Hardness of Approximation: } We prove an $\alpha_{GW} + \tilde{O}(1/d)$-hardness for bounded-degree Max-Cut instances. This improves on the previous state-of-the-art hardness result of Trevisan~\cite{Trevisan01} (combined with~\cite{KKMO07}) that shows that obtaining an $\alpha_{GW} + \omega(1/\sqrt{d})$-approximation for Max-Cut is UG-hard.  

Our main result is stated in the theorem below.

\begin{theorem}    \label{thm:ugc-hard}
	Assuming UGC and $\NP \not\subseteq \BPP$, the following holds\footnote{Recently OpenAI have uploaded a proof of the Unique Games Conjecture~\cite{OpenAI-ugc}. If this proof is correct, then this hardness result holds by just assuming $\NP\not\subseteq \BPP$.} for every large enough positive integer $d$. Given a graph $H = (V, E)$ with maximum degree at most $d$, there is no polynomial-time algorithm that distinguishes between the following cases:
	\begin{itemize}
		\item {\bf YES Case}. $\MaxCut(H) \geq c^* -  O\left(\frac{1}{d}\right)$,
		\item {\bf NO Case}. $\MaxCut(H) \leq s^* + O\left(\frac{\log^2 d}{d}\right)$. 
	\end{itemize}
    Here $\MaxCut(H) = \frac{\opt(H)}{|E|}$ is the normalized Max-Cut value of $H$. 
\end{theorem}

As an immediate corollary of the above theorem, we obtain the following result, which establishes improved hardness for Max-Cut in bounded-degree instances.

\begin{corollary}   \label{cor:ugc-hard}
	Assuming UGC, there is no polynomial-time algorithm to approximate Max-Cut beyond $\left(\alpha_{GW} + \Omega\left(\frac{\log^2 d}{d}\right)\right)$-factor in graphs with maximum degree $d$, unless $\NP \subseteq \BPP$.
\end{corollary}

We also give an integrality gap with similar parameters for the Goemans-Williamson SDP relaxation equipped with the $\ell^2_2$-triangle inequalities; this is stated in the theorem below.

\begin{theorem}				\label{thm:max-cut-gap-1}
	For every sufficiently large integer $d$, there exists a graph $G = (V,E)$ with maximum degree at most $d$ such that the following holds.
	\begin{itemize}
		\item {\bf Completeness}. The optimal value of SDP relaxation for $G$ is at least $c^* - O(1/d)$.
		\item {\bf Soundness}. The optimal Max-Cut value of $G$ is at most $s^* + O\left(\frac{\log^2(d)}{d}\right)$.
	\end{itemize}
    Consequently, the integrality gap of the SDP relaxation is at most $\alpha_{GW} + O\left(\frac{\log^2(d)}{d}\right)$.
\end{theorem}

\paragraph{Key Technical Ideas.} Loosely speaking, in the context of the Max-Cut problem, Trevisan's~\cite{Trevisan01} approach can be thought of as a ``cut-sparsification'' procedure for the hard instances of Max-Cut from \cite{KKMO07}. Specifically, given a Max-Cut instance $G = (V,E)$, Trevisan showed that if one sub-samples an $O(d)$-degree subgraph of $G$, then with high probability, all the cuts in the sub-sampled graph are preserved up to an $O(1/\sqrt{d})$ additive error, which in turn results in the $\alpha_{GW} + O(1/\sqrt{d})$-hardness for Max-Cut in the bounded-degree setting.  Furthermore, it is well-known that $O(1/\sqrt{d})$-error represents a fundamental limit for cut-sparsification up to average-degree $d$~\cite{andoni2016sketching}. 

Then how can we go beyond the aforementioned $O(1/\sqrt{d})$-error limit for cut-sparsification? To that end, our first key observation is that, instead of trying to preserve the values of all the cuts in the sub-sampled graph, it suffices to just preserve the value of the Max-Cut. This allows us to implement a {\em multi-scale version} of the union bound, wherein we allow cuts that are far from being optimal to deviate significantly, and tightly control the deviations for cuts that are near-optimal.

However, in order to make the above approach work, we would also need to ensure that the number of near-optimal cuts is much smaller than the total number of cuts. And in general, arbitrary graphs may not satisfy such a property at all\footnote{For instance, consider the complete graph on $n$-vertices, in which $(1 - e^{-\Omega(n)})$-fraction of all the cuts are nearly optimal.}! Fortunately, it turns out that the main combinatorial object underlying the Max-Cut hardness, i.e., {\em the noisy hypercube}, does indeed satisfy such a property.
To show this, we develop an asymptotically optimal robust inverse of the {\em Majority-Is-Stablest} theorem~\cite{MOO10} where we show that any (non-influential) near-optimal cut is ``close'' to a linear threshold function (LTF). 

The above structural result allows us to derive an asymptotically optimal bound on the number of near-optimal cuts in the noisy cube, which in turn allows us to effectively implement the multi-scale union bound via a sharper Bernstein inequality-based tail bound. Our tail bound is not obtained by comparing a sparsified cut to its expectation; rather, we compare it to its closest LTF and only look at the edges in their symmetric difference. This, which can be thought of as a simple \emph{chaining} argument, gives us much finer control on the deviation. At a high level, at scale $\eta$, we are able to show that (a) the number of $\eta$-optimal cuts is at most $2^{\tilde{O}(\sqrt{\eta})n}$ and (b) any $\eta$-optimal cut exceeds the max-cut value by $t$ in the sub-sampled graph with probability at most $\exp(-(\eta + t)^{3/2}nd)$, and hence the union-bound term at scale $\eta$ ends up being of the order $e^{-\Omega_d(n)}$ for $t = \tilde{O}(1/d)$, thus allowing us to take a union bound over all $\eta$.

\paragraph{AI use disclosure.} The analysis of our local improvement algorithm was developed in a long interaction with ChatGPT 5.5 Pro. The proof of Theorem~\ref{thm:robust-MIST} (the robust version of the Majority-is-Stablest theorem) is
AI-assisted. The authors devised the overall proof strategy, including all key
steps outlined in Section~\ref{sec:MiST}, and used ChatGPT 5.5 Pro to
assist with writing out the technical details. The authors have fully checked
and revised all AI-generated text and assume responsibility for all content,
including the correctness of the proof.

\section{Technical Overview}
\label{sec:techniques}

\subsection{Our Approximation Algorithm}
As in the previous works~\cite{FKL02,Floren16,hsieh2023approximating}, our starting point is the
Goemans--Williamson SDP strengthened with the triangle inequalities.  Let
$v_1,\ldots,v_n$ be an optimal SDP solution, draw $g\sim\mathcal N(0,I_n)$,
and write
\[
        z_i=\langle g,v_i\rangle,
        \qquad
        x_i^0=\operatorname{sign}(z_i).
\]
Thus $x^0$ is the usual Goemans--Williamson rounding.  For an edge
$\{i,j\}$, let $\rho_{ij}=b_{ij}\langle v_i,v_j\rangle$ be its signed SDP
correlation.  The edge-by-edge analysis of Goemans--Williamson is tight only
when $\rho_{ij}$ lies in a small constant-size interval around
$-\rho^\ast\approx 0.689$.  We call these edges \emph{tight} and denote their
set by $E_T$.  If a noticeable fraction of the SDP value comes from edges
outside $E_T$, then ordinary Goemans--Williamson rounding already beats
$\alpha_{\mathrm{GW}}$ by the amount we need.  We may therefore focus on the
case in which the tight edges carry almost all of the SDP value.

\paragraph{The Key Idea and Bottleneck in Prior Works.}

The main observation in all the prior works above is that the rounding decision at $i$ is ``noisy'' when $z_i$ is close to zero. Intuitively, this is because $z_i$ being small in magnitude means that $g$ (the Gaussian chosen in GW rounding) is nearly orthogonal to $v_i$. In this case, the rounded value of neighbors of $i$ is nearly independent of $x_i^0$. Thus, we expect the fraction of neighbors on the same side as $i$ to have a high variance with mean around $1/2$. Thus, flipping $x_i$ becomes profitable with a constant probability. 

Hsieh and Kothari~\cite{hsieh2023approximating} implement this idea by
choosing a fixed narrow slab $\{i:|z_i|\leq\varepsilon\}$ and making a full flip only when its local gain remains positive. 

This process is sequential and dependent. The ``profitable'' flips need not remain profitable when some of their neighboring vertices are flipped previously. One way to avoid this interaction is to retain only an independent set of candidates, which immediately loses a factor \(1/d\) and creates the \(d^{-3/2}\) barrier. Hsieh and Kothari~\cite{hsieh2023approximating} obtain similar protection: they
select a fixed narrow slab \(S=\{i:|z_i|\leq\varepsilon\}\) and flip \(i\) only when its gain against vertices outside \(S\) exceeds the total weight of its edges into \(S\). In effect, every edge between two candidates is pessimistically charged as a potential loss. This makes the simultaneous flips safe, but requires the
candidate slab to be very sparse.




\paragraph{Our Key Idea:} Our main idea is to avoid making each local flip ``safe''. Instead, we compute a local
bias \(f_i\) and modify the GW rounding itself:
\[
        x_i^\lambda=\operatorname{sign}(z_i+\lambda f_i).
\]
We can think of this as a ``soft'' version of the flipping strategy in previous works, but crucially, in our version, all vertices are rounded \emph{simultaneously}. Why might we gain? Heuristically, the positive effect of flipping one
endpoint of an edge is linear in the small threshold shift, whereas 
a negative effect requires both endpoints to change and is therefore quadratic. We may ``win'' if the linear gain dominates the quadratic loss. Let us explain how the local bias \(f_i\) is constructed and how the gain accumulates to $\sim 1/d$ next.



\paragraph{Computing the Scores.}
Fix a tight edge $\{i,j\}$.  From the point of view of $i$, decompose the Gaussian computed in the GW rounding at $j$, $z_j = \langle g, v_j \rangle$ as:
\[
 b_{ij}z_j
   =\rho_{ij}z_i+\sqrt{1-\rho_{ij}^2}\,h_{ij},
 \qquad
 h_{ij}:=
 \left\langle g,
 \frac{b_{ij}v_j-\rho_{ij}v_i}{\sqrt{1-\rho_{ij}^2}}
 \right\rangle .
\]
Now, $h_{ij}$ is a standard Gaussian independent of $z_i$.  If
$z_i=0$, then
$\operatorname{sign}(h_{ij})=b_{ij}\operatorname{sign}(z_j)$, which is
exactly the label that $j$ would like $i$ to take in order to satisfy their
constraint.  Thus $\operatorname{sign}(h_{ij})$ is the natural vote of $j$
for the label of $i$.  Removing the part of $z_j$ predicted by $z_i$ is important: it makes the local ``recommendation'' independent of the quantity whose rounding threshold we will move.

Without loss of generality, assume that the total weight $W$ is $1$, and let $W_i=\sum_{j:\{i,j\}\in E_T}w_{ij}$.  The most natural score would be the weighted vote
\[
        \mathsf{vote}_i
        =\frac{1}{W_i}
          \sum_{j:\{i,j\}\in E_T}w_{ij}\operatorname{sign}(h_{ij}).
\]
We believe that this score may also work, but we do not currently know how to analyze it. Instead, for the sake of analysis, we will use the following score function: 
\[
 q_\eta(t)
   :=2\Phi\!\left(\frac{\eta t}{\sqrt{1-\eta^2}}\right)-1
   =\E_Z\!\left[
       \operatorname{sign}\!\left(\eta t+\sqrt{1-\eta^2}Z\right)
     \right],
 \qquad Z\sim\mathcal N(0,1),
\]
and set
\[
        f_i=\frac{1}{W_i}
          \sum_{j:\{i,j\}\in E_T}w_{ij}q_\eta(h_{ij}).
\]

We think of $q_\eta$ as adding a small amount of fresh Gaussian noise before taking the sign and then averaging out that noise.  It leaves ``decisive'' votes close to $\pm1$ while smoothly discounting ``ambiguous'' votes.  We take
$K=(1-\eta^2)^{-1}=\Theta(\log d)$, so the amount of smoothing is small.

\paragraph{Using Scores to Modify GW Rounding.}
For a small step size $\lambda$, our final rounding is
\[
                 x_i^\lambda
                 =\operatorname{sign}(z_i+\lambda f_i).
\]

This can be viewed as a local improvement of $x^0$ but it is helpful to think of it as a new ``one-shot'' rounding rule modifying GW rounding itself.  Notice that $i$ changes sides exactly when
\[
        x_i^0 f_i<0
        \qquad\text{and}\qquad
        |z_i|<\lambda|f_i|.
\]
Unlike the fixed slab in the earlier algorithms, the candidate window now adapts to the strength of the local preference: a ``strong'' vote is allowed to move a vertex farther from the hyperplane, while a ``weak'' vote acts only very close to it.  All vertices are updated at once, and there is no sequential local search process to analyze or an independent set to restrict to.

Let us give a heuristic analysis to explain why this may work. Suppose \(|f_i|\approx d^{-1/2}\) --- this turns out to be the typical scale. Thus, using the Gaussian PDF, $i$ is flipped with probability on the order of \(\lambda|f_i|\). Conditional on
flipping, its local gain is of order \(W_i|f_i|\). Its expected contribution is
therefore
\[
        \lambda W_i f_i^2
        \approx \frac{\lambda W_i}{d}.
\]
These contribution add up over \emph{all} vertices (without any further \(1/d\) loss as in prior works). Let us now understand why the gain turns out to be linear while the loss on an edge is quadratic, allowing the gain to dominate.  


\paragraph{The Analysis.}
Let
\[
        a_i:=\E[f_i\cdot\mathsf{vote}_i],
        \qquad
        S:=\sum_i W_i a_i.
\]
The quantity \(a_i\) measures the correlation between the smoothed score used by the algorithm and the hard local neighborhood vote that creates the gain. To see why, put $\delta_i=x_i^\lambda-x_i^0$ and expand the change on an edge:
\[
 x_i^\lambda x_j^\lambda-x_i^0x_j^0
   =\delta_i x_j^0+x_i^0\delta_j+\delta_i\delta_j.
\]
The first two terms correspond to changing one endpoint of edge $\{i,j\}$, and the last term is
the interaction between two changes.  Because $f_i$ is independent of $z_i$, we can condition on the local vote and integrate over the short interval of values of $z_i$ that cross the shifted threshold.  Summed over all directed edges, the leading one-endpoint contribution is exactly $\lambda\phi(0)S$, where $\phi$ is the density of the  Gaussian distribution. This formalizes the heuristic
that the useful gain is linear in the threshold shift.

\paragraph{The Same Quantity Controls the Interaction Errors.}
The technical heart of the analysis is that the same quantity \(S\) that controls the gain also controls the losses from simultaneous flips across the endpoints of edges. 

First, observe that $\E[f_i^2]\leq a_i$,
so a large threshold shift must be accompanied by a correspondingly large useful gain. More importantly, even though the updates at adjacent vertices are dependent, we prove
\[
 \Pr\big[|z_i|\leq\lambda|f_i|,
          |z_j|\leq\lambda|f_j|\big]
      \leq O\!\left(\lambda^2\sqrt K\,(a_i+a_j)\right).
\]
After summing over edges, the total error from changing both endpoints is at most \(O(\lambda^2\sqrt K\,S)\). The smaller errors in approximating the one-endpoint gain satisfy the same bound. Consequently,
\[
 \E[\operatorname{val}_{E_T}(x^\lambda)]
 -\E[\operatorname{val}_{E_T}(x^0)]
 \geq \lambda\phi(0)S-O(\lambda^2\sqrt K\,S).
\]
Taking \(\lambda=c/\sqrt K\) for a sufficiently small constant \(c>0\) makes the linear gain dominate the quadratic errors. This is how our global analysis avoids the loss inherent in the previous pessimistic accounting.


\paragraph{Lower Bounding the Gain.}
It remains to show that $S$ is large.  The triangle inequalities imply that two votes at the same vertex cannot be almost perfectly opposed: for some
universal constant \(\beta<1\),
\[
        \E[h_{ij}h_{ik}]\geq-\beta.
\]
Put \(p_j=w_{ij}/W_i\) and
\(\gamma_{jk}=\E[h_{ij}h_{ik}]\). By Sheppard's formula,
\(a_i = \E [ f_i \mathsf{vote}_i] \) is a weighted sum of \(\arcsin(\eta\gamma_{jk})\). Individual terms in this sum may be negative, so bounding them separately is not useful. We will show that the aggregate gain is actually positive. 

The matrix \((\gamma_{jk})_{j,k}\) is a covariance matrix and hence is positive
semidefinite. By the Schur product theorem, its entrywise powers are also
positive semidefinite. Therefore, for every positive integer \(m\), $\sum_{j,k}p_jp_k\gamma_{jk}^m\geq0$.

Now expand
\[
        \arcsin(t)=\sum_{\ell\geq0}c_\ell t^{2\ell+1},
        \qquad c_\ell>0.
\]

Importantly, all Taylor coefficients in the above expansion are non-negative so we can use any subset to obtain a useful lower bound.  Let \(L=\lfloor K/10\rfloor\), and let's keep only the terms with \(\ell\in[L,2L]\). For each such \(\ell\), write
\[
        Q_\ell:=\sum_{j,k}p_jp_k\gamma_{jk}^{2\ell+1}.
\]
The diagonal contribution to \(Q_\ell\) is $\sum_jp_j^2\geq\frac1d$.

Since every negative off-diagonal entry satisfies
\(\gamma_{jk}\geq-\beta\), the total negative off-diagonal contribution is at least $-\beta^{2\ell+1}\sum_{j,k}p_jp_k \geq-\beta^{2L+1}$.
Choosing the constant in \(K=C_0\log d\) sufficiently large ensures that \(\beta^{2L+1}\leq1/(2d)\), and hence \(Q_\ell\geq1/(2d)\) throughout this
block. Moreover,
\[
        \eta^{2\ell+1}\geq\eta^{4L+1}=\Omega(1),
        \qquad
        \sum_{\ell=L}^{2L}c_\ell=\Theta(K^{-1/2}).
\]
Keeping only these terms in the \(\arcsin\) expansion therefore gives $a_i=\Omega\!\left(\frac{1}{d\sqrt K}\right)$.

\subsection{Overview of Unique Games hardness}

Next, we give a brief overview of the existing challenges and our techniques for proving Theorem \ref{thm:ugc-hard}. In particular, we will first describe how the techniques of \cite{KKMO07} and \cite{Trevisan01} can be combined to prove the existing $(\alpha_{GW} + O(1/\sqrt{d}))$-hardness for Max-Cut. We will then discuss the existing bottlenecks towards improving the analysis and our techniques, and then conclude with an informal overview of the proof of Theorem \ref{thm:ugc-hard}.

\subsubsection{Review of $\alpha_{GW}$-hardness in the general setting.} 

To simplify our exposition, we will just focus on the underlying dictatorship test gadget, namely, the {\em $\rho^*$-correlated hypercube}, which is the key combinatorial object used for proving the optimal hardness for Max-Cut. Formally, the $R$-dimensional $\rho^*$-correlated hypercube, denoted by $\cH^R_{\rho^*}$, is the graph on $\{0,1\}^R$ whose edges are distributed as pairs of strings in $\{0,1\}^R$ which are coordinate-wise $\rho^*$-correlated\footnote {That is, it is a weighted graph with total weight $1$, and the weight of $(u, v)$ with $u \neq v \in \{0, 1\}^R$ is $2 (\frac{1-\rho^*}{4})^{d_H(u,v)} (\frac{1+\rho^*}{4})^{R - d_H(u,v)}$, where $d_H(\cdot,\cdot)$ denotes the Hamming distance.}. 
The above gadget satisfies a certain set of ''structured vs random'' properties which, using standard techniques, can be used to show the $\alpha_{GW} = s^*/c^*$-factor hardness of Max-Cut. 

$\triangleright$ {\bf Completeness}. Consider the dictator set $S := \{x \in \{0,1\}^R : x(1) = 1\}$. Since the edges of $\cH^R_{\rho^*}$ are $\rho^*$-correlated strings, it is straightforward to see that the weight of edges cut by $(S,S^c)$ is exactly equal to the probability that a pair of $\rho^*$-correlated bits are unequal, which is equal to $c^* = (1 - \rho^*)/2$.

$\triangleright$ {\bf Soundness}. For the soundness direction, one would like to bound the maximum weight of edges that are cut by any ``(quasi-)random set'', i.e., sets whose indicator functions have no influential coordinates. This question was answered strongly in the form of the {\em Majority-Is-Stablest} theorem~\cite{MOO10}, which is stated informally below.

\begin{theorem}[Informal version of Theorem 4.4 in~\cite{MOO10}]
	Any quasi-random set $S \subseteq \{0,1\}^R$ cuts edges of total weight at most $(s^* + o(1))$ in $\cH^R_{\rho^*}$. Moreover, this maximum value is achieved (up to $o(1)$-error) by the Majority function\footnote{The majority function ${\sf MAJ}:\{0,1\}^R \to \{0,1\}$ is simply the function that outputs the majority of the input bits.} ${\sf MAJ}$.
\end{theorem}

The first part of the above theorem immediately implies that the maximum weight of edges cut by any quasi-random set is at most $(s^* + o(1))$. Furthermore, the above structured vs quasi-random properties of $\cH^R_{\rho^*}$ can be lifted to an $s^*/c^* = \alpha_{GW}$-factor hardness for Max-Cut~\cite{KKMO07,Rag08}. As a brief digression, we point out that while the additional fact that the {\sf MAJ} function (almost) achieves the optimal-cut value wasn't immediately necessary for establishing the soundness guarantee of the gadget, as we will see later, this property (and a certain robust inverse variant of it) will be crucially used to improve the hardness in the bounded-degree setting.

\subsubsection{Hardness in the Bounded-Degree Setting}

Now, in order to port the aforementioned hardness to the bounded-degree setting, one would like to construct a bounded-degree gadget with analogous completeness and soundness properties. To that end, note that $\cH^R_{\rho^*}$ by itself is not a bounded-degree graph; in fact, the unweighted  edge support of the graph is the entire complete graph with all the self-loops included. Consequently, the $\rho^*$-noisy cube gadget by itself doesn't yield any meaningful guarantees in this setting. To address this issue, Trevisan~\cite{Trevisan01} showed that it is possible to construct a ``bounded-degree approximation'' of the gadget by simply sub-sampling edges. 

Specifically, consider the graph $\wH^R_{\rho^*}$ which is constructed by independently sampling $nd$ edges (with $n = 2^R$) from $\cH^R_{\rho^*}$, where the edges are sampled according to the weight distribution of the edges in $\cH^R_{\rho^*}$. Clearly, this construction will preserve all the cut values in expectation, and with a little bit of work, one can also show that the cut values are approximately preserved with high probability. Specifically, for any fixed partition $(A,A^c)$, let $w_{A,A^c}$ and $\hat{w}_{A,A^c}$ denote the weight of edges cut by $A$ in $\cH^R_{\rho^*}$ and $\wH^R_{\rho^*}$ respectively. Then, using a standard application of the Chernoff bound, it is easy to see that for any $t > 0$:
\[
	\Pr[|\hat{w}_{A,A^c} - w_{A,A^c}| \geq t] \leq \exp\left(-\Omega(t^2 nd)\right),
\]
and hence by setting $t := \Theta(1/\sqrt{d})$ with a sufficiently large constant, we can union bound over all $2^n$ cuts. In other words, the sub-sampled graph $\wH^R_{\rho^*}$ preserves all the cuts (and consequently, the completeness and soundness guarantees) up to an additive error of $O(1/\sqrt{d})$.
This eventually translates to the $(\alpha_{GW} + O(1/\sqrt{d}))$-hardness for Max-Cut for graphs with maximum degree\footnote{To be precise, the sub-sampled graph as is only satisfies the slightly weaker guarantee that its average degree is $\approx 2d$. However, one can convert this into a graph with maximum degree $2d(1 + o(1))$ by deleting all the high degree vertices, at the cost of incurring some additional negligible errors.} at most $2d$.

\subsubsection{Key Observations for Improved Bounds}

How can we improve on the above bound? At a first glance, it seems that incurring an $O({d}^{-1/2})$-additive loss in the soundness guarantee of $\wH^R_{\rho^*}$ is perhaps unavoidable, since setting the deviation parameter to be equal to $t = O(d^{-1/2})$ is required to control the union bound over all $2^n$ cuts. However, we observe that in order to just preserve the soundness value, we only need to control the deviation of the maximum cut value from $s^*$, as opposed to the deviation of all the cuts from their expectation. And hence, trying to control the values of all cuts to the same order of deviation is quite inefficient.

For instance, consider a typical partition $(A,A^c)$ in $\cH^R_{\rho^*}$; such a partition cuts $\approx 1/2$-weight of edges in $\cH^R_{\rho^*}$, and hence, it will also cut $\approx 1/2$-fraction of edges in $\wH^R_{\rho^*}$ in expectation. As a result, even if the cut value of such a partition ends up deviating by a small constant, say $0.01$, from its expected value, it will not affect the overall soundness guarantee, since the target soundness value $s^* \approx 0.74$ is much larger than $1/2$. In fact, even for quasi-random partitions whose cut value is at most $(s^* - 0.02)$, a deviation of $0.01$ from their expected value will not affect the soundness guarantee. And for any fixed partition, a constant deviation in cut value actually holds with probability $1 - e^{-\Omega(nd)}$, which is enough to take a union bound over all cuts. These observations imply that, with high probability, every partition $(A,A^c)$ whose cut value is far from $s^*$ satisfies $\hat{w}_{A,A^c} \leq s^*$.

\begin{observation}				\label{obs:basic}
	Let $\cA_{\rm far}$ denote the set of quasi-random partitions $(A,A^c)$ whose cut value $w_{A,A^c}$ is at most $s^* - C / \sqrt{d}$, for some large enough constant $C > 0$. Then, with probability at least $1 - e^{-\Omega(n)}$, every partition $(A,A^c) \in \cA_{\rm far}$ satisfies $\hat{w}_{A,A^c} \leq s^*$.
\end{observation}

In other words, the same analysis as before guarantees that all quasi-random partitions $(A,A^c)$ whose cut value $w_{A,A^c}$ is $(C / \sqrt{d})$-far from the maximum value $s^*$ satisfy $\hat{w}_{A,A^c} \leq s^*$ with high probability. The above observation suggests that the actual bottleneck in controlling the error in the final soundness parameter comes from partitions $(A,A^c) \in \cA_{\rm near} := \cA_{\rm far}^c$ whose cut value is extremely close to the optimal soundness value i.e., $w_{A,A^c} \approx s^*$. And for such partitions, we indeed cannot allow their cut values to deviate from expectation by even small constant values. This, in turn, can be quite problematic if there are too many such partitions to union-bound over.  Therefore, in order to ensure that we can get a tighter control on the deviation from the soundness parameter $s^*$, we would need to answer the following question:
\begin{gather*}
	\text{\it For any small enough $\eta > 0$, how many quasi-random partitions of $\cH^R_{\rho^*}$} \\
	\text{\it have a cut value of at least $s^* - \eta$}? 
\end{gather*} 
Addressing the above question ends up being the technical core of our improved analysis of the sub-sampled gadget, and we elaborate on this in more detail below.

\subsubsection{Counting Near-Optimal Cuts via Robust Inverse of MIS }

To answer the above question, we will need a way to efficiently enumerate partitions whose cut value is nearly optimal, which in turn reduces the task to that of understanding the structure of near-optimal cuts. 
Recall that the Majority-Is-Stablest theorem implies that optimal soundness in $\cH^R_{\rho^*}$ is actually achieved by the majority function. Consequently, any quasi-random partition $(A,A^c)$ which is $\eta$-close to the majority function (or one of its symmetries) will have cut value at least $s^* - O(\eta)$. And now, one can ask if these are the only kind of partitions which satisfy $w_{A,A^c} \geq s^*-\eta$? This question was answered in the affirmative (in a weaker sense) by Mossel and Neeman~\cite{MosselNeeman15}, who established robust inverse theorems for the Majority-Is-Stablest theorem as well as Borell's Isoperimetric Inequality (the Gaussian space analog of the MIS theorem), albeit with non-optimal parameters. Subsequently, Eldan~\cite{Eldan15} gave an asymptotically optimal robust inverse theorem for Borell's inequality for positive correlations. In this work, we complete the picture by deriving an asymptotically optimal robust inverse for the Majority-Is-Stablest theorem, which we state informally below.

\begin{theorem}[Informal Version of Theorem \ref{thm:robust-MIST}] 		\label{thm:robust-mis-basic}
	The following holds for all small enough $\eta > 0$. Suppose $(A,A^c)$ is a quasi-random partition of $\cH^R_{\rho^*}$ that satisfies $w_{A,A^c} \geq s^* - \eta$. Then, there exists a linear threshold function $L:\{0,1\}^R \to \{0,1\}$ such that $\|\one_A - L\|_1 \leq O(\sqrt{\eta})$.
\end{theorem}

In other words, the above robust inverse theorem says that if the cut value of a quasi-random cut is $\eta$-close to the optimal value $s^*$, then the indicator of the partition must be $O(\sqrt{\eta})$-close to a linear threshold function. We point out that Theorem \ref{thm:robust-mis-basic} has optimal asymptotic dependence on the error parameter $\eta$ (see Remark \ref{rem:robust-opt}), and this asymptotic optimality ends up being crucial in establishing our $\tilde{O}(1/d)$-bound. We refer the reader to Section \ref{sec:robust-mist} for a high-level outline of the proof of the above theorem. Using Theorem~\ref{thm:robust-mis-basic}, we are able to give the following bound on the number of near-optimal quasi-random partitions in the $\rho^*$-correlated hypercube gadget.

\begin{lemma}			\label{lem:count-basic}
	For every small enough $\eta > 0$, the number of quasi-random cuts $(A,A^c)$ in $\cH^R_{\rho^*}$ satisfying $w_{A,A^c} \geq s^* - \eta$ is at most $2^{O(\sqrt{\eta}\log(1/\eta))n}$, where $n = 2^R$. 
\end{lemma}

The above lemma is proven by combining Theorem \ref{thm:robust-mis-basic} with a VC-dimension-based counting argument to bound the number of LTFs over $\{0,1\}^R$ and standard bounds for sizes of Hamming balls of radius $\delta$. The above lemma forms the basis of our improved bounds for the sub-sampled gadget, and using it, we give a simplified version of the improved analysis of $\wH^R_{\rho^*}$ below.

\subsubsection{A Proof Sketch of the Improved Soundness Analysis}

Here we give a simplified version of the proof which improves the sampling-error parameter from $O(d^{-1/2})$ to $O(d^{-5/8})$ in the soundness analysis, and then we will briefly discuss the additional changes required to obtain the $\tilde{O}(1/d)$-bound claimed in our main theorem. Recall that we used $\cA_{\rm far}$ and $\cA_{\rm near}$ to denote the set of quasi-random partitions $(A,A^c)$ whose cut values $w_{A,A^c}$ satisfy $w_{A,A^c} \leq s^* - C / \sqrt{d}$ and $w_{A,A^c} \in [s^* - C / \sqrt{d},s^* + o(1)]$ respectively. Let $t_{\rm new} = O(d^{-5/8} \log d)$. Then, for the soundness analysis, using a union bound, we break the analysis into near and far cases:
\begin{align*}
	&\Pr[\exists~\text{\it quasi-random} (A,A^c) : \hat{w}_{A,A^c} \geq s^* + t_{\rm new}]		\\
	&\leq \Pr[\exists (A,A^c) \in \cA_{\rm far} : \hat{w}_{A,A^c} \geq s^* + t_{\rm new}] 
	+ \Pr[\exists (A,A^c) \in \cA_{\rm near} : \hat{w}_{A,A^c} \geq s^* + t_{\rm new}]
\end{align*}
Next, for bounding the first term, using Observation \ref{obs:basic}, we can directly show that:
\begin{equation}				\label{eqn:o-1}
	\Pr[\exists (A,A^c) \in \cA_{\rm far} : \hat{w}_{A,A^c} \geq s^* + t_{\rm new}]
	\leq e^{-\Omega(n)}.
\end{equation}
For the second term, combining the bound from Lemma \ref{lem:count-basic} with Hoeffding's inequality, we observe that
\begin{align}
	\Pr[\exists (A,A^c) \in \cA_{\rm near} : \hat{w}_{A,A^c} \geq s^* + t_{\rm new}]
	&\leq |\cA_{\rm near}| \cdot \exp\left(-\Omega(t^2_{\rm new} nd)\right) 		\nonumber\\
	&\leq 2^{O(d^{-1/4}\log(d)) n} \cdot \exp\left(-\Omega(t^2_{\rm new} nd)\right) 	\nonumber\\
	&\leq e^{-\Omega(n)},				\label{eqn:o-2}
\end{align} 
where the last step follows from our choice of $t_{\rm new} = \tilde{O}(d^{-5/8})$. Combining the bounds from \eqref{eqn:o-1} and \eqref{eqn:o-2} shows that the soundness guarantee holds with $s^* + \tilde{O}(d^{-5/8})$. 

{\bf Improving the analysis to an $\tilde{O}(1/d)$-bound}. We conclude by briefly discussing how the above analysis can be improved further to obtain the $\tilde{O}(1/d)$-bound claimed in Theorem \ref{thm:ugc-hard}. 
\begin{itemize}
    \item The first key difference is that, instead of using a hard-thresholding scheme to classify the cuts into near and far sets, we use a smoother {\em bucketing-based} system which adjusts the deviations in the bound more gradually. More specifically, we consider a sequence of thresholds $(\eta^{(i)})_{i \geq 0}$ where $\eta^{(i)} = 2^i\cdot\eta^{(0)}$, and let $\cA^{(i)}$ be the set of cuts achieving at least $s^* - \eta^{(i)}$. Then, for any cut in $\cA^{(i)} \setminus \cA^{(i-1)}$, it is sufficient to bound its deviation by $\eta^{(i-1)}$, while having finer control of how far it is from the class of LTFs (namely, on the order of $\sqrt{\eta^{(i)}}$).
    \item Another key change is that, instead of using a naive union bound over Chernoff tail bounds, we look at the \emph{relative deviation} of a near-optimal cut from its closest LTF. This relative deviation is expected to be tiny by Theorem~\ref{thm:robust-mis-basic} (here the finer control from the previous point also helps), so if the LTF does not deviate much from its expectation, we should expect the same for cuts that are close to it. We can then take a union bound over the set of LTFs, which has a subexponential size (see Theorem~\ref{thm:ltf-count}). This can be thought of a one-step \emph{chaining} argument.
\end{itemize}
We remark that the analysis sketched here only concerns the noisy hypercube, but our goal is to prove Unique Games hardness for bounded-degree Max-Cut. Hence, additional care is required to lift this argument to the set-up of Unique Games hard instances of Max-Cut so that we do not lose additional factors.

\part{Algorithm}
\section{Algorithm with advantage $\tilde{\Omega}(1/d)$}

Recall that, for a graph $G = (V, E)$ on $n$ vertices with maximum degree upper bounded by $d$, a weighted Max 2-LIN(2) instance consists of set of weights $\{w_{ij}\}_{\{i,j\} \in E}$ with $w_{ij} \in \R_{\geq 0}$ and signs $\{b_{ij}\}_{\{i,j\}\in E}$ with $b_{ij} \in \{\pm 1\}$. By scaling, we assume that the total weight $W = \sum_{\{i, j \} \in E} w_{ij} = 1$. 
An assignment to a Max 2-LIN(2) instance consists of a labeling of the vertices of the graph with values in $\{\pm 1\}$, and the value of an assignment $x \in \{\pm 1\}^n$ is 
\[\mathrm{val}_E(x) = \sum_{\{i,j\}\in E} w_{ij} \frac{1 + b_{ij} x_i x_j}{2}\,.\]
Given a Max 2-LIN(2) instance, we denote by $\mathrm{OPT}$ the maximum value achieved by an assignment, that is
\[\mathrm{OPT} = \max_{x \in \{\pm 1\}^n} \mathrm{val}_E(x)\,.\]

Our main result of this section is the following theorem.
\begin{theorem}
\label{thm:main}
There exists a randomized algorithm that takes as input a Max 2-LIN(2) instance supported on a graph with $n$ vertices and maximum degree upper bounded by $d$, runs in time polynomial in $n$, and returns an assignment $x \in \Set{\pm 1}^n$ with
\[\E \mathrm{val}_E(x) \geq \Paren{\alpha_{GW} + \Omega\Paren{\frac{1}{d \log d}}} \cdot \mathrm{OPT}\,.\]
\end{theorem}

The Goemans-Williamson SDP with triangle inequalities for this problem is
\begin{align*}
\mathrm{SDP}
= \max_{v_1,\ldots,v_n \in \R^n}\quad
&\sum_{\{i,j\}\in E} w_{ij} \frac{1 + b_{ij}\Iprod{v_i,v_j}}{2}\\
\text{subject to}\quad
&\Norm{v_i}_2^2 = 1 \qquad \forall i \in [n]\,,\\
&\sigma_i\sigma_j \Iprod{v_i,v_j}
  + \sigma_i\sigma_k \Iprod{v_i,v_k}
  + \sigma_j\sigma_k \Iprod{v_j,v_k}
  \geq -1\\
&\hspace{5.5cm}
\forall i,j,k \in [n] \text{ distinct},\ 
\sigma_i,\sigma_j,\sigma_k \in \Set{\pm 1}\,.
\end{align*}
Given a solution to the SDP, draw $g \sim N(0, I_n)$ and let $x_i^0 = \sign(\langle g, v_i\rangle)$ for all $i \in [n]$.
Then the analysis of Goemans and Williamson~\cite{MR1412228-Goemans95} shows that $\E \mathrm{val}_E(x^0) \geq \alpha_{GW} \cdot \mathrm{SDP} \geq \alpha_{GW} \cdot \mathrm{OPT}$, where we recall that 
\[\rho^* = \argmin_{-1 \leq \rho \leq 1} \frac{\arccos(\rho)/\pi}{(1 - \rho)/2},\quad \alpha_{GW} = \frac{\arccos(\rho^*)/\pi}{(1 - \rho^*)/2}\,,\]
with $\rho^* \approx -0.689$ and $\alpha_{GW} \approx 0.878$.

Our algorithm is the following:

\begin{algorithm}[H]
\caption{Main Algorithm}
\label{alg:main}
\begin{algorithmic}[1]
    \State Set $K = C_0 \log d$ for some $C_0 > 0$ large enough. Set $\eta = \sqrt{1-K^{-1}}$. Set $\lambda = c_0 / \sqrt{K}$ for some $c_0 > 0$ small enough.
    \State Solve the Goemans-Williamson SDP with triangle inequalities, which yields unit vectors $v_1, \ldots, v_n \in \R^{n}$.
    \State Let $\rho_{ij} = b_{ij} \langle v_i, v_j\rangle$. Let $I=[-\rho^* - 0.01, -\rho^* + 0.01]$, and let the set $E_T \subseteq E$ consist of all edges with $\rho_{ij} \in I$.
    Let $W_i = \sum_{\Set{i,j} \in E_T} w_{ij}$ and $W_T = \sum_{\Set{i,j} \in E_T} w_{ij}$.
    \State Draw $g \sim N(0, I_n)$.
    \State Let $x^0_{i} = \sign(\langle g, v_i\rangle)$ for all $i \in [n]$.
    \State
    Let $h_{ij} = \Iprod{g, \frac{b_{ij}v_j - \rho_{ij}v_i}{\sqrt{1-\rho^2_{ij}}}}$ for all $\{i,j\} \in E_T$, and let
    \[\psi_i = \sum_{\Set{i,j} \in E_T} \frac{w_{ij}}{W_i} \sign(h_{ij})\,, \quad f_i = \sum_{\Set{i,j} \in E_T} \frac{w_{ij}}{W_i} \Paren{2\Phi\Paren{\frac{\eta h_{ij}}{\sqrt{1-\eta^2}}} - 1}\,,\]
    where $\Phi$ is the standard Gaussian CDF function.
    If $W_i = 0$, let $\psi_i = f_i = 0$.
    \State Let $x^\lambda_{i} = \sign(\langle g, v_i\rangle + \lambda f_i)$ for all $i \in [n]$. 
    \State Let
    \[\mathrm{SDP}_T = \sum_{\Set{i,j} \in E_T} w_{ij} \frac{1 + b_{ij}\langle v_i, v_j\rangle}{2}\,, \quad\quad \mathrm{SDP}_N = \sum_{\Set{i,j} \in E \setminus E_T} w_{ij} \frac{1 + b_{ij}\langle v_i, v_j\rangle}{2}\,,\]
    and $\mathrm{SDP} = \mathrm{SDP}_T + \mathrm{SDP}_N$.
    \State Let $\epsilon = \Omega(\frac{1}{d \log d})$ small enough. If $\mathrm{SDP}_N \geq \epsilon \cdot \mathrm{SDP}$, return $x^0$. Else, return $x^\lambda$.
\end{algorithmic}
\end{algorithm}

The main component of the analysis is the following guarantee, which lower bounds improvement in approximation ratio on the edges in $E_T$ that comes from using $x^{\lambda}$ instead of $x^0$:

\begin{theorem}
\label{thm:main-tight}
Suppose that, for some $0 < \rho_0 < 1$ and $0 \leq \beta < 1$, it holds for each $\{i, j\} \in E_T$ that $|\rho_{ij}| \leq \rho_0$.
Also suppose that, for each $i \in V$ and $j \neq k \in V$ such that $\Set{i,j}, \Set{i,k} \in E_T$, it holds that $\E[h_{ij}h_{ik}] \geq -\beta$.
Then
\[\E \mathrm{val}_{E_T}(x^\lambda) \geq \E \mathrm{val}_{E_T}(x^0) + \Omega\Paren{\frac{W_T}{d \log d}}\,.\]
\end{theorem}

Given Theorem~\ref{thm:main-tight}, it is not hard to prove the main guarantee.
We split the analysis into two cases, matching the last line of Algorithm~\ref{alg:main}: $\mathrm{SDP}_N \geq \epsilon \cdot \mathrm{SDP}$ or not.

\begin{lemma}
\label{lem:sdp-n}
If $\mathrm{SDP}_N \geq \epsilon \cdot \mathrm{SDP}$, then $\E \mathrm{val}_E(x^0) \geq (\alpha_{GW} + \Omega(\frac{1}{d \log d})) \cdot \mathrm{OPT}$.
\end{lemma}
\begin{proof}
If $\rho_{ij} \not\in I$, there exists some constant $\kappa > 0$ such that 
\[\frac{(\pi+2\arcsin(\rho_{ij}))/(2\pi)}{(1+\rho_{ij})/2} \geq \alpha_{GW} + \kappa\,,\]
and therefore, by the standard Goemans-Williamson analysis,
\begin{align*}
\E \mathrm{val}_{E}(x^0)
&\geq \alpha_{GW} \cdot \mathrm{SDP}_T + (\alpha_{GW} + \kappa) \cdot \mathrm{SDP}_N\\
&\geq (\alpha_{GW} + \kappa \epsilon) \mathrm{SDP}\\
&\geq (\alpha_{GW} + \kappa \epsilon) \mathrm{OPT}\,.
\end{align*}
The claim follows because $\kappa \epsilon = \Omega(\frac{1}{d \log d})$.
\end{proof}

\begin{lemma}
\label{lem:sdp-t}
If $\mathrm{SDP}_T \geq (1-\epsilon) \cdot \mathrm{SDP}$, then $\E \mathrm{val}_E(x^\lambda) \geq (\alpha_{GW} + \Omega(\frac{1}{d \log d})) \cdot \mathrm{OPT}$.
\end{lemma}
\begin{proof}
The proof will apply Theorem~\ref{thm:main-tight} on $E_T$.
Let us first argue that we satisfy the assumptions of this theorem.
First, we have $\Set{i,j} \in E_T$ only if $\rho_{ij} \in [-\rho^*-0.01,-\rho^*+0.01]$.
Using that $\rho^* \approx -0.689$, this implies that $|\rho_{ij}| \leq \rho_0$ for some $\rho_0 < 1$, e.g., $\rho_0 = 0.7$.
Second, consider $\Set{i,j}, \Set{i,k} \in E_T$ with $j \neq k$.
We are interested in the value
\begin{align*}
\E[h_{ij}h_{ik}]
&= \E \Iprod{g, \frac{b_{ij}v_j - \rho_{ij}v_i}{\sqrt{1-\rho_{ij}^2}}} \Iprod{g, \frac{b_{ik}v_k - \rho_{ik}v_i}{\sqrt{1-\rho_{ik}^2}}}
= \Iprod{ \frac{b_{ij}v_j - \rho_{ij}v_i}{\sqrt{1-\rho_{ij}^2}}, \frac{b_{ik}v_k - \rho_{ik}v_i}{\sqrt{1-\rho_{ik}^2}}}\,.
\end{align*}
The triangle inequality constraints imply that 
\begin{align*}
\langle b_{ij} v_j, b_{ik} v_k\rangle
\geq b_{ij} \langle  v_i, v_j\rangle + b_{ik} \langle  v_i, v_k\rangle - 1
= \rho_{ij} + \rho_{ik} - 1\,.
\end{align*}
Then, continuing the calculation for $\E[h_{ij}h_{ik}]$, we have
\begin{align*}
\E[h_{ij}h_{ik}]
&\geq \frac{\rho_{ij} + \rho_{ik} - 1 - \rho_{ij}\rho_{ik}}{\sqrt{1-\rho_{ij}^2}\sqrt{1-\rho_{ik}^2}}
= -\frac{(1-\rho_{ij})(1-\rho_{ik})}{\sqrt{1-\rho_{ij}^2}\sqrt{1-\rho_{ik}^2}}
= - \sqrt{\frac{(1-\rho_{ij})(1-\rho_{ik})}{(1+\rho_{ij})(1+\rho_{ik})}}\,,
\end{align*}
which, because $0.6 \leq \rho_{ij}, \rho_{ik} \leq 0.7$, is lower bounded by $-\beta_0$ for some $0 \leq \beta_0 < 1$ as required.

Then we apply Theorem~\ref{thm:main-tight}.
Note that $W_T \geq \mathrm{SDP}_T$.
Denote the coefficient of $W_T$ in the gain in Theorem~\ref{thm:main-tight} by $\Gamma$; we have $\Gamma = \Omega\Paren{\frac{1}{d \log d}}$.
Then
\begin{align*}
\E \mathrm{val}_{E_T}(x^\lambda)
&\geq \E \mathrm{val}_{E_T}(x^0) + \Gamma \cdot W_T\\
&\geq \alpha_{GW} \cdot \mathrm{SDP}_T +  \Gamma \cdot \mathrm{SDP}_T\\
&\geq \Paren{\alpha_{GW} + \Gamma} (1-\epsilon) \cdot \mathrm{SDP}\,.
\end{align*}
Selecting $\epsilon = \Gamma / 2$, we get
\[ \Paren{\alpha_{GW} + \Gamma} (1-\epsilon) = \alpha_{GW} + \Gamma - \Gamma \alpha_{GW} / 2 - \Gamma^2 / 2 = \alpha_{GW} + \Omega(\Gamma)\,,\]
giving $\E \mathrm{val}_{E_T}(x^\lambda) \geq (\alpha_{GW} + \Omega(\frac{1}{d\log d})) \cdot \mathrm{SDP} \geq (\alpha_{GW} + \Omega(\frac{1}{d\log d})) \cdot \mathrm{OPT}$.
Finally, as all edge contributions are non-negative, the same lower bound holds for $\E \mathrm{val}_{E}(x^\lambda)$.
\end{proof}

\begin{proof}[Proof of Theorem~\ref{thm:main}]
If $\mathrm{SDP}_N \geq \epsilon \cdot \mathrm{SDP}$, the result follows by Lemma~\ref{lem:sdp-n}.
Else, we have $\mathrm{SDP}_T \geq (1-\epsilon) \cdot \mathrm{SDP}$ and the result follows by Lemma~\ref{lem:sdp-t}.
\end{proof}

In the remainder of this section, we prove Theorem~\ref{thm:main-tight}.
For the remainder of the proof, fix $\rho_0, \beta$ as in its assumptions.
All implicit constants may depend on $\rho_0$ and $\beta$.

\subsection{Preliminary Facts}

For $0 \leq \eta < 1$, let $U_\eta$ be the Ornstein--Uhlenbeck operator defined on functions $f : \R^n \to \R$ satisfying
$\E_{g \sim N(0,I_n)}[f(g)^2] < \infty$ by
\[(U_\eta f)(t) = \E_{g' \sim N(0,I_n)} f(\eta t + \sqrt{1-\eta^2}\,g')
\]
for all $t \in \R^n$.
We note that, in Algorithm~\ref{alg:main}, $f_i$ is obtained by applying $U_\eta$ to $\psi_i$, which follows immediately from the following fact:

\begin{fact}
\label{fact:smooth-sign}
In dimension one, we have $(U_\eta \sign)(t) = 2\Phi(\eta t / \sqrt{1-\eta^2}) - 1$ for all $t \in \R$.
\end{fact}
\begin{proof}
We have
\begin{align*}
(U_\eta \sign)(t)
&= \E_{g'} \sign\Paren{\eta t + \sqrt{1-\eta^2}g'}\\
&= \P\Paren{\eta t + \sqrt{1-\eta^2} g' \geq 0} - \P\Paren{\eta t + \sqrt{1-\eta^2} g' < 0}\\
&= \Paren{1 - \Phi\Paren{-\eta t/\sqrt{1-\eta^2}}} - \Phi\Paren{-\eta t/\sqrt{1-\eta^2}}\\
&= 2\Phi\Paren{\eta t/\sqrt{1-\eta^2}} - 1\,.
\end{align*}
\end{proof}

Let us now deduce some properties about the variables involved in the algorithm, and in particular $h_{i,j}$.

\begin{lemma}
\label{lem:n01}
For each $\Set{i,j} \in E_T$, we have $h_{ij} \sim N(0, 1)$.
\end{lemma}
\begin{proof}
We have that $h_{i,j}$ is distributed as a Gaussian random variable with mean $0$ and variance
\begin{align*}
\Norm{b_{ij} \frac{v_{j} - \langle v_i, v_j\rangle v_i}{\sqrt{1-\langle v_i, v_j\rangle^2}}}^2
&= \frac{1}{1-\langle v_i, v_j\rangle^2} \Norm{v_{j} - \langle v_i, v_j\rangle v_i}^2
&= \frac{1}{1-\langle v_i, v_j\rangle^2} \Paren{1 - \langle v_i, v_j\rangle^2}
&= 1\,.
\end{align*}
\end{proof}

\begin{lemma}
\label{lem:indep}
For each $\Set{i,j} \in E_T$, we have that $\langle g, v_i\rangle$ is independent of $h_{ij}$.
\end{lemma}
\begin{proof}
We have that $\langle g, v_i\rangle$ and $h_{ij}$ are jointly Gaussian, so it suffices to prove that they have correlation zero, that is
\[\Iprod{v_i, b_{ij} \frac{v_{j} - \langle v_i, v_j\rangle v_i}{\sqrt{1-\langle v_i, v_j\rangle^2}}} = 0\,,\]
which follows because $\Iprod{v_i, v_{j} - \langle v_i, v_j\rangle v_i} = 0$.
\end{proof}

\begin{lemma}
\label{lem:decomp}
For each $\Set{i,j} \in E_T$, we have that 
\[b_{ij}\langle g, v_j\rangle = \rho_{ij} \langle g, v_i\rangle + \sqrt{1-\rho_{ij}^2} h_{ij}\,.\]
\end{lemma}
\begin{proof}
Expanding the definition of $h_{ij}$ in $\rho_{i,j} \langle g, v_i\rangle + \sqrt{1-\rho_{i,j}^2} h_{ij}$, we have
\begin{align*}
\rho_{ij} \langle g, v_i\rangle + \sqrt{1-\rho_{ij}^2} h_{ij}
&= \rho_{ij} \langle g, v_i\rangle + \sqrt{1-\rho_{ij}^2} \Iprod{g, \frac{b_{ij}v_j - \rho_{ij} v_i}{\sqrt{1-\rho_{ij}^2}}}\\
&= \rho_{ij} \langle g, v_i\rangle + \Iprod{g, b_{ij} v_j - \rho_{ij} v_i}\,,
\end{align*}
from which the claimed equality follows by cancellation.
\end{proof}

\subsection{Smoothed Correlation Bounds}
Let $a_i = \E[f_i \psi_i]$, which will naturally appear later in our bounds.
The main result in this section is a lower bound on $a_i$.

\begin{fact}
\label{fact:smooth-cor}
For two random variables $x, y \sim N(0, 1)$ jointly Gaussian with $\E[xy]=\rho$, we have 
\[\E[((U_\eta \sign)(x))\sign(y)] = \frac{2}{\pi} \arcsin(\eta \rho)\,.\]
\end{fact}
\begin{proof}
We have
\begin{align*}
\E[((U_\eta \sign)(x))\sign(y)]
&= \E\Brac{\E_{g' \sim N(0,1)}\Brac{\sign(\eta x + \sqrt{1-\eta^2}g')} \cdot \sign(y)}\,.
\end{align*}
Let $x' = \eta x + \sqrt{1-\eta^2} g'$, which is also distributed as $x' \sim N(0, 1)$.
Then 
\[\E[((U_\eta \sign)(x))\sign(y)] = \E[\sign(x')\sign(y)]\,,\]
and by Sheppard's formula, $\E[\sign(x')\sign(y)] = \frac{2}{\pi} \arcsin(\E[x' y])$.
The result follows since $\E[x'y] = \eta \rho$.
\end{proof}

\begin{lemma}
\label{lem:energy-lb}
For every $i \in [n]$ with $W_i > 0$, it holds that $a_i \geq \Omega\Paren{\frac{1}{d \sqrt{K}}}$.
\end{lemma}
\begin{proof}
We have by Fact~\ref{fact:smooth-cor} that
\begin{align*}
a_i &= \sum_{\substack{j, k \in [n]\\\{i,j\},\{i,k\} \in E_T}} \frac{w_{ij}}{W_i} \frac{w_{ik}}{W_i} \E[((U_\eta \sign)(h_{ij}))\sign(h_{ik})]\\
&=\frac{2}{\pi} \sum_{\substack{j, k \in [n]\\\{i,j\},\{i,k\} \in E_T}} \frac{w_{ij}}{W_i} \frac{w_{ik}}{W_i} \arcsin(\eta \E[h_{ij}h_{ik}])\,.
\end{align*}
Recall the Taylor expansion of the $\arcsin$ function:
\[\arcsin(x) = \sum_{\ell=0}^\infty c_\ell x^{2\ell + 1}\,,\quad c_\ell = \frac{(2\ell)!}{4^{\ell} (\ell!)^2 (2\ell+1)} = \Theta(\ell^{-3/2})\,.\]
Define 
\[q_\ell = \sum_{\substack{j, k \in [n]\\\{i,j\},\{i,k\} \in E_T}} \frac{w_{ij}}{W_i} \frac{w_{ik}}{W_i} \E[h_{ij}h_{ik}]^\ell\,,\]
where by the Schur product theorem we have $q_\ell \geq 0$ for all $\ell$.
Recall that, for $i\neq j\neq k$, we have $\E[h_{ij}h_{ik}] \geq -\beta$.
Let $L = \lfloor K / 10 \rfloor$ for $K$ selected large enough in terms of $\beta$ such that $\beta^{2L+1} \leq \frac{1}{2d}$.
For $\ell \geq L$, the contribution of the off-diagonal entries to $q_{2\ell+1}$ is at least $-\frac{1}{2d}$.
On the other hand, the contribution of the diagonal entries is at least
\[\sum_{\substack{j \in [n]\\\{i,j\} \in E_T}} \Paren{\frac{w_{ij}}{W_i}}^2 \geq 1/d\,,\]
where we used that $\frac{w_{ij}}{W_i}$ sum up to $1$ and there are at most $d$ such terms; then the bound follows because $\Norm{x}_2 \geq \Norm{x}_1 / \sqrt{d}$ for a $d$-dimensional vector $x$.
Therefore, for $\ell \geq L$, we have $q_{2\ell+1} \geq \frac{1}{2d}$.
Then
\begin{align*}
a_i &= \frac{2}{\pi} \sum_{\ell = 0}^{\infty} c_{\ell}\eta^{2\ell+1} q_{2\ell+1}
\geq \frac{2}{\pi} \sum_{\ell = L}^{2L} c_{\ell}\eta^{2\ell+1}q_{2\ell+1}
\geq \frac{1}{\pi d} \sum_{\ell = L}^{2L} \Theta(\ell^{-3/2} \eta^{2\ell+1})\,.
\end{align*}
Note that $\eta^{4L+1} = \sqrt{1-K^{-1}}^{4\lfloor K/10\rfloor + 1} \geq \Omega(1)$. Then
\[a_i \geq \Omega\Paren{\frac{1}{d} L^{-1/2}} = \Omega\Paren{\frac{1}{d\sqrt{K}}}\,.\]
\end{proof}

We will also want a bound on $\E f_i^2$ in terms of $a_i$.

\begin{fact}[See Proposition 11.37 in~\cite{ODonnell14}]
\label{fact:hermite-smooting}
Consider the multivariate Hermite expansion of a function $f = \sum_{\alpha} c_\alpha H_\alpha$.
Then
\[U_\eta f = \sum_\alpha \eta^{|\alpha|} c_\alpha H_\alpha\,.\]
\end{fact}

\begin{lemma}
\label{lem:energy-var}
We have $\E f_i^2 \leq a_i$.
\end{lemma}
\begin{proof}
We need to prove that $\E (U_\eta \psi_i)^2 \leq \E (U_\eta \psi_i) \psi_i$.
Write $\psi_i$ in its multivariate Hermite expansion as $\psi_i = \sum_{\alpha} c_\alpha H_\alpha$.
Then, by Fact~\ref{fact:hermite-smooting}, $U_\eta \psi_i = \sum_{\alpha} \eta^{|\alpha|} c_\alpha H_\alpha$.
Then
\[\E (U_\eta \psi_i)^2 = \sum_{\alpha} \eta^{2|\alpha|} c_\alpha^2\]
and
\[\E (U_\eta \psi_i) \psi_i = \sum_{\alpha} \eta^{|\alpha|} c_\alpha^2\,,\]
and the desired bound holds term-wise because $\eta < 1$.
\end{proof}

\subsection{Gaussian Estimates}

In this section we prove a few estimates that will be crucial in the analysis.

\begin{lemma}\label{lem:slice-max}
Let $(X,Y)$ be a standard Gaussian vector, where $X$ is one coordinate and $Y$ consists of the remaining coordinates, and let $F=U_\eta f(X,Y)$ for $0\leq \eta < 1$ and $|f| \leq 1$.
Put $K=(1-\eta^2)^{-1}$ and $A=\langle U_\eta f,f\rangle$.
There are universal constants $c_0,C_0>0$ such that, for every $0\le s\le c_0/\sqrt K$,
\[\E_Y\Brac{\sup_{|x|\le s}|F(x,Y)|^2}\le C_0\sqrt K\,A\,.\]
\end{lemma}
\begin{proof}
Let $U_\eta^X$ and $U_\eta^Y$ be the Ornstein–Uhlenbeck operators acting in the $X$ and $Y$ variables, respectively.
For $0 \leq r < 1$, define the Mehler kernel of $U_r$ as
\[M_r(x,z) = \frac{1}{\sqrt{1-r^2}} \exp\Paren{\frac{2rxz -r^2(x^2+z^2)}{2(1-r^2)}}\,,\]
such that $(U_r u)(x) = \E_{z \sim N(0,1)}[u(z) M_r(x,z)]$.
By Cauchy-Schwarz,
\[|(U_r u)(x)|^2 \leq \E_{z \sim N(0,1)}[u(z)^2] \E_{z \sim N(0,1)}[M_r(x,z)^2]\,.\]
By the semigroup identity of the kernel, $\E_{z \sim N(0,1)}[M_r(x,z)^2] = M_{r^2}(x,x)$.
We note that
\[M_{r^2}(x,x) = \frac{1}{\sqrt{1-r^4}} \exp\Paren{\frac{2r^2x^2 - 2r^4x^2}{2(1-r^4)}} = \frac{1}{\sqrt{1-r^4}}\exp\Paren{\frac{r^2x^2}{1+r^2}}\,.\]
Therefore 
\begin{equation}
\label{eq:bound-ou}
|(U_r u)(x)|^2 \leq \frac{1}{\sqrt{1-r^4}}\exp\Paren{\frac{r^2x^2}{1+r^2}} \E_{z \sim N(0,1)}[u(z)^2]\,.
\end{equation}
Let now $r = \sqrt{\eta}$ and $h = U_r f$.
We have
\[A = \langle U_\eta f, f\rangle = \langle U_r f, U_r f\rangle = \E_{(z,z') \sim N(0, I)}[h(z,z')^2]\]
and $F = U_\eta f = U_r h = U_r^X U_r^Y h$.
Let $g = U_r^Y h$, such that $F(\cdot, y)= (U_r^X g)(\cdot, y)$.
Then, applying Equation~\ref{eq:bound-ou} and using that $r = \sqrt{\eta}$,
\[|F(x, y)|^2 \leq \sqrt{K} \exp\Paren{\frac{\eta x^2}{1+\eta}} \E_{z \sim N(0, 1)}[g(z, y)^2]\,.\]
Suppose $|x| \leq s \leq 1/\sqrt{K} = \sqrt{1-\eta^2}$.
Then $\eta x^2/(1+\eta) \leq \eta (1-\eta^2)/(1+\eta) = \eta(1-\eta) \leq 1/4$.
Then, taking the supremum gives
\[\sup_{|x| \leq s}|F(x, y)|^2 \leq e^{1/4}\sqrt{K} \E_{z \sim N(0, 1)}[g(z, y)^2]\,.\]
Finally, integrating in $y$,
\begin{align*}
\E_Y \Brac{\sup_{|x| \leq s}|F(x, y)|^2}
&\leq e^{1/4}\sqrt{K} \E_{(z,z') \sim N(0,I)} [g(z,z')^2]\\
&\leq e^{1/4}\sqrt{K} \E_{(z,z') \sim N(0,I)} [h(z,z')^2]\\
&= e^{1/4}\sqrt{K}A\,,
\end{align*}
where we used that $T^Y_r$ is a contraction.
Taking $c_0 = 1$ and $C_0 = e^{1/4}$ proves the result.
\end{proof}

\begin{lemma}\label{lem:one-strip}
Let $(X,Y)$ be a standard Gaussian vector, where $X$ is one coordinate and $Y$ consists of the remaining coordinates, and let $F=U_\eta f(X,Y)$ with $0\leq \eta < 1$ and $|f|\le1$.
Put $K=(1-\eta^2)^{-1}$ and $A=\langle U_\eta f,f\rangle$.
There are universal constants $c,C>0$ such that, for $0\le t\le c/\sqrt K$,
\begin{equation}\label{eq:one-strip}
  \E\bigl[|F(X,Y)|\1\{|X|\le t|F(X,Y)|\}\bigr]
  \le Ct\sqrt K\,A.
\end{equation}
\end{lemma}
\begin{proof}
For fixed $Y = y$, let $M(y) = \sup_{|x| \leq t} |F(x,y)|$.
Because $|f| \leq 1$, it also follows that $|F| \leq 1$, so $0 \leq M(y) \leq 1$.
Note that if $|x| \leq t |F(x,y)|$, then $|x| \leq t$ and therefore we also have $|x| \leq t M(y)$ and $|F(x,y)| \leq M(y)$.
Therefore, conditioning on $Y=y$,
\[\E_X \Brac{ |F(X,y)| \1\Set{|X| \leq t|F(X,y)|} } \leq M(y) \P(|X|\leq t M(y)) \leq O(t M(y)^2)\,,\]
where we used that the density of $X$ is upper bounded by a constant.
Averaging over $Y$ and applying Lemma~\ref{lem:slice-max} with $t$ as $s$, we get
\[\E\bigl[|F(X,Y)|\1\{|X|\le t|F(X,Y)|\}\bigr] \leq O(t \sqrt{K}A)\,.\]
\end{proof}

\begin{lemma}\label{lem:transverse-two-strip}
Let $X,Y$ be standard Gaussians with $\E[XY]=\rho$ and $|\rho|\le\rho_0<1$.
Put $\sigma=\sqrt{1-\rho^2}$ and write
\begin{equation}\label{eq:H-Q-def}
Y=\rho X+\sigma H, \qquad Q=\sigma X-\rho H=\frac{X-\rho Y}{\sigma}\,.
\end{equation}
Thus $H$ is independent of $X$, and $Q$ is independent of $Y$.
Let $Z$ be an additional standard Gaussian vector, independent of $(X,H)$.
Suppose for $0\leq \eta < 1$
\begin{equation}\label{eq:transverse-functions}
F=U_\eta f(H,Z), \qquad G=U_\eta g(Q,Z),\qquad |f|,|g|\le1.
\end{equation}
Let
\begin{equation}\label{eq:AF-AG-def}
A_F=\langle U_\eta f,f\rangle, \qquad A_G=\langle U_\eta g,g\rangle, \qquad K=(1-\eta^2)^{-1}.
\end{equation}
There are constants $c_{\rho_0},C_{\rho_0}>0$ such that, for $0\le t\le c_{\rho_0}/\sqrt K$,
\begin{equation}\label{eq:transverse-two-strip}
\Prb\bigl(|X|\le t|F|,\ |Y|\le t|G|\bigr) \le C_{\rho_0}t^2\sqrt K\,(A_F+A_G).
\end{equation}
\end{lemma}
\begin{proof}
Let $\sigma_0 = \sqrt{1-\rho_0^2}$ and $R = \frac{1+\rho_0}{\sigma_0}$.
Note that $\sigma \geq \sigma_0$.
For fixed $Z=z$, define $M_F(z) = \sup_{|h|\leq R t} |F(h,z)|$ and $M_G(z)=\sup_{|q|\leq Rt} |G(q,z)|$.
Because $|f|, |g| \leq 1$, we also have $|F|,|G| \leq 1$.
Then, note that if $|X| \leq t|F|$ and $|Y| \leq t|G|$, we have $|X| \leq t$ and $|Y| \leq t$. In this event, also
\[|H| \leq \frac{|Y|+|\rho| |X|}{\sigma} \leq Rt\,,\quad |Q| \leq \frac{|X|+|\rho||Y|}{\sigma} \leq Rt\,.\]
Hence, in this event we have $|F(H,z)| \leq M_F(z)$ and $|G(Q,z)| \leq M_G(z)$.
Because the joint density of $(X, Y)$ is upper bounded by $\frac{1}{2\pi \sigma} \leq \frac{1}{2\pi \sigma_0}$, and because $|F(H,z)| \leq M_F(z)$ and $|G(Q,z)| \leq M_G(z)$ if $|X| \leq t|F|$ and $|Y| \leq t|G|$, we can upper bound
\[\P(|X| \leq t|F|, |Y| \leq t|G| \mid Z=z) \leq C_{\rho_0} t^2 M_F(z) M_G(z)\,.\]
Averaging over $Z$ and applying Cauchy-Schwarz,
\begin{align}
\P(|X| \leq t|F|, |Y| \leq t|G|)
&\leq C_{\rho_0} t^2 \E[M_F(Z) M_G(Z)]
\leq C_{\rho_0} t^2 \sqrt{\E[M_F(Z)^2]} \sqrt{\E[M_G(Z)^2]}\,.
\end{align}
Choose $c_{\rho_0}$ small enough such that $R t$ is small enough to apply Lemma~\ref{lem:slice-max} with $Rt$ as $s$.
Then, Lemma~\ref{lem:slice-max} gives $\E[M_F(Z)^2] \leq O(\sqrt{K} A_F)$ and $\E[M_G(Z)^2] \leq O( \sqrt{K} A_G)$.
Therefore
\[\P(|X| \leq t|F|, |Y| \leq t|G|) \leq O\Paren{C_{\rho_0} t^2 \sqrt{K} \sqrt{A_F A_G}} \leq O\Paren{C_{\rho_0} t^2 \sqrt{K} (A_F + A_G)}\,.\]
\end{proof}

\subsection{Lower Bounding the Gain}
We are interested in lower bounding
\[\Delta_{\mathrm{val}} = \E \mathrm{val}_{E_T}(x^\lambda) - \E \mathrm{val}_{E_T}(x^0) = \E \sum_{\Set{i,j} \in E_T} w_{ij} \frac{b_{ij}(x_i^\lambda x_j^\lambda - x_i^0 x_j^0)}{2}\,.\]
Define $\delta_i = x_i^\lambda - x_i^0$.
Then, for each $\Set{i,j} \in E_T$, we can write
\[x_i^\lambda x_j^\lambda - x_i^0 x_j^0 = \delta_i x_j^0 + x_i^0 \delta_j + \delta_i \delta_j\,.\]
Define $L_{i \to j} = \frac{1}{2} w_{ij} b_{ij} \E[\delta_i  x_j^0]$ and $B_{ij} = \frac{1}{2} w_{ij} b_{ij} \E[\delta_i\delta_j]$, such that
\[\Delta_{\mathrm{val}} = \sum_{\substack{i, j \in [n]\\ \Set{i,j} \in E_T}} L_{i \to j} + \sum_{\Set{i,j} \in E_T} B_{ij}\,.\]
Note that the indices in the first sum are ordered, so that we sum up both $L_{i \to j}$ and $L_{j \to i}$. Next, let us rewrite $L_{i \to j}$.
Define 
\[M_{i \to j} = w_{ij} \lambda \varphi(0) \E[f_i \cdot \sign(h_{ij})]\,,\]
\[D_{i \to j} = w_{ij} \E \int_{0}^{\lambda f_i} (\varphi(t) - \varphi(0)) \Paren{\sign\Paren{\sqrt{1-\rho_{ij}^2} h_{ij} - \rho_{ij}t}} dt\,,\]
\[R_{i \to j} = w_{ij} \E \int_{0}^{\lambda f_i} \varphi(0) \Paren{\sign\Paren{\sqrt{1-\rho_{ij}^2} h_{ij} - \rho_{ij}t} - \sign(h_{ij})} dt\,,\]
where $\varphi : \R \to \R$ is the pdf of the standard Gaussian, and where the integrals are oriented (it is possible that $\lambda f_i < 0$).
\begin{lemma}
\label{lem:rw-l}
For $i, j \in [n]$ such that $\Set{i,j} \in E_T$, we have
\[L_{i \to j} = M_{i \to j} + D_{i \to j} + R_{i \to j}\,.\]
\end{lemma}
\begin{proof}
We have
\begin{align*}
L_{i \to j}
&= \frac{1}{2} w_{ij} \E\Brac{\Paren{\sign\Paren{\langle g, v_i\rangle + \lambda f_i}- \sign\Paren{\langle g, v_i\rangle}}\Paren{\sign(b_{ij} \langle g, v_j\rangle)}}\,.
\end{align*}
Let $g^{(i)}_\bot = g - \langle g, v_i\rangle v_i$, which is independent of $\langle g, v_i\rangle$.
By Lemma~\ref{lem:indep}, all $\{h_{ik}\}_{\{i,k\} \in E_T}$ are independent of $\langle g, v_i\rangle$, and hence also $\psi_i, f_i$ are independent of $\langle g, v_i\rangle$, so they are fixed when conditioning on $g^{(i)}_\bot$.
Also, by Lemma~\ref{lem:decomp}, $b_{ij}\langle g, v_j\rangle = \rho_{i,j} \langle g, v_i\rangle + \sqrt{1-\rho_{i,j}^2} h_{i,j}$.
Then 
\begin{align*}
L_{i \to j}
&= \frac{1}{2} w_{ij}\E_{g_i^{\bot}}\Brac{ \E_{g' \sim N(0, 1)} (\sign(g' + \lambda f_i) - \sign(g'))(\sign(\rho_{ij} g' + \sqrt{1-\rho_{ij}^2} h_{ij})) \Bigg| g^{(i)}_\bot}\\
&= \frac{1}{2} w_{ij} \E \int_{-\infty}^\infty \varphi(t) \Paren{\sign(t + \lambda f_i) - \sign(t)}\Paren{\sign\Paren{\rho_{ij}t + \sqrt{1-\rho_{ij}^2} h_{ij}}} dt\,.
\end{align*}
Next, note that $\sign(t + \lambda f_i) - \sign(t) = 2 \cdot 1[t \in [-\lambda f_i, 0]]$ if $\lambda f_i \geq 0$ and $\sign(t + \lambda f_i) - \sign(t) = -2 \cdot 1[t \in [0,-\lambda f_i]]$ if $\lambda f_i < 0$.
We can write this concisely as an oriented integral from $-\lambda f_i$ to $0$, 
\begin{align*}
L_{i \to j}
&= w_{ij} \E \int_{-\lambda f_i}^0 \varphi(t) \Paren{\sign\Paren{\rho_{ij}t + \sqrt{1-\rho_{ij}^2} h_{ij}}} dt\\
&= w_{ij} \E \int_{0}^{\lambda f_i} \varphi(t) \Paren{\sign\Paren{\sqrt{1-\rho_{ij}^2} h_{ij} - \rho_{ij}t}} dt\,.
\end{align*}
Then it is easily seen that
\[L_{i\to j} = M_{i \to j} + D_{i \to j} + R_{i \to j}\,,\]
where $M_{i \to j}$, $D_{i \to j}$, and $R_{i \to j}$ are as in the lemma statement.
\end{proof}

Next, define the global terms
\[\cM = \sum_{(i,j) : \Set{i,j}\in E_T} M_{i \to j}\,, \quad \cD = \sum_{(i,j) : \Set{i,j}\in E_T} D_{i \to j}\,, \quad \cR = \sum_{(i,j) : \Set{i,j}\in E_T} R_{i \to j}\,, \quad \cB = \sum_{\Set{i,j}\in E_T} B_{i,j}\,.\]
Note then that 
\[\Delta_{\mathrm{val}} = \cM + \cD + \cR + \cB\,.\]
We now bound each of these terms.
Let $S = \sum_{i \in [n]} W_i a_i$.

\begin{lemma}
\label{lem:bound-m}
We have
\[\cM = \lambda \varphi(0) S\,.\]
\end{lemma}
\begin{proof}
We have
\begin{align*}
\cM
&= \lambda \varphi(0)  \sum_{\substack{i, j \in [n]\\ \Set{i,j} \in E_T}} w_{ij} \E[f_i \cdot \sign(h_{ij})]
= \lambda \varphi(0) \sum_{i \in [n]} \E \Brac{f_i \sum_{\substack{j \in [n]\\\Set{i,j} \in E_T}} w_{ij} \sign(h_{ij})}\\
&= \lambda \varphi(0) \sum_{i \in [n]} W_i \E[f_i \psi_i]
= \lambda \varphi(0) S\,.
\end{align*}
\end{proof}

\begin{lemma}
\label{lem:bound-d}
We have
\[|\cD| \leq O(\lambda^2 S)\,.\]
\end{lemma}
\begin{proof}
We will use that $|\varphi(x) - \varphi(0)| \leq O(|x|)$ and $|\sign(x)| \leq 1$ for all $x \in \R$.
Then, using from Lemma~\ref{lem:energy-var} that $\E f_i^2 \leq a_i$,
\begin{align*}
|D_{i \to j}| &= \Abs{w_{ij} \E \int_{0}^{\lambda f_i} (\varphi(t) - \varphi(0)) \Paren{\sign\Paren{\sqrt{1-\rho_{ij}^2} h_{ij} - \rho_{ij}t}} dt}\\
&\leq O\Paren{w_{ij} \E \int_{0}^{\lambda |f_i|} t\, dt}
\leq O\Paren{w_{ij} \lambda^2 \E f_i^2}
\leq O\Paren{w_{ij} \lambda^2 a_i}\,.
\end{align*}
Then
\begin{align*}
|\cD|
&\leq \sum_{\substack{i, j \in [n]\\ \Set{i,j} \in E_T}} |D_{i \to j}|
\leq O\Paren{\lambda^2 \sum_{i \in [n]}  a_i \sum_{\substack{j \in [n]\\ \Set{i,j} \in E_T}} w_{ij}}
\leq O(\lambda^2 S)\,.
\end{align*}
\end{proof}

\begin{lemma}
\label{lem:bound-r}
We have
\[|\cR| \leq O(\lambda^2 \sqrt{K} S)\,.\]
\end{lemma}
\begin{proof}
We have
\begin{align*}
R_{i \to j} &= w_{ij} \E \int_{0}^{\lambda f_i} \varphi(0) \Paren{\sign\Paren{\sqrt{1-\rho_{ij}^2} h_{ij} - \rho_{ij}t} - \sign(h_{ij})} dt\,.
\end{align*}
Because $|\rho_{ij}| \leq \rho_0 < 1$, we have that $\sign\Paren{\sqrt{1-\rho_{ij}^2} h_{ij} - \rho_{ij}t}$ and $\sign(h_{ij})$ differ only if $|h_{ij}| \leq C |t|$, for a constant $C$ depending on $\rho_0$.
Then
\[|R_{i \to j}| \leq O\Paren{ w_{ij} \lambda \E\Brac{|f_i| \cdot \1\Brac{|h_{ij}| \leq C \lambda |f_i|}}}\,.\]
We will apply Lemma~\ref{lem:one-strip} with  $h_{ij}$ as $X$, the rest of the Gaussian coordinates as $Y$, $\psi_i$ as $f$, $f_i$ as $F$, and $C \lambda$ as $t$.
Also note that $|\psi_i| \leq 1$.
We note that $\lambda$ can be chosen appropriately small compared to $K$ such that $\lambda \ll 1/\sqrt{K}$ and the assumption of the lemma is satisfied.
Then Lemma~\ref{lem:one-strip} implies that 
\[\E\Brac{|f_i| \cdot \1[|h_{ij}| \leq C \lambda |f_i|]} \leq O(\lambda \sqrt{K} \cdot a_i)\,.\]
Then
\begin{align*}
|\cR|
&\leq \sum_{\substack{i, j \in [n]\\ \Set{i,j} \in E_T}} |R_{i \to j}|
\leq O\Paren{\lambda^2 \sqrt{K} \sum_{i \in [n]} a_i \sum_{\substack{j \in [n]\\ \Set{i,j} \in E_T}} w_{ij}}
\leq O(\lambda^2 \sqrt{K} S)\,.
\end{align*}
\end{proof}

\begin{lemma}
\label{lem:bound-b}
We have
\[|\cB| \leq O(\lambda^2 \sqrt{K} S)\,.\]
\end{lemma}
\begin{proof}
Note that $|\delta_i| \leq 2 \cdot 1[|\langle g, v_i\rangle| \leq \lambda |f_i|]$.
Then we have
\begin{align*}
|B_{ij}| &= \frac{1}{2} w_{ij}|\E[\delta_i\delta_j]|
\leq 2w_{ij} \P(|\langle g, v_i\rangle| \leq \lambda |f_i|, |\langle g, v_j\rangle| \leq \lambda |f_j|)\,.
\end{align*}
Then we will apply Lemma~\ref{lem:transverse-two-strip}.
Let us argue that we satisfy its assumptions.
We use $\langle g, v_i\rangle$ as $X$, $b_{ij}\langle g, v_j\rangle$ as $Y$, $\psi_i$ as $f$, $\psi_j$ as $g$, $f_i$ as $F$, and $f_j$ as $G$.
Note that $\E[b_{ij}\langle g, v_i\rangle\langle g, v_j\rangle] = \rho_{ij}$, which we use as $\rho$.
Then by Lemma~\ref{lem:decomp} we can decompose
\[b_{ij}\langle g, v_j\rangle = \rho_{ij} \langle g, v_i\rangle + \sqrt{1-\rho_{ij}^2} h_{ij}\,,\]
so we use $h_{ij}$ as $H$.
Also 
\[\frac{\langle g, v_i\rangle - \rho_{ij} b_{ij}\langle g, v_j\rangle}{\sqrt{1-\rho_{ij}^2}} = b_{ij} h_{ji}\,,\]
so we use $b_{ij} h_{ji}$ as $Q$.
We use as $Z$ the coordinates, in an orthonormal basis, of the projection of $g$ onto $\mathrm{span}\{v_i,v_j\}^{\perp}$.
We argue now that $\psi_i = \psi_i(H, Z)$ and $\psi_j = \psi_j(Q, Z)$.
We have that $\psi_i$ only depends on $g$ through terms $\{h_{ik}\}_{\{i,k\} \in E_T}$, all of which are projections of $g$ onto directions orthogonal to $v_i$; it follows that $\psi_i$ is only a function of $H$ and $Z$.
The argument is analogous for $\psi_j$.
Also note that $|\psi_i|, |\psi_j| \leq 1$.

We also use $\lambda$ as $t$, and we note that $\lambda$ can be chosen appropriately small compared to $K$ such that $\lambda \ll 1/\sqrt{K}$.
Then Lemma~\ref{lem:transverse-two-strip} implies that 
\[\P(|\langle g, v_i\rangle| \leq \lambda |f_i|, |\langle g, v_j\rangle| \leq \lambda |f_j|) \leq O(\lambda^2 \sqrt{K} (a_i + a_j))\,.\]
Then
\begin{align*}
|\cB|
&\leq \sum_{\Set{i,j}\in E_T} |B_{ij}|
\leq O\Paren{\lambda^2 \sqrt{K} \sum_{\Set{i,j}\in E_T} w_{ij} (a_i+a_j)}
\leq O(\lambda^2 \sqrt{K} S)\,.
\end{align*}
\end{proof}

\begin{proof}[Proof of Theorem~\ref{thm:main-tight}]
We have
\[\E \mathrm{val}_{E_T}(x^\lambda) - \E \mathrm{val}_{E_T}(x^0) = \cM + \cD + \cR + \cB\,.\]
Plugging in the bounds from Lemmas~\ref{lem:bound-m},~\ref{lem:bound-d},~\ref{lem:bound-r}, and~\ref{lem:bound-b}, we get that
\[\E \mathrm{val}_{E_T}(x^\lambda) - \E \mathrm{val}_{E_T}(x^0) \geq \lambda \varphi(0) S - O(\lambda^2 S) - O(\lambda^2 \sqrt{K} S) = \lambda S \Paren{\varphi(0) - O(\lambda \sqrt{K})}\,.\]
We note that $\lambda$ can be chosen appropriately small compared to $K$ such that the term $O(\lambda \sqrt{K})$ is a small constant.
Then we get
\[\E \mathrm{val}_{E_T}(x^\lambda) - \E \mathrm{val}_{E_T}(x^0) \geq \Omega(\lambda S)\,.\]
Finally, we lower bound $\lambda S$ using that by Lemma~\ref{lem:energy-lb} $a_i \geq \Omega(\frac{1}{d\sqrt{K}})$:
\begin{align*}
\lambda S 
&= \lambda \sum_{i \in [n]} W_i a_i 
\geq \Omega\Paren{\frac{\lambda}{d \sqrt{K}} \sum_{i \in [n]} W_i}
= \Omega\Paren{\frac{2\lambda}{d\sqrt{K}} W_T}
= \Omega\Paren{\frac{W_T}{d \log d}}\,.
\end{align*}
This proves the theorem.
\end{proof}

\part{Hardness}
\section{Preliminaries for Unique Games Hardness}

Throughout, we will use $\cL_R$ to denote the set of LTFs over $\{0,1\}^R$. For any function $f:\{0,1\}^R \to \mathbb{R}$ and $p \geq 1$, we use $\|f\|_p = \E_{x \sim\{0,1\}^R}[|f(x)|^p]^{1/p}$ to denote the $L_p$-expectation norm, and we use ${\sf dist}(f,\cL_R)$ to denote the $L_1$-distance of $f$ to the set $\cL_R$, i.e.,
\[
	{\sf dist}(f,\cL_R) := \min_{L \in \cL_R} \|f - L\|_1.
\]
For any string $x \in \{0,1\}^R$ and permutation $\pi:[R] \to [R]$, we use $\pi(x)$ to denote the permuted string $(x(\pi(1)),x(\pi(2)),\ldots,x(\pi(R)))$. Finally, for any function $f:\{0,1\}^R \to \mathbb{R}$ and permutation $\pi:[R] \to [R]$, we use $\pi \circ f$ to denote the function $f$ after permuting its input according to $\pi$, i.e., $\pi \circ f(x) = f(\pi(x))$. 

\subsection{Fourier Analysis}

We recall some basic concepts and tools from Fourier analysis on the Boolean cube from \cite{ODonnell14}. For any integer $R \in \mathbb{N}$, let $L_2(\{0,1\}^R)$ denote the vector space of functions $f:\{0,1\}^R \to \mathbb{R}$ endowed with the expectation inner product
\[
	\langle f,g \rangle = \Ex_{x \sim \{0,1\}^R}\big[f(x)g(x)\big].
\] 

{\bf Fourier Basis}. Let $\sigma:\{0,1\} \to \{\pm 1\}$ be defined as $\sigma(x) := 1 - 2x$. For any $S \subseteq [R]$, let $\chi_S:\{0,1\}^R \to \{\pm 1\}$ be the function defined as 
\[
	\chi_{S}(x) = \prod_{i \in S} \sigma(x(i)).
\]
It is well known that the set of functions $\{\chi_S\}_{S \subseteq [R]}$ forms the {\em Fourier basis}, which is an orthonormal basis for $L_2(\{0,1\}^R)$. In particular, any function $f \in L_2(\{0,1\}^R)$ can be written in the Fourier basis as
\[
	f = \sum_{S \subseteq [R]} \wh{f}(S)\chi_S,
\]
and $\{\wh{f}(S)\}_{S \subseteq [R]}$ are known as the Fourier coefficients of $f$.

{\bf Influence}. For any $f \in L_2(\{0,1\}^R)$, and any $i \in [R]$, we define the $i^{th}$ influence of $f$ as 
\[
	\Inf{i}{f} = \sum_{S \ni i} \wh{f}(S)^2,
\]
i.e., it is the sum of squares of the Fourier coefficients which contain $i$.

{\bf Correlated Sampling}. For any $\rho \in [-1,1]$, and $x \in \{0,1\}^R$, we can sample a $\rho$-correlated copy $y$ of $x$ using the following process. For every $i \in [R]$, sample $y(i)$ independently using the following distribution:
\[
y(i) = 
\begin{cases}
	x(i) & \text{ with probability } \frac{1 + \rho}{2},		\\
	1 - x(i) & \text{ with probability } \frac{1 - \rho}{2}.
\end{cases}
\]
We use the notation $y \sim_\rho x$ to denote that $y$ is being sampled as a $\rho$-correlated copy of $x$. 

{\bf Noise Operator}. For any $\rho \in [-1,1]$, we use $T_\rho:L_2(\{0,1\}^R) \to L_2(\{0,1\}^R)$ to denote the {\em $\rho$-correlated noise operator}, which is defined in the following way. For any $f \in L_2(\{0,1\}^R)$, we define $T_\rho f$ as
\[
	T_\rho f(x) = \Ex_{y \sim_\rho x}\big[f(y)\big].
\]
In other words, $T_\rho f$ averages $f$ around $\rho$-correlated copies of the input string. The following formula for the Fourier expansion of $T_\rho f$ is well known.

\begin{fact}[Proposition 2.47~\cite{ODonnell14}]
	For any $\rho \in [-1,1]$ and function $f:\{0,1\}^R \to \mathbb{R}$ we have
	\[
		T_\rho f = \sum_{S \subseteq [R]} \rho^{|S|}\wh{f}(S)\chi_S.
	\] 
\end{fact}

\subsection{Majority-Is-Stablest}

We will now state the Majority-Is-Stablest (MIS) theorem which is the key technical tool underlying the hardness guarantee of Theorem \ref{thm:ugc-hard}. Before stating it, we will need to introduce some additional notation. For any $\rho \in [-1,1]$, and a function $f \in L_2(\{0,1\}^R)$, let us define the {\bf $\rho$-stability} of a function $f$ as 
\[
		\Stab_{\rho}(f) := \Ex_{x \sim_\rho y}\big[f(x)f(y)\big].
\]
Let $g_1 \sim_\rho g_2$ denote a pair of $\rho$-correlated Gaussians, i.e.,
\[
(g_1,g_2) \sim \cal{N}\left(\begin{bmatrix} 0 \\ 0 \end{bmatrix}, \begin{bmatrix} 1 & \rho \\ \rho & 1 \end{bmatrix}\right).
\] 
The Gaussian quadrant function $\Gamma_\rho:[0,1] \to [0,1]$ is defined as 
\[
		\Gamma_\rho(\delta) := \Pr_{g_1 \sim_\rho g_2}\Big[g_1 \leq \Phi^{-1}(\delta), g_2 \leq \Phi^{-1}(\delta)\Big],
\]
where $\Phi^{-1}$ is the inverse of the standard Gaussian CDF function. 

\begin{theorem}[~\cite{MOO10},~\cite{ODonnell14}]		\label{thm:maj-stablest}
	There exists a constant $\tau_0 \in (0,1)$ such that the following holds. Fix any $\rho  \in (-1,0)$, let $\tau \in [0,\tau_0]$, and let $R \in \mathbb{N}$ be large enough as a function of $\tau$. Let $g:\{0,1\}^R \to [0,1]$ be a function satisfying
	\begin{equation}				\label{eqn:low-inf}
		\max_{i \in [R]} \Inf{i}{g} \leq \tau.
	\end{equation}
	Then,
	\[
	\Stab_{\rho}(g) \geq \Gamma_\rho(\Ex[g]) - O\left(\frac{\log\log(1/\tau)}{\log(1/\tau)}\right) \cdot \frac{1}{1 - \rho}.
	\]
\end{theorem}

\paragraph{Cut Version of MIS.} For the soundness analysis of the Max-Cut reduction, it would be more convenient to work with the ``cut'' version of the above theorem. Let ${\sf Cut}(x,y) := x + y - 2xy$ be the quadratic form of the cut-function, i.e., on Boolean inputs $x,y$, we have ${\sf Cut}(x,y) = \one(x \neq y)$. The following is an immediate corollary of the above theorem.

\begin{corollary}[~\cite{MOO10},~\cite{ODonnell14}]			\label{cor:maj-stablest}
	There exists a constant $\tau_0 \in (0,1)$ such that the following holds for any $\rho  \in (-1,0)$. Let $\tau \in [0,\tau_0]$, and let $R \in \mathbb{N}$ be large enough as a function of $\tau$. Suppose $g:\{0,1\}^R \to [0,1]$ is a function which satisfies
	$\max_i \Inf{i}{g} \leq \tau$.
	Then,
	\[
	\Ex_{x_1 \sim_\rho x_2}\Big[{\sf Cut}(g(x_1),g(x_2))\Big] \leq \frac{\arccos \rho}{\pi} + 		 O\left(\frac{\log\log(1/\tau)}{\log(1/\tau)}\right) \cdot \frac{1}{1 - \rho}.
	\]
    Additionally, the above bound is attained by the majority function {\sf MAJ} up to $o_{R}(1)$-error.
\end{corollary}

\subsection{A Useful $L_1$-Averaging Lemma}

We state a useful lemma that will be used several times in the proof to bound the average $L_1$-distance from Boolean-valued functions.

\begin{lemma}				\label{lem:avg}		
	Suppose $\cD$ is a discrete distribution over functions $f:\{0,1\}^R \to [0,1]$, and let $g:\{0,1\}^R \to \{0,1\}$ be a Boolean-valued function. Then,
	\[
	\left\|\Ex_{f \sim \cD} f - g \right\|_1 = \Ex_{f \sim \cD}\Big[\|f - g\|_1\Big].
	\] 
\end{lemma}
\begin{proof}
	Let us define $A_s \subseteq \{0,1\}^R$ as
	\[
	A_s := \left\{x \in \{0,1\}^R~\Big|~ g(x) = s\right\}.
	\]
	Since every function in the support of $\cD$ is $[0,1]$-valued, we can write:
	\begin{align*}
		\left\|\Ex_{f \sim \cD} f - g \right\|_1
		&= \frac{1}{2^R}\sum_{x \in A_0}\left(\Ex_{f \sim \cD}f(x) - g(x)\right) + \frac{1}{2^R}\sum_{x \in A_1}\left(g(x) - \Ex_{f \sim \cD}f(x)\right) \\
		&= \frac{1}{2^R}\sum_{x \in A_0}\Ex_{f \sim \cD} \Big[f(x) - g(x)\Big] + \frac{1}{2^R}\sum_{x \in A_1}\Ex_{f \sim \cD}\Big[g(x) - f(x)\Big] \\
		&= \frac{1}{2^R}\sum_{x \in A_0}\Ex_{f \sim \cD} \Big[|f(x) - g(x)|\Big] + \frac{1}{2^R}\sum_{x \in A_1}\Ex_{f \sim \cD} \Big[|g(x) - f(x)|\Big] \\
		&= \Ex_{f \sim \cD}\left[\frac{1}{2^R}\sum_{x \in \{0,1\}^R} \left|f(x) - g(x)\right|\right] \\
		&= \Ex_{f \sim \cD}\Big[\|f - g\|_1\Big].
	\end{align*}		
\end{proof}

We now record a useful consequence of the above lemma.

\begin{lemma}			\label{lem:ltf}
	Let $\mu \in [0,1]$. Suppose $f:\{0,1\}^R \to \{0,1\}$ is a function such that there exists an LTF $L(x) = \one\left\{\sum^R_{i = 1}w_ix_i + \theta \geq 0\right\}$ satisfying $\|T_{1 - \mu}f - L\|_1 = \epsilon$. Then there exists an LTF $L':\{0,1\}^R \to \{0,1\}$ such that $\|f - L'\|_1 \leq \epsilon$.
\end{lemma}
\begin{proof}
	For any $z \in \{0,1\}^R$, let $f_z:\{0,1\}^R \to \{0,1\}$ be the function $f$ whose inputs are shifted by $z$, i.e., $f_z(x) = f(x \oplus z)$ where $\oplus$ is the bit-wise XOR. Then, 
	\[
	T_{1 - \mu}f = \Ex_{z \sim \{0,1\}^R_{\mu/2}}~f_z,
	\]
	where $\{0,1\}^R_{\mu/2}$ is the $(\mu/2)$-biased measure on $\{0,1\}^R$. Hence, using Lemma \ref{lem:avg} we get that
	\[
	\min_{z \in \{0,1\}^R} \|f_z - L\|_1 
	\leq \Ex_{z \sim \{0,1\}^R_{\mu/2}} \Big[\|f_z - L\|_1\Big] 
	= \left\|\Ex_{z \sim \{0,1\}^R_{\mu/2}}~f_z - L \right\|_1 
	= \|T_{1 - \mu}f - L\|_1
	= \epsilon,
	\]	
	which in turn implies that there exists a $z^* \in \{0,1\}^R$ such that $\|f_{z^*} - L\|_1 \leq \epsilon$. Furthermore, note that 
	\[
	\|f_{z^*} - L\|_1 = \|f - L_{z^*}\|_1,
	\]
	and by definition,
	\[
	L_{z^*}(x) = \one\left\{\sum_{i:z_i = 0} w_ix_i + \sum_{i:z_i = 1} w_i(1 - x_i) + \theta \geq 0\right\}
	\]
	is also an LTF. The above observations together establish the lemma.
\end{proof}

The above lemma immediately yields the following corollary.

\begin{corollary}					\label{cor:noise}
	Let $\cL_R$ denote the set of LTFs over $\{0,1\}^R$. Then for any Boolean function $f:\{0,1\}^R \to \{0,1\}$ and $\mu \in [0, 1]$ we have 
	\[
	{\sf dist}(f,\cL_R) \leq {\sf dist}(T_{1 - \mu}f,\cL_R).
	\]	
\end{corollary}

\paragraph{Related Claims.} We now state a few more $L_1$-averaging-based observations that will be used in the proof.

\begin{claim}				\label{cl:invar}
	For any function $f:\{0,1\}^R \to \{0,1\}$ and permutation $\pi:[R] \to [R]$, 
	\[
	{\sf dist}(\pi \circ f,\cL_R) = {\sf dist}(f,\cL_R).
	\]
\end{claim}
\begin{proof}
	We will show ${\sf dist}(\pi \circ f, \cL_R) \leq {\sf dist}(f,\cL_R)$ and the other direction follows similarly. Let $L^* \in \cL_R$ be such that $\|f - L^*\|_1 = {\sf dist}(f,\cL_R)$. Now, note that the function $\pi\circ L^*$ is also an LTF, and hence,
	\[
	{\sf dist}(\pi \circ f, \cL_R) \leq \|\pi \circ  f - \pi \circ L^*\|_1 = \|f - L^*\|_1 = {\sf dist}(f,\cL_R).
	\]
\end{proof}

\begin{lemma}				\label{lem:avg-1}
	Let $\cD$ be a discrete distribution over functions $f:\{0,1\}^R \to [0,1]$. Then,
	\[
		\Ex_{f \sim \cD}\Big[{\sf dist}(f,\cL_R)\Big] \leq {\sf dist}\left(\Ex_{f \sim \cD}f, \cL_R\right).
	\]
\end{lemma}

\begin{proof}
	Let $L^* \in \cL_R$ be such that $\| \Ex_{f \sim \cD} f - L^*\|_1 = {\sf dist}(\Ex_{f \sim \cD}f, \cL_R)$. Then using Lemma \ref{lem:avg} we have 
	\[
	\Ex_{f \sim \cD}\left[{\sf dist}(f,\cL_R)\right]
	\leq \Ex_{f \sim \cD}\left[\|f - L^*\|_1\right]
	= \left\|\Ex_{f \sim \cD} f - L^* \right\|_1
	= {\sf dist}\left(\Ex_{f \sim \cD}f, \cL_R\right),
	\]
	where the last step follows using the definition of $L^*$.
\end{proof}

\subsection{The Unique Games Conjecture}

Finally, our proof of Theorem \ref{thm:ugc-hard} is a reduction from the Unique Games problem, which we define below.

\begin{definition}[Unique Games]
    An instance $\cG = (U,V,E,[R],\{\pi_{u,v}\}_{(u,v) \in E})$ of the Unique Games problem is characterized by a bipartite graph with left and right vertex sets $U$ and $V$, and edge set $E \subseteq U \times V$. Each edge $(u,v) \in E$ is identified with a bijection constraint $\pi_{u,v}:[R] \to [R]$. A labeling of $\cG$ is a mapping $\sigma: U \cup V \to [R]$ that assigns a value to each vertex in $U \cup V$ from the label set $[R]$. A labeling $\sigma$ is said to satisfy an edge $(u,v) \in E$ if $\sigma(u) = \pi_{u,v}(\sigma(v))$. 

    The {\bf optimal value} of the unique game $\cG$ is the maximum fraction of constraints that satisfied by any labeling of $\cG$, and we denote it by $\opt(\cG)$.
\end{definition}

We will use the following hardness of the Unique Games problem as the start point for the proof of Theorem \ref{thm:ugc-hard}.

\begin{conjecture}[\cite{Khot02},\cite{KR08}]           \label{conj:ugc}
    There exists a constant $\eps_0 \in (0,1/2]$ such that for every $\eps \in (0,\eps_0]$, there exists some integer $R(\eps)$ such that the following holds for all $R \geq R(\eps)$. Given a {\bf bi-regular} Unique Games instance $\cG = (U,V,E,[R],\{\pi_{u,v}\}_{(u,v) \in E})$, it is $\NP$-hard to distinguish between the following cases:
    \[
    {\bf YES\,\,Case}: \opt(\cG) \geq 1 - \eps
    \qquad\qquad\text{and}\qquad\qquad
    {\bf NO\,\,Case}: \opt(\cG) \leq \eps.
    \]
\end{conjecture}


\section{Max-Cut Reduction}			\label{sec:analysis}

We begin by recalling the standard reduction from Unique Games to Max-Cut in Figure \ref{fig:reduction}. 

\begin{figure}[ht!]
	\begin{mdframed}
	Let $\cG = (U,V,E,[R],\{\pi_e\}_{e \in E})$ be a bi-regular Unique Games instance from Conjecture \ref{conj:ugc} with $\eps$ chosen as below, and $R = R(\eps)$ defined as in Conjecture \ref{conj:ugc}. Then we construct a Max-Cut instance $H = (V_H,E_H,w_H)$ as follows: \\
	
	{\bf Vertex Set}. The vertex set of $H$ is $V_H = V \times \{0,1\}^R$ (where $R$ is chosen large enough). \\
	
	{\bf Edges}. The distribution over the edges is given by the following process: \\
	
	\begin{itemize}
		\item Sample $u \sim U$, and $v_1,v_2 \sim N_\cG(u)$. \\[-6pt]
		\item Sample $x_1 \sim \{0,1\}^R$ and a $\rho^*$-correlated copy $x_2 \underset{\rho^*}{\sim} x_1$ from $\{0,1\}^R$.
		\item Sample $x'_1 \underset{1 - \mu}{\sim} x_1$ and $x'_2 \underset{1 - \mu}{\sim} x_2$ independently. 
		\item Output the edge $\big\{(v_1,\pi_{u,v_1}(x'_1)), (v_2,\pi_{u,v_2}(x'_2))\big\}$.
	\end{itemize}	
	\end{mdframed}
	\caption{Unique Games to Max-Cut Reduction}
	\label{fig:reduction}
\end{figure}

{\bf Discussion of parameters}. For a fixed choice of $\eps$, we will define the following parameters as a function of $\eps$. 
\begin{itemize}
	\item $\tau(\eps) := 2\eps^{1/8}$.
	\item $\mu(\eps) := 2\eps^{1/16}$.
	\item $\nu(\eps) := \eps^{1/2}$.
	\item $c(\eps) := O\left(\frac{\log\log(1/\tau(\eps))}{\log(1/\tau(\eps))}\right)$, where the constant in $O(\cdot)$ is from Theorem \ref{thm:maj-stablest}.
	\item $c_0(\eps) := O(c(\eps) + \nu(\eps))$. 
\end{itemize}

Note that all the functions defined above are continuous functions that (a) are increasing in $\eps$ and (b) evaluate to $0$ at $\eps = 0$. At various points of the proof, we will introduce additional error terms which will be functions of $\eps$; these error terms will also have the property that they are increasing continuous functions of $\eps$ that decrease to $0$ as $\eps \to 0$.

Finally, we will choose $\eps$ to be small enough as a function of $d$, the explicit dependence on $d$ will be based on the proof of Theorem \ref{thm:ugc-hard} in Section \ref{sec:main-proof}. 

{\bf Analysis of the reduction}. The following lemmas from \cite{KKMO07} give the completeness and soundness guarantees of the above reduction. 

\begin{lemma}[Completeness~\cite{KKMO07}]		\label{lem:comp}
	If $\opt(\cG) \geq 1 - \epsilon$, then $H$ has a cut of size at least $(c^* - 2(\epsilon + \mu))$.
\end{lemma}

\begin{lemma}[Soundness~\cite{KKMO07}]		\label{lem:sound}
	If $\opt(\cG) \leq \epsilon$, then the max-cut of $H$ is at most $(s^* + c_0(\epsilon))$, where $c_0(\epsilon) \downarrow 0$ as $\epsilon \downarrow 0$.
\end{lemma}

For the sake of making the presentation self-contained, we will analyze the above reduction and prove the above lemmas in full detail in Appendix~\ref{sec:KKMO}. In addition to the above, we will need the following additional structural lemma which characterizes the near-optimal cuts of $H$ in the NO case. In order to state the lemma, we need to introduce an additional notation. For any $f: V_H \to \{0, 1\}$ and vertex $v \in V$, let $f_v:\{0,1\}^R \to \{0,1\}$ be the local table defined as $f_v(x) = f(v,x)$. 

\begin{lemma}[Robust Inverse of Soundness]		\label{lem:near-opt}
	There exist some constant $\eta_0 \in (0, 1)$ and $C > 0$ such that the following holds. Let $\eta \in [0,\eta_0]$ and $\eps \in [0, 1]$. Suppose $\opt(\cG) \leq \epsilon$ and let $H = (V_H,E_H,w_H)$ be the corresponding Max-Cut instance output by the reduction. Let $f:V_H \to \{0,1\}$ be such that $S_f := {\rm supp}(f)$ cuts at least $(s^*-\eta)$-fraction of edges in $H$. Then we have 
	\[
	\Ex_{v \sim V}\Big[{\sf dist}(f_v,\cL_R)\Big] \leq C \cdot \sqrt{\eta} + c(\eps),
	\]
	where $c(\eps)$ is a continuous function satisfying $c(\eps)\downarrow 0$ as $\eps \downarrow 0$.
\end{lemma}

In other words, the above lemma states the following: the soundness analysis (Lemma \ref{lem:sound}) of the reduction shows that in the NO case, the max-cut value of the graph $H$ can be at most $(s^* + o_\eps(1))$, and such a cut is achieved (up to negligible $o_\eps(1)$ error) by (all symmetries of) the assignment $f_v = {\sf MAJ}$ for all $v \in V$. The above lemma gives a robust inverse of this statement: if a cut $f$ cuts at least $(s^* - \eta)$ edges in $H$, then $f_v$ must be close to an LTF for most choices of $v \in V$. The above lemma is one of the key technical tools that we use to improve the union bound analysis in the proof of Theorem \ref{thm:ugc-hard}. The rest of this section is devoted to proving the above lemma. 

\subsection{Useful Lemmas from KKMO Soundness Analysis}

We now record some useful properties that are extracted from the soundness analysis of the above reduction in \cite{KKMO07}. Throughout this section, we will assume that $\cG$ is a Unique Games instance with $\opt(\cG) \leq \epsilon$. Let $H = (V_H,E_H,w_H)$ be the Max-Cut instance constructed by instantiating the above reduction with $\cG$. Let us fix any cut $f:V_H \to \{0,1\}$ in $H$. Recall that
\[
{\sf Cut}(x,y) = x + y - 2xy.
\]
denotes the cut function. Then, using the definition of graph $H$ from Figure \ref{fig:reduction}, we can express the weight of edges cut by $f$ as:
\begin{equation}			\label{eqn:e0}
	\Ex_{u \sim U}\Ex_{v_1,v_2 \sim N_{\cG}(u)}\Ex_{x_1 \underset{\rho^*}{\sim} x_2}\Ex_{x'_i \underset{1 - \mu}{\sim} x_i}\Big[{\sf Cut}\Big(f_{v_1}(\pi_{u,v_1}(x'_1)),f_{v_2}(\pi_{u,v_2}(x'_2))\Big)\Big].
\end{equation}
Using the bi-linearity of the ${\sf Cut}(\cdot,\cdot)$ function, we rewrite the above as 
\begin{equation}			\label{eqn:e1}
	\Ex_{u \sim U}\Ex_{v_1,v_2 \sim N_{\cG}(u)}\Ex_{x_1 \underset{\rho^*}{\sim} x_2}\Big[{\sf Cut}\Big(T_{1 - \mu}f_{v_1}(\pi_{u,v_1}(x_1)),T_{1 - \mu}f_{v_2}(\pi_{u,v_2}(x_2))\Big)\Big].
\end{equation}

{\bf Averaged Long Codes}. As is standard, for every left vertex $u \in U$, we define the averaged long-code table $g_u: \{0,1\}^R \to [0,1]$ as 
\[
g_{u}(x) :=  \Ex_{v \sim N_{\cG}(u)}\left[T_{1 - \mu}f_v(\pi_{u,v}(x))\right].
\]
Using the bi-linearity of ${\sf Cut}(\cdot,\cdot)$ once more, we rewrite \eqref{eqn:e1} as:
\[
\eqref{eqn:e1} = \Ex_{u \sim U}\Ex_{x_1 \underset{\rho^*}{\sim} x_2}\left[{\sf Cut}(g_u(x_1),g_u(x_2))\right].
\]
It will be convenient to define the fraction of edges cut conditioned on the choice of $u \in U$ as
\[
{\sf Cut}_f(u) := \Ex_{x_1 \underset{\rho^*}{\sim} x_2}\left[{\sf Cut}(g_u(x_1),g_u(x_2))\right],
\]
which gives $\eqref{eqn:e1} = \Ex_{u \sim U}\left[{\sf Cut}_f(u)\right]$.

{\bf Soundness Analysis Lemmas}. Let $\tau(\epsilon)$ be as defined below Figure \ref{fig:reduction}. Let us define
\[
U_{{\rm good}} := \left\{u \in U~\Big|~ \max_{i \in [R]}\Inf{i}{g_u} \leq \tau(\epsilon)\right\}. 
\]
We now state two lemmas from the \cite{KKMO07} soundness analysis.

\begin{lemma}			\label{lem:u-good}
	Suppose $\opt(\cG) \leq \epsilon$. Then, $|U_{\rm good}| \geq (1 - \nu(\epsilon))|U|$, where $\nu(\epsilon) = \sqrt{\eps}$.
\end{lemma}

\begin{lemma}			\label{lem:good-cut}
	Let $R$ be chosen large enough as a function of $\epsilon$. Then for every $u \in U_{\rm good}$ we have 
	\[
	{\sf Cut}_f(u) \leq s^* + c(\epsilon),
	\]
	where $c(\epsilon) = O\left(\frac{\log\log(1/\eps)}{\log(1/\eps)}\right)$.
\end{lemma}

We include the proofs of the above lemma as part of the soundness analysis of the reduction in Appendix \ref{sec:KKMO}.

\subsection{Properties of Near-Optimal Cuts}

Now let $f:V_H \to \{0,1\}$ be the indicator of a set that cuts at least $(s^* - \eta)$-weight of edges in $H$, i.e., it satisfies $\Ex_{u \sim U}\big[{\sf Cut}_f(u)\big] \geq s^* - \eta$. Based on this assumption, we will now derive some useful claims that allow us to identify a large subset of ``ideal'' vertices within $U_{\rm good}$ on which we will carry out the rest of the analysis.

\begin{claim}		\label{cl:c-1}
	\[
		\Ex_{u \sim U_{\rm good}}\left[{\sf Cut}_f(u)\right] \geq s^* - \eta - \nu(\epsilon).
	\]
\end{claim}
\begin{proof}
	This follows immediately by observing that ${\sf Cut}_f(u)$ is $[0,1]$-valued, and $|U_{\rm good}| \geq (1 - \nu(\epsilon))|U|$ (Lemma \ref{lem:u-good}). 
\end{proof}

Next, let $\eta_0$ be the constant from Lemme~\ref{lem:near-opt}. Let $U'_{\rm good} \subseteq U_{\rm good}$ be defined as 
\[
	U'_{\rm good} := \left\{u \in U_{\rm good}~|~ {\sf Cut}_f(u) \geq s^* - \eta_0\right\}.
\] 

\begin{claim}		\label{cl:c-2}
	\[
		|U'_{\rm good}| \geq  \left(1 - \frac{2(c'(\eps) + \eta)}{\eta_0}\right)\cdot|U_{\rm good}|,
	\]
	where $c'(\eps) = c(\eps) + \nu(\eps)$ satisfies $c'(\eps) \downarrow 0$ as $\eps \downarrow 0$.
\end{claim}
\begin{proof}
	Let $p := \Pr_{u \sim U_{\rm good}}[{\sf Cut}_f(u) \geq s^* - \eta_0]$. Then,
	\begin{align*}
	s^* - \eta - \nu(\eps) \leq \Ex_{u \sim U_{\rm good}}[{\sf Cut}_f(u)] 
	&= (1 - p) \cdot \Ex_{u \notin U'_{\rm good}}[{\sf Cut}_f(u)] + p \cdot \Ex_{u \in U'_{\rm good}}[{\sf Cut}_f(u)] \\
	&\leq (1 - p)(s^* - \eta_0) + p(s^* + c(\epsilon)), 
	\end{align*}
	where in the last step, the bound on the first expectation term is by the definition of $U'_{\rm good}$ and the bound on the second term is due to Lemma \ref{lem:good-cut}. Rearranging the above gives the bound 
	\[
	p \geq \frac{\eta_0 - \eta - \nu(\eps)}{\eta_0 + c(\epsilon)} \geq 1 - \frac{2(\eta + c(\epsilon) +  \nu(\eps))}{\eta_0},
	\]
	where the last inequality holds because
    \[
    \frac{\eta_0 - \eta - \nu(\eps)}{\eta_0 + c(\epsilon)} - \left( 1 - \frac{2(\eta + c(\epsilon) +  \nu(\eps))}{\eta_0}\right) = \frac{(\eta + c(\epsilon) + \nu(\epsilon))(\eta_0 + 2c(\epsilon))}{\eta_0(\eta_0 + c(\epsilon))} > 0.
    \]
    The proof is concluded by setting $c'(\eps) := c(\eps) + \nu(\eps)$. 
\end{proof}

\begin{claim}		\label{cl:c-3}
	\[
		\Ex_{u \sim U'_{\rm good}}[{\sf Cut}_f(u)] \geq s^* - \frac{3}{\eta_0}(\eta + c''(\eps)),
	\]
	where $c''(\eps)$ is a continuous function satisfying $c''(\eps) \downarrow 0$ as $\eps \downarrow 0$. 
\end{claim}
\begin{proof}
	Note that since $|U'_{\rm good}| \geq (1 - \frac{2}{\eta_0}(\eta + c'(\eps)))|U_{\rm good}|$ we have 
	\begin{align*}
		\Ex_{u \sim U'_{\rm good}}[{\sf Cut}_f(u)]
		&\geq \Ex_{u \sim U_{\rm good}}[{\sf Cut}_f(u)] - \frac{2}{\eta_0}(\eta + c'(\eps)) \\
		&\geq s^* - \eta - \nu(\eps) - \frac{2}{\eta_0}(\eta + c'(\eps))  \\
		&\geq s^* -\frac{3}{\eta_0}(\eta + c''(\eps)),
	\end{align*}
	where in the last step we have denoted $c''(\eps) = c'(\eps) + \nu(\eps)$.
\end{proof}

\begin{remark}
	Since $\eta_0$ is an absolute constant, we will simply treat the expressions $\frac{2}{\eta_0}(\eta + c'(\eps))$ and $\frac{3}{\eta_0}(\eta + c''(\eps))$ from Claims \ref{cl:c-2} and \ref{cl:c-3} as $O(\eta + c'(\eps))$ and $O(\eta + c''(\eps))$ respectively for the rest of the analysis.
\end{remark}	
	
\subsection{Proof of Lemma \ref{lem:near-opt}}

We will prove Lemma \ref{lem:near-opt} in this section. The key technical ingredient used in the proof is the following asymptotically optimal robust inverse to the Majority-Is-Stablest theorem.

\begin{theorem}				\label{thm:struct}
	There exists some constant $\eta_0 \in (0, 1)$ such that the following holds. Let $\eta \in (0,\eta_0)$ and $\tau \in (0,1)$. Let $g:\{0,1\}^R \to [0,1]$ be a function satisfying
	$\max_{i}\Inf{i}{g} \leq \tau$ and
	\[
	\Ex_{x \underset{\rho^*}{\sim} y}\Big[{\sf Cut}(g(x_1),g(x_2))\Big] \geq s^* - \eta.
	\]
	Then the function $g$ satisfies 
	\[
	{\sf dist}(g,\cL_R) \leq O(\sqrt{\eta}) + c_{1,1}(\tau),
	\]
	where $c_{1,1}$ is a continuous increasing function satisfying $c_{1,1}(0) = 0$.
\end{theorem}

We refer the reader to Section \ref{sec:robust-mist} for a discussion on Theorem \ref{thm:struct} and its proof. Equipped with the above theorem, we are now ready to prove Lemma \ref{lem:near-opt}. Recall that we use $\cL_R$ to denote the set of LTFs over $\{0,1\}^R$. Our proof of the lemma proceeds in two steps which are briefly outlined below:

\begin{itemize}
    \item {\bf Step 1}. In Lemma \ref{lem:bound-1}, we first bound the average distance of the right table functions to linear thresholds by the average $L_1$-distance of the left table functions to linear thresholds i.e.,
    \[
        \Ex_{v \sim V}\Big[{\sf dist}(f_v,\cL_R)\Big] \leq \Ex_{u \sim U}\Big[{\sf dist}(g_u,\cL_R)\big].
    \]
    \item {\bf Step 2}. In Lemma \ref{lem:bound-2}, we use the robust inverse of Majority-Is-Stablest (Theorem \ref{thm:struct}) to bound the RHS of the above equation with $O(\sqrt{\eta}) + o_\eps(1)$.
\end{itemize}

Putting the two bounds together then immediately yields the bound claimed in the lemma. We now formally implement the two steps described above. 

\paragraph{Average Closeness to LTFs.} We can bound the average $L_1$-distance of the right vertex tables to the set of LTFs $\cL_R$ as follows.

\begin{lemma}				\label{lem:bound-1}
	\[
	\Ex_{v \sim V}\Big[{\sf dist}(f_v,\cL_R)\Big] \leq \Ex_{u \sim U}\Big[{\sf dist}(g_u,\cL_R)\Big].
	\]
\end{lemma}
\begin{proof}
\begin{align*}
	\Ex_{v \sim V}\Big[{\sf dist}(f_v,\cL_R)\Big]
	&= \Ex_{u \sim U} \Ex_{v \sim N_{\cG}(u)}\left[{\sf dist}(f_v,\cL_R\Big)\right] &		\text{(Bi-regularity of $\cG$)} \\
	&= \Ex_{u \sim U} \Ex_{v \sim N_{\cG}(u)}\left[{\sf dist}\Big(\pi_{u,v} \circ f_v,\cL_R\Big)\right] &	\text{(Claim \ref{cl:invar})} \\
	&\leq \Ex_{u \sim U} \Ex_{v \sim N_{\cG}(u)}\left[{\sf dist}\Big(T_{1 - \mu}(\pi_{u,v} \circ f_v),\cL_R\Big)\right]		&\text{(Corollary \ref{cor:noise})} \\
	&\leq \Ex_{u \sim U} \left[{\sf dist}\Big(\Ex_{v \sim N_G(u)}T_{1 - \mu}(\pi_{u,v} \circ f_v), \cL_R\Big)\right]		&\text{(Lemma \ref{lem:avg-1})} \\
	&= \Ex_{u \sim U}\Big[{\sf dist}(g_u,\cL_R)\Big].		&	\text{(Definition of $g_u$)}
\end{align*}
\end{proof}

\paragraph{Bounding the RHS.} Now we proceed to upper bound the RHS of the above.

\begin{lemma}				\label{lem:bound-2}
	\[
	\Ex_{u \sim U}\Big[{\sf dist}(g_u,\cL_R)\Big] \leq O(\eta^{1/2}) + c_2(\epsilon),
	\]
	where $c_2(\epsilon)$ is a continuous function satisfying $c_2(\eps) \downarrow 0$ as $\eps \downarrow 0$.
\end{lemma}
\begin{proof}
	We proceed with the proof in steps. First observe that using Lemma \ref{lem:u-good} we have
	\begin{equation}			\label{eqn:step-1}
	\Ex_{u \sim U}\Big[{\sf dist}(g_u,\cL_R)\Big] \leq \Ex_{u \sim U_{\rm good}}\Big[{\sf dist}(g_u,\cL_R)\Big] + \nu(\epsilon).
	\end{equation}
	Next, using the bound from Claim \ref{cl:c-2} we have 
	\begin{equation}			\label{eqn:step-2}
	\Ex_{u \sim U_{\rm good}}\Big[{\sf dist}(g_u,\cL_R)\Big] \leq \Ex_{u \sim U'_{\rm good}}\Big[{\sf dist}(g_u,\cL_R)\Big] + O(\eta + c'(\epsilon)).	
	\end{equation}
	Furthermore, for every $u \in U'_{\rm good}$, let us define the variable 
	\[
	Z_u := \max\left\{s^* - {\sf Cut}_f(u),0\right\}.
	\] 
	Now, by definition we know that for every $u \in U'_{\rm good}$ we have that 
	\[
	\max_{i \in [R]} \Inf{i}{g_u} \leq \tau
	\qquad\quad\textnormal{and}\qquad\quad
	\Ex_{x_1 {\sim_\rho^*} x_2} \Big[{\sf Cut}(g_u(x_1),g_u(x_2))\Big]
	\geq s^* - \eta_0,
	\]
    where $\tau = \tau(\eps)$ is as defined under Figure \ref{fig:reduction}. In other words, the function $g_u$ satisfies the premise of Theorem \ref{thm:struct} for every $u \in U'_{\rm good}$. Hence, using Theorem \ref{thm:struct} we get that:
	\begin{align}
	\Ex_{u \sim U'_{\rm good}}\Big[{\sf dist}(g_u,\cL_R)\Big]	
	&\leq \Ex_{u \sim U'_{\rm good}}\Big[(O(Z_u))^{1/2} + c_{1,1}(\tau)\Big] 		\nonumber\\
	&\leq O\left(\Ex_{u \sim U'_{\rm good}}\big[Z_u\big]\right)^{1/2} + c_{1,1}(\tau) 		\nonumber\\
	&\leq O\Big(\eta + c''(\epsilon)\Big)^{1/2} + c_{1,1}(\tau)			\nonumber\\
	&\leq O(\eta^{1/2}) + O(c''(\epsilon)^{1/2}) + c_{1,1}(\tau), 
	\label{eqn:step-3}
	\end{align}
	where the second inequality is due to Jensen's inequality, the third inequality is the following application of Claim \ref{cl:c-3}:
	\[
	\Ex_{u \sim U'_{\rm good}}[Z_u] \leq \Ex_{u \sim U'_{\rm good}}[s^* - {\sf Cut}_f(u)] \leq O(\eta + c''(\epsilon)),
	\]
	and the last inequality uses the observation that $(a + b)^{1/2} \leq a^{1/2} + b^{1/2} $ for any $a,b \geq 0$. Putting together the bounds from \eqref{eqn:step-1}, \eqref{eqn:step-2}, and \eqref{eqn:step-3}, we get that 
	\begin{align*}
	\Ex_{u \sim U}\left[{\sf dist}(g_u,\cL_R)\right] 
	&\leq O(\eta^{1/2}) + O(c''(\epsilon)^{1/2}) + c_{1,1}(\tau(\epsilon)) + \nu(\epsilon) \\
	&= O(\eta^{1/2}) + c_2(\epsilon),
	\end{align*}
	where we have denoted $c_2(\epsilon) = O(c''(\epsilon)^{1/2}) + c_{1,1}(\tau(\epsilon)) + \nu(\epsilon)$.
\end{proof}
\paragraph{Finishing Up.} 

By combining the bounds from Lemma \ref{lem:bound-1} and Lemma \ref{lem:bound-2} we can bound:
\begin{equation}\label{eq:f_v_dist}
	\Ex_{v \sim V}\left[{\sf dist}(f_v,\cL_R)\right] 
	\leq \Ex_{u \sim U}\left[{\sf dist}(g_u,\cL_R)\right]
	\leq O(\sqrt{\eta}) + c_2(\epsilon)
\end{equation}
which concludes the proof of Lemma \ref{lem:near-opt}.

\section{Proof of Theorem~\ref{thm:ugc-hard}}		\label{sec:main-proof}

Let $\cG = (U,V,E,[R],\{\pi_e\}_{e \in E})$ be a Unique Games instance, and let $H = (V_H,E_H,w_H)$ be the Max-Cut instance output the reduction in Figure \ref{fig:reduction}. Let us denote $n = |V_{H}| = |V|\cdot2^R$. Following \cite{Trevisan01}, we now construct an unweighted simple graph $\tilde{H} = (\tilde{V}_H,\tilde{E}_H)$ with maximum degree at most $10d$ as follows:

\begin{itemize}
	\item[(R1)] We first construct an intermediate multi-graph $H' = (V_H,E'_H)$ by independently sampling $nd$ edges $e_1, \ldots, e_{nd}$ from $H$ using the probability distribution induced by the edge weights $w_H$. Note that this multi-graph is unweighted.
	\item[(R2)] Then, we construct $\tilde{H} = (\tilde{V}_H,\tilde{E}_H)$ from $H'$ by first replacing each set of parallel edges by a single edge, and then arbitrarily deleting  additional edges, if necessary, so that every vertex has degree at most $10d$.
\end{itemize} 

Theorem~\ref{thm:ugc-hard} follows directly from the following theorem.

\begin{theorem}\label{thm:h_tilde}
    Fix some integer $d > 0$. Let $\epsilon = \epsilon(d) > 0$ be sufficiently small, and $R = R(\epsilon(d), d)$ be sufficiently large. The graph $\tilde{H}$ satisfies the following when $n$ is sufficiently large:
    \begin{itemize}
        \item Completeness. If $\Opt(\cG) \geq 1 - \epsilon$, then with probability at least $2/3$ we have \[\MaxCut(\tilde{H}) \geq c^\ast - O(1/d).\]
        \item Soundness. If $\Opt(\cG) \leq \epsilon$, then with probability at least $2/3$ we have \[\MaxCut (\tilde{H}) \leq s^\ast + O(\log^2(d) / d).\]
    \end{itemize}
\end{theorem}

    

\subsection{Analysis of $H'$}			\label{sec:H'-analysis}

We start our analysis by analyzing the intermediate (multi-)graph $H'$. We prove the following lemma.

\begin{lemma}\label{lem:H'}
    Fix some integer $d > 0$ and let $\epsilon > 0$ be sufficiently small depending on $d$. The (multi-)graph $H'$ satisfies the following when $n$ is sufficiently large:
    \begin{itemize}
        \item Completeness. If $\Opt(\cG) \geq 1 - \epsilon$, then ${\sf MaxCut} (H') \geq c^\ast - 1/d$ with probability at least $.99$.
        \item Soundness. If $\Opt(\cG) \leq \epsilon$, then ${\sf MaxCut} (H') \leq s^\ast + \log^2(d)/d$ with probability at least $.99$.
    \end{itemize}
\end{lemma}

Given a cut $f: V_H \to \{0, 1\}$, we use $E(f) \subseteq E_H$ to denote its cut edges in H. We use $W_H(f) = \sum_{e \in E(f)}w_H(e)$ to denote its cut value in $H$ and $W_{H'}(f) = \frac{1}{nd}\sum_{i \in [nd]} \one\{e_i \in E(f)\}$ to denote its cut value in $H'$. Let $Y(f) = W_{H'}(f) - W_{H}(f)$.

The following version of Bernstein's lemma is well-known.
\begin{lemma}[See e.g. Theorem 2.9.5 in~\cite{vershynin2025high}]\label{lem:Bernstein}
    Let $X_1, \ldots, X_k$ be independent zero-mean random variables. Assume that $\Pr[|X_i| \leq M] = 1$ for every $i \in [k]$, then for every $s > 0$ we have 
	\[
	    \Pr\left[\sum_{i = 1}^k X_i \geq s \right] \leq \exp\left(-\frac{s^2}{2(\sum_{i = 1}^k\E[X_i^2] + Ms / 3)}\right).
	\]
\end{lemma}

\paragraph{Completeness} Assume that $\Opt(\cG) \geq 1- \epsilon$, then by Lemma~\ref{lem:comp}, there exists a cut $f: V_H \to \{0, 1\}$ such that $W_H(f)\geq c^* - o_\epsilon(1)$. For every $i \in [nd]$, let $X_i = \one\{e_i \in E_f\}$, then $\E[X_i] = W_H(f)$ and $W_{H'}(f) = \frac{1}{nd}\sum_{i \in [nd]}X_i$. It follows that $W_H(f) - W_{H'}(f) = \frac{1}{nd}\sum_{i \in [nd]}(\E[X_i] - X_i)$ is a sum of bounded mean-zero independent random variables, and therefore by Lemma~\ref{lem:Bernstein} we have
\begin{align*}
    \Pr[W_{H'}(f) \geq c^\ast - o_\epsilon(1) - 1/d] & \geq \Pr[W_{H}(f) - W_{H'}(f) \leq 1/d] \\
    & \geq 1 - \exp(-\Omega(nd(1/d)^2)) \\
    & \geq .99
\end{align*}
when $n$ is sufficiently large.

\paragraph{Soundness} For the soundness analysis, we would like to compare how the randomized reduction affects two cuts that are similar to each other. For this we need the following mild regularity property on the edge weights $w_H$.
\begin{lemma}\label{lem:prob_in_S}
	For every set $S \subseteq V_H$, 
	\begin{equation*}
	    \Pr_{e \sim w_H}[e \text{ has at least one endpoint in } S] \leq \frac{2|S|}{n}.
	\end{equation*}
\end{lemma}
\begin{proof}
    The marginal distribution on either endpoint of $e$ is uniform, so the lemma follows by applying the union bound.
\end{proof}

\begin{lemma}\label{lem:A_B_diff}
    Let $f, g: V_H \to \{0, 1\}$ be two cuts in $H$ such that $\|f - g\|_1 \leq \delta$. For every $t > 0$, we have 
	\begin{equation*}
	    \Pr[Y(f) - Y(g) \geq t] \leq \exp\left(-\frac{ndt^2}{4(\delta + t)}\right).
	\end{equation*}
\end{lemma}

\begin{proof}
    For every $i \in [nd]$, let $X_i = \one\{e_i \in E(f)\} - \one\{e_i \in E(g)\}$, then $X_i$ is a $\{-1, 0, 1\}$-valued random variable, and we have
	\begin{equation*}
	    \E[X_i^2] = \Pr[X_i \neq 0] \leq 2\delta
	\end{equation*}
	where the inequality is from Lemma~\ref{lem:prob_in_S} and the assumption that $|\supp(f) \triangle \supp(g)| \leq \delta n$. Let $Y_i = X_i - \E[X_i]$, then $|Y_i| \leq 2$ and $\Var[Y_i] \leq \E[X_i^2] \leq 2\delta$. Note that $Y(f) - Y(g) = \frac{1}{nd}\sum_{i \in [nd]} Y_i$ and therefore by Lemma~\ref{lem:Bernstein} we have
	\begin{align*}
	    \Pr[Y(f) - Y(g) \geq t] = \Pr\left[\sum_{i \in [nd]}Y_i\geq ndt\right] &\leq \exp\left(-\frac{(ndt)^2}{2(2nd\delta + 2ndt/3)}\right)\\
		& = \exp\left(-\frac{ndt^2}{4\delta + 4t/3}\right)\\
		& \leq \exp\left(-\frac{ndt^2}{4(\delta + t)}\right).
	\end{align*}
\end{proof}

\begin{lemma}\label{lem:sound_H'_near_opt}
    Let $\epsilon = \epsilon(d) > 0$  be sufficiently small, and assume $\opt(\cG) \leq \epsilon$. Let $\eta_0$ be as in Lemma~\ref{lem:near-opt}. For any $\eta > 0$ let \[\cA_\eta = \{f: V_H \to \{0, 1\} \mid W_H(f) \geq s^\ast - \eta\}.\] We have
    \begin{equation*}
        \Pr\left[\exists f \in \cA_{\eta_0 / 2}: W_{H'}(f) \geq s^\ast + \frac{2\log^2 d}{d}\right] \leq 2^{O(R^2|V|)}\cdot C_d \cdot \exp(-c_d \cdot n)
    \end{equation*}
    where $C_d, c_d > 0$ are some constants depending only on $d$.
\end{lemma}

\begin{proof}
    Let $\eta^{(0)} = \log^2(d) / d$ and for every $i \geq 0$ let $\eta^{(i+1)} = 2\eta^{(i)}$. Let $k$ be the maximum integer such that $\eta^{(k)} \leq \eta_0$.  Define $\cA^{(0)} = \varnothing$ and $\cA^{(j)} = \cA_{\eta^{(j)}}$ for every $j \in [k]$. We have
	\[
	    \cA_{\eta_0 / 2} \subseteq \cA^{(k)} = \left(\bigcup_{j = 1}^k \cA^{(j)} \setminus \cA^{(j-1)}\right).
	\]

    Consider the set of functions $\cL = \{L: V_H \to \{0, 1\} \mid \forall v \in V, L_v \text{ is an LTF}\}$. By Lemma~\ref{lem:near-opt}, if $\eta < \eta_0$, then for any cut $f \in \cA_\eta$, we have\footnote{Here we naturally extend the definition of $\dist$ to the space of functions $f:V_H \to \{0,1\}$ by defining $\dist(f,\cL) = \Ex_{v \sim V}[\dist(f_v,\cL_R)]$.} $\dist(f, \cL) \leq C\sqrt{\eta} + c(\epsilon)$ for some constant $C > 0$ and function $c(\eps)\downarrow 0$ as $\eps \downarrow 0$. By setting $\epsilon$ to be sufficiently small, for every $\eta \geq \eta^{(0)}$ and $f \in \cA_\eta$, we may write $\dist(f, \cL) \leq C\sqrt{\eta}$ by adjusting the value of $C$.
    For any $L \in \cL$ and any $j \in \{1, \ldots, k\}$ we define 
    \[N^{(j)}(L) = \left\{f: V_H \to \{0, 1\} \mid \|f - L\|_1 \leq C\sqrt{\eta^{(j)}}\right\}.
    \]
    The above discussion implies that $\cA^{(j)} \subseteq \cup_{L \in \cL} N^{(j)}(L)$.


	Define $\cB^{(j)}(L) = (\cA^{(j)} \setminus \cA^{(j-1)}) \cap N^{(j)}(L)$. We have therefore obtained the following decomposition:
	\begin{equation*}
	    \cA_{\eta_0 / 2} \subseteq \bigcup_{L \in \cL}\bigcup_{j \in [k]}\cB^{(j)}(L).
	\end{equation*}

	By Theorem \ref{thm:ltf-count}, we have $|\cL| = 2^{O(R^2|V|)}$ and therefore by the union bound, to establish our theorem it is sufficient to show that for every $L \in \cL$,
	\begin{equation*}
	\Pr\left[\exists f \in \bigcup_{j \in [k]}\cB^{(j)}(L):W_{H'}(f) \geq s^\ast + \frac{2\log^2 d}{d}\right] \leq  C_d \cdot \exp(-c_d \cdot n)
	\end{equation*}
	for some constants $C_d, c_d > 0$ depending only on $d$.
	
	Now fix some $L \in \cL$. Recall that for a cut $f$ we defined $Y_f = W_{H'}(f) - W_H(f)$. By Lemma~\ref{lem:A_B_diff}, for every $j \in [k]$ and $f \in \cB^{(j)}(L)$ we have
	\begin{align*}
	    \Pr[Y_{f} - Y_{L} \geq \eta^{(j-1)}] \leq \exp\left(-\frac{nd(\eta^{(j-1)})^2}{4\left(C\sqrt{\eta^{(j-1)}} + \eta^{(j-1)}\right)}\right) = \exp\left(-\Omega\left(d \cdot (\eta^{(j-1)})^{3/2}\right)n\right).
	\end{align*}
	By the union bound, we have
    \begin{align*}
        \Pr[\exists f \in \cB^{(j)}(L): Y_{f} - Y_{L} \geq \eta^{(j-1)} ] & \leq | \cB^{(j)}(L)|\cdot\exp\left(-\Omega\left(d \cdot (\eta^{(j-1)})^{3/2}\right)n\right)\\
		& \leq |N^{(j)}(L)|\cdot\exp\left(-\Omega\left(d \cdot (\eta^{(j-1)})^{3/2}\right)n\right)\\
		& \leq \exp\left(O\left(\sqrt{\eta^{(j)}}\log(1/ \eta^{(j)})\right)n\right)\exp\left(-\Omega\left(d(\eta^{(j-1)})^{3/2}\right)n\right)
    \end{align*}
	where we used Fact~\ref{fact:bin} in the last step. Since
	\[
	\frac{d(\eta^{(j-1)})^{3/2}}{\sqrt{\eta^{(j)}}\log(1/ \eta^{(j)})} \geq \frac{d \eta^{(0)}}{\sqrt{2}\log(1 / \eta^{(0)})}= \Omega\left(\log d \right),
	\]
	we have
	\[
	\Pr[\exists f \in \cB^{(j)}(L): Y_{f} - Y_{L} \geq \eta^{(j-1)}] \leq \exp\left(-\Omega\left(d(\eta^{(j-1)})^{3/2}\right)n\right).
	\]
	By the Chernoff bound, we have
	\begin{equation*}
	    \Pr[Y_L \geq \log^2(d)/d] \leq \exp(-\Omega(nd(\log^2(d)/d)^2)) = \exp(-\Omega(n\log^4(d)/d)).
	\end{equation*}
	Note that for any $j \in [k] \setminus \{1\}$ and $f \in \cB^{(j)}(L)$, 
	\[
	 W_{H'}(f) \geq s^\ast + \frac{2\log^2 d}{d} \Longrightarrow Y_f \geq \frac{2\log^2(d)}{d} + \eta^{(j-1)},
	\]
	and for any $f \in \cB^{(1)}(L)$, 
	\[
	 W_{H'}(f) \geq s^\ast + \frac{2\log^2 d}{d} \Longrightarrow Y_f \geq \frac{3\log^2(d)}{2d} = \frac{\log^2(d)}{2d} + \eta^{(0)}.
	\]
	(Here we set $\epsilon$ to be sufficiently small so that by Lemma~\ref{lem:sound} we have $W_{H}(f) \leq s^* + \frac{\log^2(d)}{2d}$ for every $f$.) So we may conclude using the union bound that 
	\begin{align*}
	    &\Pr\left[\exists f \in \bigcup_{j \in [k]}\cB^{(j)}(L):W_{H'}(f) \geq s^\ast + \frac{2\log^2 d}{d}\right] \\
		\leq\, & \Pr[Y_L \geq \log^2(d)/d] + \sum_{j = 1}^k \Pr[\exists f \in \cB^{(j)}(L): Y_{f} - Y_{L} \geq \eta^{(j-1)}] \\
		\leq\, & C_d \cdot \exp(- c_d \cdot n)
	\end{align*}
	where we used the fact that $k = O(\log d)$.
\end{proof}

We are now ready to prove the main soundness lemma for $H'$.

\begin{lemma}			\label{lem:H'-sound}
    Let $d$ be sufficiently large and $\epsilon = \epsilon(d) > 0$  be sufficiently small, and assume $\opt(\cG) \leq \epsilon$. Then when $n$ is sufficiently large, with probability at least $.99$ we have
    \begin{equation*}
        W_{H'}(f) \leq s^* + \frac{2\log^2(d)}{d}
    \end{equation*}
    for every cut $f: V_H \to \{0,1\}$.
\end{lemma}
\begin{proof}
     Let $\eta_0$ be as in Lemma~\ref{lem:near-opt}. By Lemma~\ref{lem:sound_H'_near_opt} we have
    \begin{equation}\label{eq:close}
        \Pr\left[\exists f \in \cA_{\eta_0 / 2}: W_{H'}(f) \geq s^\ast + \frac{2\log^2 d}{d}\right] \leq 2^{O(R^2|V|)}\cdot C_d \cdot \exp(-c_d \cdot n)
    \end{equation}
    for some constants $C_d, c_d > 0$ depending only on $d$. Since $n = |V| \cdot 2^{R} \leq 2^{R|V|}$, the term $2^{O(R^2|V|)} \leq n^{O(R)}$ is at most polynomial in $n$, so the right hand side of Equation~\eqref{eq:close} tends to 0 as $n \to \infty$. Now for every cut $f \not\in \cA_{\eta_0 / 2}$, we have $W_{H}(f) \leq s^* - \eta_0 / 2$ and therefore by the Chernoff bound we have
    \[
    \Pr\left[W_{H'}(f) \geq s^\ast + \frac{2\log^2 d}{d}\right] \leq \Pr\left[W_{H'}(f) - W_H(f) \geq \eta_0 / 2\right] \leq \exp(-\Omega(nd(\eta_0/2)^2)) = \exp(-\Omega(nd)),
    \]
    since $\eta_0$ is an absolute constant independent from $d$. Using the union bound we have
    \begin{equation}\label{eq:far}
    \Pr\left[\exists f \not\in \cA_{\eta_0 / 2}: W_{H'}(f) \geq s^\ast + \frac{2\log^2 d}{d}\right] \leq 2^n\exp(-\Omega(nd)).
    \end{equation}
    The lemma now follows by taking a union bound over Equation~\eqref{eq:far} and Equation~\eqref{eq:close}.
\end{proof}
\subsection{Analysis of $\tilde{H}$}

We now analyze the number of edges that must be deleted to turn $H'$ into a simple, bounded-degree graph $\tilde{H}$.

\begin{lemma}\label{lem:max_edge_weight}
    There exists some constant $c \in (0, 1)$ independent of $d$ such that $w_H(e) \leq \frac{2}{n} \cdot c^R$ for every $e \in E_H$.
\end{lemma}
\begin{proof}
    Given any $e = \{(v_1, x_1), (v_2, x_2)\} \in E_H$, the probability of sampling $(v_1, x_1)$ as the first endpoint is $1/n$ by symmetry, and conditioned on this we sample $x_2$ (ignoring the requirement on $v_2$) with probability $c^R$ for some $c \in (0, 1)$ since there are $R$ independent coordinates. 
\end{proof}

\begin{lemma}\label{lem:multi-edge_deletion}
    For each $e \in E(H)$, define $N_e = \sum_{i = 1}^{nd} \one\{e_i = e\}$ to be the number of times we include $e$ in $H'$. Then for any fixed $d > 0$, if $R$ is sufficiently large then we have 
    \[
    \sum_{e \in E(H)} \E[[N_e - 1]^+] \leq n.
    \]
\end{lemma}
\begin{proof}
    By Lemma~\ref{lem:max_edge_weight}, we may take $R$ to be sufficiently large so that $\max_e w_e \leq \frac{2}{nd^2}$. For any $e \in E_H$, we have $N_e \sim \mathrm{Bin}(nd, w_e)$ and $[N_e - 1]^+ = N_e - \one\{N_e \geq 1\}$, so
    \begin{align*}
    \sum_{e \in E(H)}\E[[N_e - 1]^+] & = \sum_{e \in E(H)}\E[N_e - \one\{N_e \geq 1\}] \\
    & = nd - \sum_{e \in E}(1 - (1 - w_H(e))^{nd}) \\
    & \leq nd - \sum_{e \in E}\left(1 - \left(1 - nd w_H(e) + \binom{nd}{2}w_H(e)^2\right)\right) \\
    & = \sum_{e \in E}\binom{nd}{2}w_H(e)^2,
    \end{align*}
    where we used the inequality $(1 - p)^k \leq 1 - pk + \binom{k}{2}p^2$ for $0 < p < 1$ and integer $k > 0$. The lemma now follows by observing that $\sum_{e \in E}\binom{nd}{2}w_H(e)^2 \leq \max_{e}w_H(e)\sum_{e \in E}\binom{nd}{2}w_H(e) \leq \frac{2}{nd^2} \leq \frac{n}{\binom{nd}{2}}$.
\end{proof}

\begin{lemma}\label{lem:degree_deletion}
    Let $d$ be sufficiently large. For any vertex $v \in V(H)$, we have $\E[[\deg_{H'}(v) - 10d]^+] \leq 1$. 
\end{lemma}
\begin{proof}
    By symmetry of $H$ we have $\deg_{H'}(v) \sim \mathrm{Bin}(nd, 2/n)$. By the Chernoff bound, we have $\Pr[\deg_{H'}(v) \geq k] \leq \exp(-\Omega(k))$ for any $k \geq 10d$. So we have \[\E[[\deg_{H'}(v) - 10d]^+] = \sum_{k \geq 10d} \Pr[\deg_{H'}(v) \geq k] \leq 1\] when $d$ is sufficiently large. 
\end{proof}

\begin{proof}[Proof of Theorem~\ref{thm:h_tilde}]
    By Lemma~\ref{lem:multi-edge_deletion} and Lemma~\ref{lem:degree_deletion}, the number of edges we delete in the second stage of the reduction is at most $2n$ in expectation, so by Markov's inequality with probability at least 0.99 we remove an $O(1/d)$ fraction of the total number of edges, in which case the value of the maximum cut in the graph also changes at most by $O(1/d)$. By Lemma~\ref{lem:H'}, the theorem then follows.
\end{proof}

\section{Robust Majority Is Stablest Theorem}           \label{sec:robust-mist}
In this section, we prove the following theorem.
\label{sec:MiST}
\begin{theorem}\label{thm:robust-MIST}
Fix $\rho\in(-1,0)$ and let $s_\rho=\arccos(\rho)/\pi$.  There are a constant
$C_\rho<\infty$ and a function $h_\rho:(0,1]\to[0,\infty)$ with
$h_\rho(\tau)=o_{\tau\downarrow 0}(1)$ such that the following holds.  Let $\eta\ge0$ and
$f:\ZeroCube\to[0,1]$.  Assume that $\max_i\Influence_i(f)\le \tau$ and
\[
        \Cut_\rho(f)\ge s_\rho-\eta .
\]
Then
\[
        \dist(f,\cL_R)\le C_\rho\sqrt\eta+h_\rho(\tau).
\]
\end{theorem}

The theorem strengthens Theorem 1.10 of Mossel and Neeman~\cite{MosselNeeman15}, which gives an upper bound of the form $C_\rho \eta^c + h_\rho(\tau)$ for some unspecified constant $c>0$. For our application, the key point is the optimal exponent $\nicefrac12$, which can in fact be attained.

\begin{remark}[Optimality of Theorem \ref{thm:robust-MIST}]        \label{rem:robust-opt}
    We illustrate the asymptotic optimality of Theorem~\ref{thm:robust-MIST}.
    Consider the function $f: \{0, 1\}^R \to \{0, 1\}$ defined by \[(x_1, \ldots, x_R) \mapsto \begin{cases}
        \one\{\sum_{i = 1}^R x_i \geq R / 2\} & \text{if } |\sum_{i = 1}^R x_i - R / 2| \geq \eta \sqrt{R},\\
        1 - \one\{\sum_{i = 1}^R x_i \geq R / 2\} & \text{if }|\sum_{i = 1}^R x_i - R / 2| < \eta \sqrt{R}.
    \end{cases}\]
    It can be verified that this function indicates a quasi-random partition of $\cH^R_{\rho^*}$ which achieves a cut value of $s^\ast - O(\eta^2)$, and its $\ell_1$-distance to the closest linear threshold function $\one\{\sum_{i = 1}^R x_i \geq R / 2 + \eta\sqrt{R}\} $ is $\Omega(\eta)$. 
\end{remark}

The high-level approach behind the proof of Theorem~\ref{thm:robust-MIST} is similar to that of Mossel and Neeman:
\begin{enumerate}
\item Obtain a counterpart of Theorem~\ref{thm:robust-MIST} for Gaussian space: informally, if $f$ has noise stability close to the optimal Gaussian value, then it is close to a linear threshold function.
\item Using the invariance principle, translate this result from Gaussian space to the Boolean cube. Namely, we consider the smoothing $T_\lambda f$ of $f$ as a multilinear function and reinterpret it as a function on Gaussian space. Using item~(1), we show that the truncated multilinear function is close to a linear threshold function in Gaussian space, and then control the truncation error to get the result for the original multilinear function. We then transfer this result back to the Boolean cube.
\item From item~(2), we obtain that $T_{\lambda}f$ is close to $T_{\lambda} L_\lambda$ for some LTF $L_{\lambda}$. Choosing $\lambda$ appropriately and upper bounding the $L_1$ distances between $f$ and $T_\lambda f$, and between $L_\lambda$ and $T_{\lambda}L_\lambda$, we obtain the desired result.
\end{enumerate}

In~\cite{MosselNeeman15}, item~1 follows from a robust version of Borell's stability theorem proved by Mossel and Neeman in the same paper. However, the constant $c$ obtained there is suboptimal. Instead, we rely on a result of Heilman~\cite{Heilman21}, which achieves the optimal exponent of $\nicefrac12$. However, Heilman's result applies only to sets, or equivalently to indicator functions. While the extension is not particularly difficult, adapting the result to arbitrary functions $f:\mathbb{R}^R \to [0,1]$ requires considerable technical work. The argument uses the standard ``epigraph'' reduction, which represents such a function as a subset of $\mathbb{R}^{R+1}$ (note the extra dimension). Part of the technical work is to ensure that this representation preserves the relevant stability and coefficient information needed later in the proof.

It is not enough simply to replace the Gaussian input in~\cite{MosselNeeman15} by Heilman's sharper theorem. Several later steps in the proof of~\cite{MosselNeeman15} are designed only to preserve some positive power of $\eta$, and would therefore degrade the dependence on $\eta$ if used as black boxes. In particular, to go from Gaussian space back to the Boolean cube, they use $L_2$ estimates, replace the approximating LTF by one with modified coefficients, and remove the smoothing requirement. These steps substantially degrade the exponent. For our application, these steps have to be replaced or sharpened so that the $\eta^{1/2}$ dependence on the stability deficit from Heilman's Gaussian theorem is preserved.

The coefficient issue is especially important. The argument in~\cite{MosselNeeman15} only needs closeness of $T_{\rho}f$ to some LTF. Lemma~7.2 of~\cite{MosselNeeman15} does obtain a bound on the coefficients of this LTF in terms of $\tau$, but it does so by modifying the coefficients of the approximating LTF, and this modification substantially degrades the dependence on the stability deficit. In contrast, we need to show that $T_{\rho}f$ is close to an LTF whose coefficients are given by the Fourier coefficients of $f$ in the Hermite basis. This ensures that the coefficient upper bounds depend only on $\tau$, as otherwise we would lose the optimal $\eta^{1/2}$ dependence.

Another important difference from~\cite{MosselNeeman15} is that, to obtain a sharp bound, we need to upper bound most distances between the functions of interest directly in the $L_1$ metric, rather than in the $L_2$ metric as in~\cite{MosselNeeman15}. A frequent technical issue throughout the proof is to keep the parameter losses sharp at each stage, including Gaussian stability, invariance, smoothing, and the final $L_1$ comparisons, so that the exponent $\nicefrac12$ is not degraded. We also use several new technical tools, including the anti-concentration bound for regular linear forms from~\cite{MORS10} and Peres's bound on the noise sensitivity of arbitrary threshold functions~\cite{Peres04}.

\begin{remark}
    Our proof of Theorem \ref{thm:robust-MIST} is based on Heilman's~\cite{Heilman21} asymptotically optimal robust inverse of Borell's isoperimetric inequality. We point out that to the best of our knowledge, at the time of writing of this manuscript  \cite{Heilman21} has not appeared in a peer-reviewed journal yet. This dependence can be avoided by using the slightly weaker, published estimate of Eldan~\cite{Eldan15}; with this substitution and an otherwise unchanged argument, we can obtain a version of Theorem \ref{thm:robust-MIST} with the following error estimate:
    \[
        \dist(f,\cL_R)\le C_\rho\sqrt{\eta \log(1/\eta)}+h_\rho(\tau).
    \]
    Propagating this estimate through the soundness analyses replaces the $O(\log^2(d)/d)$-term with $O(\log^3(d)/d)$; hence Theorems \ref{thm:ugc-hard} and \ref{thm:max-cut-gap-1} continue to hold with $O(\log^3(d)/d)$ in place of $O(\log^2(d)/d)$ using \cite{Eldan15}'s bound.
\end{remark}

\subsection{Notation}
As the proof heavily uses Fourier analysis on the Boolean cube and Gaussian space, we first recall some standard definitions.

We use Fourier analysis on the $\pm1$-cube $\Cube$.  We identify $x\in\ZeroCube$ with
$\sigma\in\Cube$ by setting $\sigma_i=2x_i-1$.  Thus, for
$f:\ZeroCube\to\bbR$,
\[
        f(x)=\sum_{S\subseteq[R]}\widehat f(S)\chi_S(x),
        \qquad
        \chi_S(x)=\prod_{i\in S}\sigma_i(x).
\]
We use the same notation for functions defined directly on $\Cube$.  The
influence of coordinate $i$ is
\[
        \Influence_i(f)=\sum_{S:i\in S}\widehat f(S)^2 .
\]
For $\lambda\in[-1,1]$, the noise operator is
\[
        T_\lambda f=\sum_{S\subseteq[R]}\lambda^{|S|}\widehat f(S)\chi_S .
\]
For $t\in[-1,1]$, we write $X\sim_t Y$ if the sign coordinates of
$X,Y\in\ZeroCube$ are coordinatewise $t$-correlated.  For Gaussian vectors,
$G\sim_t H$ means that $G$ and $H$ are standard Gaussian vectors with
coordinatewise covariance $t$.  The noise stability of $f$ is
$\Stab_t(f)=\E[f(X)f(Y)] = \sum_{S\subseteq [R]} t^{|S|} \widehat f(S)^2$, where $X\sim_t Y$.
For $a,b\in[0,1]$, let $\Cut(a,b)=a + b -2ab$ and
$\Cut_\rho(f)=\E[\Cut(f(X),f(Y))]$, where $X\sim_\rho Y$.

Below we use the following parametrization of linear threshold functions (LTF) on $\ZeroCube$:
\[
        L(x)=\one\left(\sum_{i=1}^R w_i(2x_i-1)+\theta\ge 0\right) .
\]
Constant $0$ and constant $1$ functions are also LTFs.
Note that this parametrization slightly differs from that we used before (this parametrization is more convenient when we consider threshold linear functions on $\Cube$). Recall that $\cL_R$ is the set of LTFs on $\ZeroCube$.  Slightly abusing the  notation, we will also denote the set of LTFs on $\Cube$ by $\cL_R$. For a function
$h$ on either cube, write
\[
        \dist(h,\cL_R)=\inf_{L\in\cL_R}\|h-L\|_1 .
\]
We also use LTFs on Gaussian space; these are $0$-$1$ indicator function of affine halfspaces.
We write $\gamma_R$ for the standard Gaussian measure on $\bbR^R$.
Denote the density of the standard normal distribution by $\varphi$ and its
distribution function by $\Phi$.  For $t\in(-1,1)$, define
\[
        \Lambda_t(a,b)=\Pr(G\le\Phi^{-1}(a),\ H\le\Phi^{-1}(b)) \text{ where }
        G\sim_t H.
\]
Let $\Gamma_t(a)=\Lambda_t(a,a)$.  We also define $\Gamma_t^\pm$ (informally, a variant of $\Gamma_t(a)$ for $\pm 1$-valued indicator functions)
\[
        \Gamma_t^\pm(m)=4\Gamma_t\left(\frac{1+m}{2}\right)-2(1+m)+1.
\]

\subsection{Proof overview}
As outlined in the beginning of this section, the proof consists of the three main steps. In this section, we formally present these steps, Lemmas~\ref{lem:gaussian-main}, \ref{lem:smoothed-boolean-transfer}, and \ref{lem:smoothing-removal}, and then show that Theorem~\ref{thm:robust-MIST} follows from them.  We present the proofs of these lemmas in Appendix~\ref{apx:MiST}. Below, all constants in the $O$ and $o$ notation may depend on $\rho$.

First, we prove a Gaussian analogue of Theorem~\ref{thm:robust-MIST}.  
For a function $F:\bbR^R\to[0,1]$ on Gaussian space, write
\[
        \widehat F_i=\E[F(G)G_i]
        \qquad\text{and}\qquad
        \widehat F^{(1)}=(\widehat F_1,\ldots,\widehat F_R),
\]
where $G$ is a standard Gaussian vector.  This is the vector of first-order Fourier coefficients of $F$ in the Hermite basis.

\begin{lemma}[Analogue of Theorem~\ref{thm:robust-MIST} in Gaussian space]\label{lem:gaussian-main}
Fix $0<t<1$ and a closed interval $I\subset(0,1)$.  There are positive constants
$C$, $c_0$, and $\delta_0$ such that the following holds.  Let
$F:\bbR^R\to[0,1]$ be measurable. Denote $m=\E F$ and assume $m\in I$. Let
\[
        \Delta=\Gamma_t(m)-\E_{G\sim_t H}[F(G)F(H)].
\]
If $\Delta\le\delta_0$, then $\|\widehat F^{(1)}\|\ge c_0$ and there is an LTF
$L$ whose coefficient vector is proportional to $\widehat F^{(1)}$ such that
\[
        \|F-L\|_1\le C\sqrt\Delta .
\]
\end{lemma}

The second step transfers the Gaussian result to the Boolean setting, assuming that the function is first smoothed by the noise operator $T_\lambda$. 

\begin{lemma}\label{lem:smoothed-boolean-transfer}
Fix $0<t<1$ and $\lambda\in[19/20,1)$.  Let
$u:\Cube\to[0,1]$.  Suppose that $\E u\in[1/4,3/4]$ and
$\max_i\Influence_i(u)\le\tau$.  Let
$D=\pos{\Lambda_t(\E u,\E u)-\Stab_t(u)}$.  Then there is an LTF $L_\lambda$, either
a  $0/1$-constant function or of the form
\[
        L_\lambda(\sigma)=\one\left(\langle w_\lambda,\sigma\rangle\ge\theta_\lambda\right)
        \qquad \text{where }\|w_\lambda\|_2=1 \quad \text{and}\quad
        \|w_\lambda\|_{\infty}=o_{\tau\downarrow 0}(1),
\]
such that
\[
        \|T_\lambda u-T_\lambda L_\lambda\|_1
        \le C_t\sqrt D+C_t\sqrt{1-\lambda}+o_{\tau\downarrow 0}(1).
\]
Here the $o_{\tau\downarrow 0}(1)$ terms may depend on $t$ and $\lambda$ and go to $0$ when $\tau \to 0$.
\end{lemma}

The third step removes the smoothing.
\begin{lemma}\label{lem:smoothing-removal}
There is a universal constant $C<\infty$ such that the following holds.  Let
$u:\Cube\to[0,1]$ and let $L$ be either a $0/1$-constant function or of the form
$L(\sigma)=\one\left(\langle w,\sigma\rangle\ge\theta\right)$, where $\|w\|_2=1$. Denote $\mu =\|w\|_\infty$; in the constant case, $\mu=0$.  Then, for every
$\lambda\in[0,1)$,
\[
        \|u-L\|_1
        \le \|T_\lambda u-T_\lambda L\|_1+C\sqrt{1-\lambda}+o_{\mu\downarrow0}(1).
\]
\end{lemma}

\begin{corollary}\label{cor:positive-transfer}
Fix $0<t<1$.  Let $u:\Cube\to[0,1]$.  Assume that
$\max_i\Influence_i(u)\le\tau$ and $\E u\in[1/4,3/4]$.  Then
\[
        \dist(u,\cL_R)\le C_t\sqrt{\pos{\Gamma_t(\E u)-\Stab_t(u)}}+o_{\tau\downarrow 0}(1).
\]
\end{corollary}

\begin{proof}
Let $D=\pos{\Lambda_t(\E u,\E u)-\Stab_t(u)}$.  Fix
$\lambda\in[19/20,1)$.  By Lemma~\ref{lem:smoothed-boolean-transfer} and
Lemma~\ref{lem:smoothing-removal}, there is an LTF $L_\lambda$ such that
\[
        \|u-L_\lambda\|_1
        \le C_t\sqrt D+C_t\sqrt{1-\lambda}+\varepsilon_\lambda(\tau),
\]
where $\varepsilon_\lambda(\tau)\to0$ as $\tau\to0$, for every fixed
$\lambda<1$.  Define
\[
        h_t(\tau)=\inf_{19/20\le\lambda<1}
        \left(C_t\sqrt{1-\lambda}+\varepsilon_\lambda(\tau)\right).
\]
Then $\dist(u,\cL_R)\le C_t\sqrt D+h_t(\tau)$.  Also $h_t(\tau)=o_{\tau\downarrow 0}(1)$:
given $\varepsilon>0$, choose a fixed $\lambda<1$ with
$C_t\sqrt{1-\lambda}<\varepsilon/2$, and then take $\tau$ small enough so that
$\varepsilon_\lambda(\tau)<\varepsilon/2$.
\end{proof}

\begin{proof}[Proof of Theorem~\ref{thm:robust-MIST}]
Let $u:\Cube\to[0,1]$ be the version of $f$ on $\Cube$:
\[
        u(\sigma)=f\left(\frac{1+\sigma_1}{2},\ldots,\frac{1+\sigma_R}{2}\right).
\]
The influences and the cut values of $f$ and $u$ are the same.  Let 
$g=2u-1$ and $m=\E g$. 

\paragraph{Upper bound on mean $m$.} 
We first show that \(\E u\in[1/4,3/4]\), or equivalently \(|m|\le1/2\).
By the definition of $\Cut_\rho$:
$\Cut_\rho(f)=(1-\Stab_\rho(g))/2$.  Since
$s_\rho=(1-\Gamma_\rho^\pm(0))/2$, the assumption on the cut value of $f$ gives
\begin{equation}\label{eq:stab-condition}
        \Stab_\rho(g)\le \Gamma_\rho^\pm(0)+2\eta .
\end{equation}
Let
$g_{\mathrm{even}}(x)=(g(x)+g(-x))/2$ and
$g_{\mathrm{odd}}(x)=(g(x)-g(-x))/2$ be the even and odd parts of $g$, respectively; $g_{\mathrm{even}}$ contains the even-degree terms of the Fourier expansion of $g$ and $g_{\mathrm{odd}}$ contains the odd-degree terms. 
Clearly, $\E g_{\mathrm{odd}}=0$ and $\E g_{\mathrm{even}} = m$. Also, we have 
$$\Stab_\rho(g_{\mathrm{even}}) = \sum_{\substack{S\subseteq [R]\\|S| \text{ is even}}} \rho^{|S|} \widehat f(S)\geq \widehat f(\varnothing) = m^2.$$
Further, $\max_i\Influence_i(g_{\mathrm{odd}})\le \max_i\Influence_i(g) \leq 4\tau$.  Applying
Theorem~\ref{thm:maj-stablest} to $\tilde u=(g_{\mathrm{odd}}+1)/2$ gives (since $\rho < 0$):
$\Stab_\rho(g_{\mathrm{odd}})\ge\Gamma_\rho^\pm(0) - o_{\tau\downarrow 0}(1)$.  
Note that $\Stab_\rho(g)= \Stab_\rho(g_{\mathrm{even}})+\Stab_\rho(g_{\mathrm{odd}})$ and therefore
\[m^2 \leq \Stab_t(g_{\mathrm{even}}) = \Stab_\rho(g) - \Stab_\rho(g_{\mathrm{odd}}) \stackrel{\tiny\text{by \eqref{eq:stab-condition}}}{\le}
\bigl(\Gamma_\rho^\pm(0)+2\eta\bigr) - \bigl(\Gamma_\rho^\pm(0) -o_{\tau\downarrow 0}(1)\bigr)\le   2\eta+o_{\tau\downarrow 0}(1).
\]
If the right-hand side is larger than $1/4$, then the desired conclusion follows after increasing $C_\rho$ and $h_\rho$, since $L_1$ distances are at most $1$.  Thus we may assume below that the right-hand side is at most $1/4$.  Then $|m|\le1/2$ and thus $\E u\in[1/4,3/4]$.

Now, 
\[\Stab_{-\rho}(g) = \Stab_{\rho}(g_\mathrm{even}) - \Stab_{\rho}(g_\mathrm{odd}) = 2\Stab_{\rho}(g_\mathrm{even}) - \Stab_{\rho}(g)
\geq -\Stab_{\rho}(g) \geq -\Gamma_\rho^\pm(0) - 2\eta,\]
here we used that $\Stab_{\rho}(g_\mathrm{even}) \geq 0$.
As we will prove in Lemma~\ref{lem:gaussian-facts}, item \ref{it:lambda-quadratic},
$\Gamma_t^\pm(m)-\Gamma_t^\pm(0)\le C_t m^2$ for $t\geq 0$. Applying this claim with $t= -\rho$ and using that $\Gamma_\rho^\pm(0) = - \Gamma_t^\pm(0)$, we get
\[
        \pos{\Gamma_t^\pm(m)-\Stab_t(g)}\le C_\rho\eta+o_{\tau\downarrow 0}(1).
\]
Since
$\Gamma_t^\pm(m)-\Stab_t(g)=4(\Gamma_t(\E u)-\Stab_t(u))$, Corollary~\ref{cor:positive-transfer} gives
\[
        \dist(u,\cL_R)\le C_\rho\sqrt\eta+o_{\tau\downarrow 0}(1).
\]
Going back to $\ZeroCube$, we get
$\dist(f,\cL_R)\le C_\rho\sqrt\eta+o_{\tau\downarrow 0}(1)$.
\end{proof}

\part{Integrality Gap}
\section{Integrality Gap for Bounded-Degree Instances}			\label{sec:sdp-gap}

In this section, we show that Goemans-Williamson relaxation for Max-Cut with triangle inequalities, referred to as $\tSDP$ henceforth, also admits a $(\alpha_{GW} + \tilde{O}(1/d))$-integrality gap. 

\begin{theorem}[Restatement of Theorem \ref{thm:max-cut-gap-1}]				\label{thm:max-cut-gap}
	For every sufficiently large integer $d$, there exists a graph $G = (V,E)$ with maximum degree at most $d$ such that the following holds.
	\begin{itemize}
		\item {\bf Completeness}. The optimal value of $\tSDP$ for $G$ is at least $c^* - O(1/d)$.
		\item {\bf Soundness}. The optimal Max-Cut value of $G$ is at most $s^* + O\left(\frac{\log^2(d)}{d}\right)$.
	\end{itemize}
\end{theorem}

The proof of the above theorem largely follows the proof of Theorem \ref{thm:ugc-hard}; the only notable difference is that we use the integrality gap instances of Unique Games constructed by Khot and Vishnoi~\cite{KV15} as the starting point (instead of its UGC-based hardness), and we then apply the same reduction on this instance to construct instances that are bounded-degree integrality gap instances for Max-Cut. As in the proof of Theorem \ref{thm:ugc-hard}, our bounds for the max-cut value of the sub-sampled noisy cube will be the key ingredient driving the analysis of the gap instance. 

\begin{remark}
    It should be possible to derive quantitatively similar lower bounds for stronger relaxations such as Lov\'asz-Schrijver and Sherali-Adams by combining our techniques with the guarantees of the Unique Games gap instances given by \cite{RS09,KS09} instead.
\end{remark}

\subsection{Integrality Gap Instance}

To begin with, we recall the following integrality gap construction for the Max-Cut problem from Khot-Vishnoi \cite{KV15}.

{\bf Khot-Vishnoi Max-Cut Integrality Gap}. Our starting point is the following classical gap construction for Max-Cut, which in turn begins with the following SDP integrality gap for the Unique Games problem.

\begin{theorem}[Theorem 3.7~\cite{KV15}]			\label{thm:ugc-gap}
	For every sufficiently small constant $\eps \in (0,1)$, there exists a Unique Games instance $\cG_{\eps}$ on $R = R_\eps$ labels such that the following holds:
	\begin{itemize}
		\item The optimal value of the SDP relaxation with triangle inequalities is at least $1 - \eps$.
		\item The optimal integral value of $\cG_{\eps}$ is at most $\eps$.
	\end{itemize} 
\end{theorem}

Now, as shown in \cite{KV15}, by using the standard reduction from Unique Games to Max-Cut from Figure \ref{fig:reduction}, one can lift the above gap for Unique Games to a corresponding integrality gap for Max-Cut. 

\begin{theorem}[Implicit in \cite{KV15}, also see~\cite{Rag08}]			\label{thm:kv-gap}
	The following holds for all small enough $\eps \in (0,1)$. Let $\cG_\eps$ denote the integrality gap instance for Unique Games from Theorem \ref{thm:ugc-gap}. Let us apply the reduction from Figure \ref{fig:reduction} on $\cG_{\eps}$, and let $H_{\eps} = (V_{\eps},E_{\eps},w_{\eps})$ be the Max-Cut instance output by the reduction. Then, 
	\begin{itemize}
		\item $\tSDP(H_{\eps}) \geq c^* - \delta(\eps)$,
		\item $\MaxCut(H_\eps) \leq s^* + \delta(\eps)$,
	\end{itemize}
	where $\delta(\eps)$ is a continuous function satisfying $\delta(\eps) \downarrow 0$ as $\eps \downarrow 0$. Here, $\tSDP(H_{\eps})$ denotes the optimal value for the $\tSDP$ relaxation for the instance $H_{\eps}$.
\end{theorem}	

For the next part, we will fix $H := H_\eps$ by choosing $\eps$ small enough to ensure the following conditions hold:
\begin{itemize}
    \item Firstly, we choose $\eps \leq \eps(d)$ to be small enough as a function of $d$ to ensure that $\delta(\eps) \leq 1/d$.
    \item Furthermore, for the fixed choice of $d$, we will ensure $\eps$ is small enough to ensure the alphabet set $[R_\eps]$ is large enough as specified in Theorem \ref{thm:h_tilde}.
\end{itemize}

{\bf Sub-sampling Step}. Next, let us apply the degree-reduction step (Steps (R1) and (R2)) from Section \ref{sec:main-proof} on $H$ to construct an instance $\tilde{H}$ with maximum degree at most $d$. We restate the steps for ease of reference below.

\begin{itemize}
	\item[(R1)] We first construct an intermediate multi-graph $H' = (V,E')$ by independently sampling $nd$ edges $e_1, \ldots, e_{nd}$ (where $n = |V(H)|$) from $H$ using the probability distribution induced by the edge weights $w_H$. Note that this multi-graph is unweighted.
	\item[(R2)] Then, we construct $\tilde{H} = (\tilde{V},\tilde{E})$ from $H'$ by first replacing each set of parallel edges by a single edge, and then arbitrarily deleting  additional edges, if necessary, so that every vertex has degree at most $10d$.
\end{itemize} 

The final graph $\tilde{H}$ will be the instance that will be used to certify the integrality gap of the SDP relaxation $\tSDP$.

\subsection{Proof of Theorem \ref{thm:max-cut-gap}}

We now establish the completeness and soundness guarantees of the instance $\tilde{H}$ constructed above. To begin with, observe that since the construction of $\tilde{H}$ from the Unique Games instance $\cG_{\eps}$ is identical to the reduction used in the proof of Theorem \ref{thm:ugc-hard}, as a consequence we can immediately obtain the following soundness guarantee of this construction.

\begin{lemma}[Soundness]			\label{lem:gap-sound}
	Let $\tilde{H} := \tilde{H}_{\eps}$ be the (random) Max-Cut instance described above. Then with probability at least $0.99$, we have that
	\[
	\MaxCut(\tilde{H}) \leq \frac{\arccos \rho^*}{\pi} + O\left(\frac{\log^2(d)}{d}\right).
	\]
\end{lemma}
The proof of the above lemma is identical to that of the soundness item in Theorem \ref{thm:h_tilde} and hence we omit it here. On the other hand, note that the nature of the completeness guarantee here is qualitatively different -- i.e.,  it needs to certify that the sub-sampled instance approximately preserves the optimal SDP value of $H$. However, this difference doesn't introduce any additional complications in the proof, and we can establish it using the same techniques as in the completeness direction of Theorem \ref{thm:h_tilde}. 

\begin{lemma}[Completeness]				\label{lem:gap-comp}
	Suppose $\eps \leq \eps(d)$ is small enough. Then with probability at least $0.98$, we have 
	\[
	\tSDP(\tilde{H}) \geq \frac{1 - \rho^*}{2} - O(1/d).
	\]
\end{lemma}
\begin{proof}
	To begin with, note that from the guarantee of Theorem \ref{thm:kv-gap} we know that $\tSDP(H) \geq \frac{1 - \rho^*}{2} - \delta(\eps)$, and let $\{v_i\}_{i \in V}$ be an optimal SDP solution achieving the optimal value $\tSDP(H)$. Now, recall that $\tilde{H}$ is constructed from $H$ as follows\footnote{We recall the steps here just to make it easier to verify that the SDP solution for $H$ continues to be a feasible SDP solution for $H'$ and $\tilde{H}$.}:
	
	\begin{itemize}
		\item[(R1)] We first construct an intermediate multi-graph $H' = (V,E')$ by independently sampling $nd$ edges $e_1, \ldots, e_{nd}$ from $H$ (where $n = |V(H)|)$ using the probability distribution induced by the edge weights $w_H$. Note that this multi-graph is unweighted.
		\item[(R2)] Then, we construct $\tilde{H} = (\tilde{V},\tilde{E})$ from $H'$ by first replacing each set of parallel edges by a single edge, and then arbitrarily deleting  additional edges, if necessary, so that every vertex has degree at most $10d$.
	\end{itemize} 
	
	Observe that the edges in $H'$ are a multiset of a subset of the edges in $H$, the vector solution $\{v_i\}_{i \in V}$ is also a feasible solution for the SDP relaxation for $H'$. Furthermore, since $\tilde{H}$ is a sub-graph of $H'$, the same conclusion holds for $\tilde{H}$ as well. To conclude the proof, we will need to bound the objective value achieved by the vector solution $\{v_i\}_{i \in V}$ for the SDP relaxation for $H'$. As in the proof of Theorem \ref{thm:ugc-hard}, this can be done in a couple of steps. 
	
	Firstly, analogous to the completeness analysis from Section \ref{sec:H'-analysis}, we have the following lemma.
	
	\begin{claim}		\label{cl:H'-1}
		With probability at least $0.99$, we have 
		\[
		\Ex_{\{i,j\} \sim E'}\left[ \frac{\|v_i - v_j\|^2}{4} \right] \geq c^* - O(1/d).
		\]
		Here $\{i,j\} \sim E'$ corresponds to sampling an edge from $E'$ with probability proportional to the number of copies of $\{i,j\} \in E'$.
	\end{claim}
	\begin{proof}
		We proceed identically as in the completeness analysis in Section \ref{sec:H'-analysis}. For every edge $e = \{i,j\} \in E$, let us denote 
		\[
		Z_e := \frac{\|v_i - v_j\|^2_2}{4}.
		\]
		Then, using the completeness guarantee of Theorem \ref{thm:kv-gap}, we have that
		\[
		\Ex_{e \sim w_H}[Z_e] = \tSDP(H) \geq c^*  - \delta(\eps),
		\] 
        where $w_H$ corresponds to sampling according to the edge-weight distribution of $H$. Next, for the $a^{th}$ random edge $e_a$, let us denote $X_a := Z_{e_a}$. Then, using Lemma \ref{lem:Bernstein} we again have
		\[
		\Pr_{H'}\left[\frac{1}{nd}\sum_{i = 1}^{nd} X_{i} \geq \tSDP(H) - O(1/d)\right] \geq 1 - \exp(-\Omega(n/d)) \geq 0.99,
		\]
		Hence, with probability at least $0.99$, we have
		\[
		\Ex_{\{i,j\} \sim E'}\left[\frac{\|v_i - v_j\|^2}{4}\right] = \Ex_{a \sim [nd]}[X_a]
		\geq \tSDP(H) - O(1/d) \geq c^* - \delta(\eps) - O(1/d) \geq c^* - O(1/d),
		\]
		where the last step follows by choosing $\eps$ small enough as a function of $d$.
	\end{proof}
	
	Next, we establish the following lower bound on the objective value for the SDP relaxation of $\tilde{H}$.
	
	\begin{claim}		\label{cl:H'-2}
		With probability at least $0.99$, we have
		\[
		\Ex_{\{i,j\} \sim \tilde{E}}\left[\frac{\|v_i - v_j\|^2}{4}\right] \geq \Ex_{\{i,j\} \sim E'}\left[\frac{\|v_i - v_j\|^2}{4}\right] - O(1/d).
		\]
	\end{claim}
	\begin{proof}
		Recall that the proof of Theorem \ref{thm:ugc-hard} shows that with probability at least $0.99$, the graph $\tilde{H}$ is constructed by deleting at most $O(1/d)$-fraction of edges in $H'$, and hence the claim follows.
	\end{proof}
	
	Combining the two bounds above, we get that with probability at least $0.98$,
	\[
	\Ex_{\{i,j\} \sim \tilde{E}}\left[\frac{\|v_i - v_j\|^2}{4}\right] \geq c^* - O(1/d),
	\] 
	which directly implies that $\tSDP(\tilde{H}) \geq c^* - O(1/d)$.
\end{proof}

Finally, we are ready to prove Theorem \ref{thm:max-cut-gap}.

\begin{proof}[Proof of Theorem \ref{thm:max-cut-gap}]
	Let $\eps$ be fixed to be small enough such that the error parameter from Theorem \ref{thm:kv-gap} satisfies $\delta(\eps) \leq 1/d$. Let $H := H_\eps$ be the corresponding integrality gap instance from Theorem \ref{thm:kv-gap}, and let $\tilde{H}$ be the  instance obtained by applying steps (R1) and (R2) from above. Then, by construction, note that the final graph $\tilde{H}$ has maximum degree $10d$. Furthermore, by combining the guarantees of Lemma \ref{lem:gap-comp} and Lemma \ref{lem:gap-sound}, with probability at least $0.97$, we have that 
	\[
	\tSDP(\tilde{H}) \geq c^* - O(1/d)
	\qquad\qquad\text{and}\qquad\qquad
	{\sf MaxCut}(\tilde{H}) \leq s^* + O\left(\frac{\log^2(d)}{d}\right),
	\]
	which concludes the proof of Theorem \ref{thm:max-cut-gap}.
\end{proof}

\bibliographystyle{alpha}
\bibliography{main}

\part{Appendix}
\appendix
\section{Bounding the number of LTFs}

\begin{theorem}\label{thm:ltf-count}
The number of linear threshold functions $L:\{0,1\}^R \to \{0,1\}$ is at most $2^{O(R^2)}$.
\end{theorem}

\begin{proof}
To upper bound the number of distinct linear threshold functions on $\{-1,1\}^R$, we employ the following standard facts about VC dimension.

\begin{theorem}[VC Dimension of Affine Halfspaces; see, e.g., Example 8.3.5 in~\cite{vershynin2025high}]\label{thm:vc-dim-halfspaces}
The class of affine halfspaces in $\mathbb R^R$ (i.e., indicator functions of sets $\{x : \langle w, x\rangle \leq \theta\}$) has VC dimension $R+1$.
\end{theorem}

\begin{theorem}[Sauer--Shelah Lemma; see, e.g., Lemma 8.3.9 in~\cite{vershynin2025high}]\label{thm:Sauer-Shelah}
Let $\mathcal{F}$ be a class of Boolean functions on an $M$-point set $X$. Suppose the VC dimension of $\mathcal{F}$ is $d$. Then,
\[
    |\mathcal{F}|
    \le
    \sum_{i=0}^d \binom{M}{i}\leq
    \left(\frac{eM}{d}\right)^d.
\]
\end{theorem}

By Theorems~\ref{thm:vc-dim-halfspaces} and~\ref{thm:Sauer-Shelah}, the number of distinct linear threshold functions on $\{-1,1\}^R$ is upper bounded by
\begin{equation}\label{eq:tlf-count}
\left(\frac{e2^R}{R+1}\right)^{R+1} = 2^{O(R^2)}.
\end{equation}
\end{proof}

\section{Hamming Ball Estimates}

The following estimate is well-known.

\begin{fact}		\label{fact:bin}
	Let $\delta \in (0,1)$ be small enough. Then for all $n \geq n(\delta)$, we have
	\[
		{n \choose {\leq \delta\cdot n}} \leq 2^{O(\delta \log(1/\delta))\cdot n}.
	\]
\end{fact}

The following is well known.

\begin{lemma}		\label{lem:ball}
	Let $\delta \in (0,1)$ be small enough and $S$ a finite set. Then for any given function $f:S \to \{0,1\}$, the number of functions $g:S \to \{0,1\}$ such that $\|f- g\|_1 \leq \delta$ is at most 
	\[
		2^{O(\delta\log(1/\delta))\cdot |S|}.
	\]
\end{lemma}
\begin{proof}
	Note that the set of functions $g:S \to \{0,1\}$ satisfying $\|f - g\|_1 \leq \delta$ is exactly the set of functions which are in a Hamming ball of radius $\delta \cdot |S|$ around $f$. Hence, the number of such functions is at most
	\[
	{|S| \choose {\leq \delta \cdot |S|}} \leq 2^{O(\delta\log(1/\delta))\cdot |S|},
	\]
 	where the inequality follows from Fact \ref{fact:bin}.
\end{proof}
 	
\section{Max-Cut Reduction Analysis~\cite{KKMO07}}			\label{sec:KKMO}			

In this section, we analyze the reduction described in Figure \ref{fig:reduction}. We point out that the analysis presented here is identical to that of \cite{KKMO07}, we include it here for the sake of completeness and easy verification of the various parameters used in the proof of Lemma \ref{lem:near-opt}. Let us begin by recalling the reduction from Figure \ref{fig:reduction} here.

\begin{figure}[ht!]
	\begin{mdframed}
		Given a {\bf right-regular} UG instance $\cG = (U,V,E,[R],\{\pi_e\}_{e \in E})$, we construct a Max-Cut instance $H = (V_H,E_H,w_H)$ as follows: \\
		
		{\bf Vertex Set}. The vertex set of $H$ is $V_H = V \times \{0,1\}^R$ (where $R$ is chosen large enough). \\
		
		{\bf Edges}. The distribution over the edges is given by the following process: \\
		
		\begin{itemize}
			\item Sample $u \sim U$, and sample two neighbors $v_1,v_2 \sim N_\cG(u)$. \\[-6pt]
			\item Sample a $\rho^*$-correlated pair $x_1 \underset{\rho^*}{\sim} x_2$ from $\{0,1\}^R$.
			\item Sample $x'_1 \underset{1 - \mu}{\sim} x_1$ and $x'_2 \underset{1 - \mu}{\sim} x_2$ independently. 
			\item Output the edge $\big\{(v_1,\pi_{u,v_1}(x'_1)), (v_2,\pi_{u,v_2}(x'_2))\big\}$.
		\end{itemize}	
	\end{mdframed}
	\caption{Unique Games to Max-Cut Reduction}
	\label{fig:reduction-copy}
\end{figure}

We analyze the completeness and soundness guarantees of the above reduction below.

\subsection{Proof of Lemma \ref{lem:comp}}

Suppose $\opt(\cG) \geq 1 - \eps$. Let $\sigma: U \cup V \to [R]$ be the labeling of $\cG$ that satisfies $\opt(\cG)$-fraction of edges. Then consider the dictator cut $f:V_{H} \to \{0,1\}$ which is defined as $f(v,x) = x_{\sigma(v)}$. Now, observe that 
\begin{align*}
\Pr_{e \sim H}[\text{$f$ cuts $e$}]
&= \Ex_{u,v_1,v_2}\Pr_{(x'_1,x'_2)}[f_{v_1}(\pi_{u,v_1}(x'_1)) \neq f_{v_2}(\pi_{u,v_2}(x'_2))] \\
&= \Ex_{u,v_1,v_2}\Pr_{(x'_1,x'_2)}[x'_1(\pi_{u,v_1}(\sigma(v_1))) \neq x'_2(\pi_{u,v_2}(\sigma(v_2)))] 
\end{align*}
Now, fix choice of $u \in U$ and $v_1,v_2 \in N_{\cG}(u)$ such that $\sigma$ satisfies the edges $(u,v_1)$ and $(u,v_2)$. Then, conditioned on such a choice of $u,v_1,v_2$, the inner probability term can be bounded as follows:
\[
\Pr_{(x'_1,x'_2)}[(x'_1(\pi_{u,v_1}(\sigma(v_1)))) \neq x'_2(\pi_{u,v_2}(\sigma(v_2)))]
= \Pr_{(x'_1,x'_2)}[(x'_1(\sigma(u))) \neq x'_2(\sigma(u))]
\geq \frac{1 - \rho^*}{2} - 2\mu,
\]
where the last step uses the observations that (i) a pair of $\rho^*$-correlated bits are not equal with probability $(1 - \rho^*)/2$ and (ii) $x'_i = x_i$ with probability at least $1 - \eta$ for $i = 1,2$. Together with the observation that the probability of the event $\sigma$ satisfies the edges $(u,v_1)$ and $(u,v_2)$ holds with probability at least $1 - 2\eps$, this implies that we have 
\[
\Pr_{e \sim H}[\text{$f$ cuts $e$}] \geq \frac{1 - \rho^*}{2} -2 \eps - 2\mu = c^* - 2(\mu + \eps).
\]

\subsection{Proof of Lemma \ref{lem:sound}}

Suppose $\cG$ is a Unique Games instance satisfying $\opt(\cG) \leq \eps$. Fix a cut $f:V_{H} \to \{0,1\}$ in $H$. As previously defined in Section \ref{sec:analysis}, for every $v \in V$, let $f_v:\{0,1\}^R \to \{0,1\}$ be the right-vertex function defined as $f_v(x) = f(v,x)$, and for every left vertex $u \in U$, the corresponding averaged long code table $g_u:\{0,1\}^R \to [0,1]$ defined as 
\[
g_{u}(x) := \Ex_{v \sim N_{\cG}(u)}\Ex_{x' \underset{1 - \eta}{\sim} x}\Big[f_v\Big(\pi_{u,v}(x')\Big)\Big].
\]
Now, using the description of the graph $H$ and the bilinearity of the ${\sf Cut}(\cdot,\cdot)$ function, we can express the fraction of edges cut by $f$ as 
\begin{align}
	\Pr_{e \sim H}\Big[\text{$f$ cuts $e$}\Big]
	&= \Ex_{u \sim U}\Ex_{v_1,v_2 \sim N_{\cG}(u)}\Ex_{x_1 \underset{\rho^*}{\sim} x_2}\Ex_{x'_i \underset{1 - \mu}{\sim}x_i}\Big[{\sf Cut}\Big(f_{v_1}(\pi_{u,v_1}(x'_1)),f_{v_2}(\pi_{u,v_2}(x'_2))\Big)\Big] 		\nonumber\\
	&= \Ex_{u \sim U}\Ex_{v_1,v_2 \sim N_{\cG}(u)}\Ex_{x_1 \underset{\rho^*}{\sim} x_2}\Big[{\sf Cut}\Big(T_{1 - \mu}f_{v_1}(\pi_{u,v_1}(x'_1)),T_{1 - \mu}f_{v_2}(\pi_{u,v_2}(x'_2))\Big)\Big] 		\nonumber\\
	&= \Ex_{u \sim U}\Ex_{x_1 \underset{\rho^*}{\sim} x_2}\Big[{\sf Cut}\Big(\Ex_{v_1 \sim N_{\cG}(u)}T_{1 - \mu}f_{v_1}(\pi_{u,v_1}(x_1)),\Ex_{v_2 \sim N_{\cG}(u)}T_{1 - \mu}f_{v_2}(\pi_{u,v_2}(x_2))\Big)\Big] 		\nonumber\\
	&= \Ex_{u \sim U}\Ex_{x_1 \underset{\rho^*}{\sim} x_2}\Big[{\sf Cut}\Big(g_{u}(x_1),g_u(x_2)\Big)\Big].		\label{eqn:cut} 
\end{align}

Now, recall that the set of good vertices $U_{\rm good}$ was defined as
\[
U_{\rm good} = \left\{u \in U : \max_{i \in [R]} \Inf{i}{g_u} \leq \tau \right\}.
\]
Next we recall the first key claim in the soundness analysis which says that most of the left vertices are good.

\begin{lemma}[Restatement of Lemma \ref{lem:u-good}]			\label{lem:u-good-1}
	Assuming $\opt(\cG) \leq \eps$ we have that
	\[
		|U_{\rm good}| \geq (1 - \nu(\eps))|U|,
	\]
	where $\nu(\eps) = \sqrt{\eps}$.
\end{lemma}

Let us defer the proof of the above lemma to Appendix \ref{sec:u-good} for now, and continue with the proof of the soundness analysis. We then state the key technical step in the soundness analysis.

\begin{lemma}[Restatement of Lemma \ref{lem:good-cut}]			\label{lem:good-cut-1}
	For every $u \in U_{\rm good}$ we have that 
	\[
		\Ex_{x_1 \underset{\rho^*}{\sim} x_2}\Big[{\sf Cut}(g_u(x_1),g_u(x_2))\Big] \leq s^* + c(\eps).
	\]
\end{lemma}
\begin{proof}
	Fix a vertex $u \in U_{\rm good}$. Then by definition of $U_{\rm good}$ we have 
	\[
		\max_{i \in [R]} \Inf{i}{g_u} \leq \tau,
	\]
	i.e., the function $g_u$ satisfies the premise of Corollary \ref{cor:maj-stablest}. Hence, by applying Corollary \ref{cor:maj-stablest} we get that
	\begin{align*}
	\Ex_{x_1 \underset{\rho^*}{\sim} x_2}\Big[{\sf Cut}(g_u(x_1),g_u(x_2))\Big] 
	&\leq \frac{\arccos(\rho^*)}{\pi} + O\left(\frac{\log \log(1/\tau)}{\log(1/\tau)}\right) \\
	&= s^* + c(\eps).
	\end{align*}
\end{proof}

Therefore, combining Equation \eqref{eqn:cut}, Lemma \ref{lem:u-good-1}, and Lemma \ref{lem:good-cut-1}, we get that
\begin{align*}
	\Pr_{e \sim H}[\text{$f$ cuts $e$}]
	&= \Ex_{u \sim U}\Ex_{x_1 \underset{\rho^*}{\sim} x_2}\Big[{\sf Cut}(g_u(x_1),g_u(x_2))\Big] \\
	&\leq \Ex_{u \sim U_{\rm good}}\Ex_{x_1 \underset{\rho^*}{\sim} x_2}\Big[{\sf Cut}(g_u(x_1),g_u(x_2))\Big] + \nu(\eps) \\
	&\leq s^* + c(\eps) + \nu(\eps),
\end{align*}
and the proof of Lemma \ref{lem:sound} is concluded by setting $c_0(\eps) = c(\eps) + \nu(\eps)$.

\subsection{Proof of Lemma \ref{lem:u-good-1}}			\label{sec:u-good}

For contradiction, let us assume that $U_{\rm bad} := U \setminus U_{\rm good}$ satisfies $|U_{\rm bad}| \geq \nu(\eps) = \sqrt{\eps}$. For every $u \in U_{\rm bad}$, let us fix a label $i(u)$ such that 
$\Inf{i(u)}{g_u} \geq \tau$. Then, for every $u \in U_{\rm bad}$, using the convexity of influences we have 
\begin{align*}
\Ex_{v \sim N_{\cG}(u)}\Big[\Inf{\pi^{-1}_{u,v}(i(u))}{T_{1 - \mu}f_v}\Big]
&= 	\Ex_{v \sim N_{\cG}(u)}\Big[\Inf{i(u)}{\pi_{u,v}\circ T_{1 - \mu}f_v}\Big] 		\\
&\geq  \Inf{i(u)}{\Ex_{v \sim N_{\cG}(u)}\pi_{u,v}\circ T_{1 - \mu}f_v}		\\
&\geq \tau,
\end{align*}
and hence by averaging, for every such $u \in U$, we have
\begin{equation}			\label{eqn:infl-bound}
	\Pr_{v \sim N_{\cG}(u)}\left[\Inf{\pi^{-1}_{u,v}(i)}{T_{1 - \mu}f_v} \geq \frac{\tau}{2}\right] \geq \frac{\tau}{2}.
\end{equation}
For every left vertex $u \in U$ and every right vertex $v \in V$, we define lists $L_u$ and $L_v$ as follows:
\[
L_u := \left\{i \in [R]~\Big|~ \Inf{i}{g_u} \geq \frac{\tau}{2}\right\}
\qquad\quad\text{and}\quad\qquad
L_v := \left\{i \in [R]~\Big|~ \Inf{i}{T_{1 - \mu}f_v} \geq \frac{\tau}{2}\right\}.
\]
Note that by the standard bound on the sum of noisy influences (Proposition 2.54~\cite{ODonnell14}), we have $|L_u|,|L_v| \leq 2/(\mu\tau)$. Now consider the following randomized labeling scheme for $\cG$. 

\begin{itemize}
	\item For every $u \in U$ do the following independently. If $L_u \neq \emptyset$, then sample $\sigma(u) \sim L_u$, otherwise sample $\sigma(u) \sim [R]$ uniformly.
	\item For every $v \in V$ do the following independently. If $L_v \neq \emptyset$, then sample $\sigma(v) \sim L_v$, otherwise sample $\sigma(v) \sim [R]$ uniformly.
\end{itemize}

Let us bound the expected fraction of edges satisfied by the above labeling. To that end, let us fix a $u \in U_{\rm bad}$, and $v \in N_{\cG}(u)$ for which \eqref{eqn:infl-bound} holds. Then, randomizing over $\sigma$, the edge $(u,v)$ will be satisfied with probability at least
\[
	\frac{1}{|L_u|\cdot|L_v|} \geq \frac{\mu^2\tau^2}{4}.
\] 
Furthermore, note that by our assumption, the above event, i.e.,
\[
\{\text{$u \in U$ and $v \in N_{\cG}(u)$ satisfies the event in \eqref{eqn:infl-bound}}\},
\]
occurs with probability at least $\nu(\eps) \cdot \tau/2$. Therefore, the expected fraction of edges in $\cG$ that is satisfied by the above labeling is at least 
\[
\frac{\nu(\eps)\tau}{2} \cdot \frac{\mu^2\tau^2}{4} = \frac{\nu(\eps) \tau^3\mu^2}{8} \geq \frac{\sqrt{\eps} \cdot(2\eps^{1/8})^4}{8} =  2\epsilon,
\]
which contradicts the condition $\opt(\cG) \leq \eps$ stated in the setting of the lemma.

\section{Proof of Lemmas~~\ref{lem:gaussian-main}, \ref{lem:smoothed-boolean-transfer}, and~\ref{lem:smoothing-removal}}
\label{apx:MiST}
\subsection{Proof of Lemma~\ref{lem:gaussian-main}}

In this section, we prove Lemma~\ref{lem:gaussian-main}.  
We will use the following results.

\begin{theorem}[see Theorem 1.2~\cite{MosselNeeman15}]\label{thm:MN-functional-borell}
Let $f,g:\bbR^R\to[0,1]$ be measurable and let $X\sim_t Y$ be standard Gaussian vectors.
For $0<t<1$, we have
\[
        \E \Lambda_t(f(X),g(Y))\le \Lambda_t(\E f,\E g).
\]
For $-1<t<0$, the inequality is reversed.
\end{theorem}

\begin{theorem}[see~Theorem~1.4 in~\cite{MosselNeeman15}; presented as stated in Theorem~1.4 in~\cite{Heilman21}]
\label{thm:Heilman-set-robust}
Let $0<t<1$ and $0<a<1$.  There are constants $C(t,a)<\infty$ and
$c(t,a)=\frac14\frac{(1-t)(1-t^2)}{1+3t}$ such that the following holds.  If
$\Omega\subseteq\bbR^R$ is measurable with $\gamma_R(\Omega)=a$ and
\[
        \delta=\Stab_t(H)-\Stab_t(\Omega),
\]
where $H$ is a halfspace of Gaussian measure $a$, then, after rotating $H$ if necessary,
\[
        \gamma_R(\Omega\triangle H)\le C(t,a)\delta^{c(t,a)}.
\]
\end{theorem}

\begin{theorem}[Corollary of Theorem 1.6 in~\cite{Heilman21}]
\label{thm:Heilman-penalty}
Fix $0<t<1$, a closed interval $I\Subset(0,1)$, and $z_0>0$.  There is
$c=c(t,I,z_0)>0$ such that the following holds.  Let $a\in I$, and let
$\Omega\subseteq\bbR^R$ be measurable with $\gamma_R(\Omega)=a$.  Let
\[
        z_\Omega=\E[\one(G\in \Omega) \cdot G],
        \qquad
        \alpha=-\Phi^{-1}(a).
\]
Assume that $z_\Omega\ne0$ and $\|z_\Omega\|\ge z_0$.  Let
$\nu=\|z_\Omega\|^{-1}z_\Omega$, and let $H$ be the halfspace with
$\gamma_R(H)=a$ whose Gaussian centroid is a positive multiple of $z_\Omega$.
Define
\[
        \delta_t(\Omega)=\Stab_t(H)-\Stab_t(\Omega)
\]
and
\[
        \eta_{t/2}(\Omega)=\E\left[
        \Phi\left(\frac{(t/2)\langle\nu,G\rangle-\alpha}{\sqrt{1-t^2/4}}\right)
        (\one(G\in H)-\one(G \in \Omega))\right].
\]
Then $\eta_{t/2}(\Omega)\ge0$ and
\[
        c\,\eta_{t/2}(\Omega)\le \delta_t(\Omega).
\]
\end{theorem}

\begin{proof}
Apply Theorem 1.6 in~\cite{Heilman21} with $\beta=t/2$.  The explicit
constant in that theorem is positive for every fixed $t\in(0,1)$, every
$a\in I$, and every $z_0>0$.  Since $I$ is compact, its minimum over $a\in I$ is
positive.  This gives the stated constant $c$.  The nonnegativity of
$\eta_{t/2}$ is part of the same theorem.
\end{proof}

The next lemma presents some elementary facts about Gaussian random variables used in the proof of Lemma~\ref{lem:gaussian-main}. 

\begin{lemma}\label{lem:gaussian-facts}
The following statements hold.
\begin{enumerate}[label=(\alph*),ref=(\alph*)]
\item\label{it:small-density-centroid}
There is a universal constant $C$ such that for every measurable function $r:\bbR^R\to[0,1]$ with $p=\E[r(G)]\le1/2$, we have
\[
        \left\|\E[r(G)G]\right\|
        \le C p\sqrt{\ln\frac{e}{p}}.
\]

\item\label{it:penalty-symdiff}
Fix $0<t<1$, $0<\beta\le t$, and a compact interval $I\subset(0,1)$.  There is a constant $c=c(t,\beta,I)>0$ such that the following holds.  Let $a\in I$, let
\[
        H=\{x:\langle\nu,x\rangle\ge\alpha\},
        \qquad \|\nu\|=1,
        \qquad \gamma_R(H)=a,
\]
and let $A\subseteq\bbR^R$ be a subset of Gaussian measure $a$.  Define
\[
        w(x)=\Phi\left(\frac{\beta\langle\nu,x\rangle-\alpha}{\sqrt{1-\beta^2}}\right).
\]
Then
\[
        \E[w(G)(\one(G\in H)-\one(G\in A))]
        \ge c\,\gamma_R(A\triangle H)^2.
\]

\item\label{it:Lambda-facts}
The function $\Lambda_r$ satisfies the following identities
\[
        \Lambda_r(p,q)=\Lambda_r(q,p),
        \qquad
        \Lambda_r(1-p,1-q)-(1-p)(1-q)=\Lambda_r(p,q)-pq.
\]
For $-1<r<1$ and $0\le p,q\le1$,
\[        
        \begin{cases}
        \Lambda_r(p,q)-pq \ge 0, & \text{if } r\ge0,\\
        \Lambda_r(p,q)-pq \le 0, & \text{if } r\le0.
        \end{cases}
\]
Moreover,
\[
        \Gamma_r(1/2)=\frac14+\frac{1}{2\pi}\arcsin(r),
        \qquad
        \Gamma_r^\pm(0)=\frac{2}{\pi}\arcsin(r),
\]
and $\Gamma_r$ is Lipschitz on $[0,1]$.

\item\label{it:lambda-quadratic}
For every fixed $-1<r<1$ and every $m\in[-1,1]$,
\[
        |\Gamma_r^\pm(m)-\Gamma_r^\pm(0)|\le C_r m^2.
\]

\item\label{it:pointwise-gap}
Fix $0<t<1$ and let $d(a)=\min\{a,1-a\}$ for $a\in[0,1]$.  There is a constant $c_t>0$ such that, for all $p,q\in[0,1]$,
\[
        \Lambda_t(p,q)-pq\ge c_t d(p)d(q).
\]
Consequently, for every measurable function $u:\bbR^R\to[0,1]$,
\[
        \E\,d(u(G))\le C_t
        \left(\E\bigl[\Lambda_t(u(G),u(H))-u(G)u(H)\bigr]\right)^{1/2},
\]
where $G\sim_t H$.

\item\label{it:epigraph}
In this item, we show how to view a function \(F:\mathbb R^R\to[0,1]\)
as a subset of \(\mathbb R^{R+1}\), so that the mean and noise stability
of \(F\) become the corresponding quantities for this subset.
Let $u:\bbR^R\to[0,1]$ be measurable and define
\[
        A_u=\{(x,z)\in\bbR^R\times\bbR: z\le \Phi^{-1}(u(x))\}.
\]
Then $\gamma_{R+1}(A_u)=\E u$.  If $(G,Z)\sim_t (H,W)$ in $\bbR^{R+1}$, then
\[
        \Pr((G,Z)\in A_u,(H,W)\in A_u)
        =\E \Lambda_t(u(G),u(H)).
\]
Moreover, if $B\subseteq\bbR^{R+1}$ is a halfspace and
\[
        g(x)=\gamma_1\{z:(x,z)\in B\},
\]
then $g$ is either an indicator of a Gaussian halfspace in $\bbR^R$, or has the form
\[
        g(x)=\Phi(\alpha+\langle a,x\rangle)
\]
for some $a\in\bbR^R$ and $\alpha\in\bbR$.  In all cases,
\[
        \E |u-g|\le\gamma_{R+1}(A_u\triangle B).
\]
If, in addition, the vertical sections of $B$ are lower half-lines, or if $B$ is independent of the last coordinate, then equality holds:
\[
        \E |u-g|=\gamma_{R+1}(A_u\triangle B).
\]
Finally, the first $R$ coordinates of the centroid of $A_u$ are
\[
        \E[u(G)G],
\]
and the last coordinate is
\[
        -\E[\varphi(\Phi^{-1}(u(G)))]\le0.
\]
\end{enumerate}
\end{lemma}

\begin{proof}
We prove the six parts in order.

\emph{Proof of part \ref{it:small-density-centroid}.}
Let \(U(a)=\varphi(\Phi^{-1}(a))\) be the Gaussian isoperimetric function.  By Definition 5.26 and Proposition 5.27 in~\cite{ODonnell14},
\(U(p)=O(p\sqrt{\ln(e/p)})\) for \(0<p\le1/2\).  For every unit vector \(v\), the
quantity \(\E[r(G)\langle v,G\rangle]\), subject to \(0\le r\le1\) and
\(\E r=p\), is maximized by taking \(r\) to be the indicator of a Gaussian tail
in the direction \(v\).  Thus
\[
        \E[r(G)\langle v,G\rangle]\le U(p).
\]
Taking the supremum over unit vectors \(v\) gives the desired bound.

\emph{Proof of part \ref{it:penalty-symdiff}.}
Write $y=\langle\nu,x\rangle$.  Since $a$ ranges in a compact subinterval of $(0,1)$, the boundary point $\alpha$ ranges in a compact interval.  Choose $\ell_0>0$ so that, for all such $\alpha$, the derivative of
\[
        y\mapsto \Phi\left(\frac{\beta y-\alpha}{\sqrt{1-\beta^2}}\right)
\]
is bounded below by a positive constant $m_0$ on $[\alpha-\ell_0,\alpha+\ell_0]$.  Let $M=(2\pi)^{-1/2}$ be the maximum of the one-dimensional Gaussian density.

Let $B=H\setminus A$, $C=A\setminus H$, and $r_0=\gamma_R(B)=\gamma_R(C)=\gamma_R(A\triangle H)/2$.  Since $w$ depends only on $y$ and is increasing in $y$,
\[
\begin{aligned}
        \E[w(G)(\one(G\in H)-\one(G\in A))]
        &=\E[\one(G\in B)(w(G)-w(\alpha))]
          +\E[\one(G\in C)(w(\alpha)-w(G))].
\end{aligned}
\]
Let
\[
        \ell=\min\left\{\ell_0,\frac{r_0}{4M}\right\}.
\]
If $\ell=r_0/(4M)$, then the strip $\{|y-\alpha|\le\ell\}$ has Gaussian measure at most $2M\ell\le r_0/2$.  Therefore at least $r_0/2$ of the mass of $B$ lies in $\{y\ge\alpha+\ell\}$ and at least $r_0/2$ of the mass of $C$ lies in $\{y\le\alpha-\ell\}$.  Hence the displayed expectation is at least $m_0\ell r_0\ge c r_0^2$.

If instead $\ell=\ell_0$, then $r_0\ge4M\ell_0$.  Applying the preceding argument with $r_0$ replaced by $4M\ell_0$ gives a fixed positive lower bound for the expectation, and this fixed lower bound is at least a constant times $r_0^2$, since $r_0\le1$.  This proves the second part.

\emph{Proof of part \ref{it:Lambda-facts}.}
These formulas are from
Exercise 11.19 in~\cite{ODonnell14}.  The identities at $1/2$ are Sheppard's
formula~\cite{Sheppard99}; see also Section 5.2 and Exercise
11.19 in~\cite{ODonnell14}.

\emph{Proof of part \ref{it:lambda-quadratic}.}
Let $a_0=(1+m)/2$ and $q_0=\Phi^{-1}(a_0)$.  By the derivative formula in
Exercise 11.19(d) in~\cite{ODonnell14},
\[
        \Gamma_r'(a_0)=2\Phi\left(q_0\sqrt{\frac{1-r}{1+r}}\right).
\]
Hence
\[
        (\Gamma_r^\pm)'(m)=4\Phi\left(\Phi^{-1}\left(\frac{1+m}{2}\right)
        \sqrt{\frac{1-r}{1+r}}\right)-2.
\]
The right side is negative for $m<0$, zero at $m=0$, and positive for $m>0$.
Differentiating at $m=0$ gives
\[
        (\Gamma_r^\pm)''(0)=2\sqrt{\frac{1-r}{1+r}}>0.
\]
Thus $(\Gamma_r^\pm(m)-\Gamma_r^\pm(0))/m^2$ extends continuously to $m=0$.  Since $[-1,1]$ is compact, the quotient is bounded.  

\emph{Proof of part \ref{it:pointwise-gap}.}
For $p,q\in(0,1)$, part \ref{it:Lambda-facts} gives $\Lambda_t(p,q)-pq>0$. Consider
\[
        R(p,q)=\frac{\Lambda_t(p,q)-pq}{d(p)d(q)}
\]
on $(0,1)^2$.  The function $R$ is continuous and positive there.  It remains to check that its lower limit at the boundary of the square is positive.  The identities
$\Lambda_t(p,q)=\Lambda_t(q,p)$
and
$ \Lambda_t(1-p,1-q)-(1-p)(1-q)=\Lambda_t(p,q)-pq
$ 
reduce the boundary analysis to the following three cases.

First let $p\to0$ while $q$ stays in a compact subinterval of $(0,1)$.  Write $a=\Phi^{-1}(p)$ and $b=\Phi^{-1}(q)$, and let $\sigma_t=(1-t^2)^{1/2}$.  From the conditional representation
\[
        \Lambda_t(p,q)=\int_{-\infty}^a
        \varphi(x)\Phi\left(\frac{b-tx}{\sigma_t}\right)\,dx,
\]
the conditional factor tends to $1$ uniformly for $x\le a$ and for $b$ in any fixed bounded interval.  Hence, for all sufficiently small $p$,
\[
        \Lambda_t(p,q)\ge \frac{1+q}{2}\,p.
\]
Therefore
\[
        \Lambda_t(p,q)-pq\ge \frac{1-q}{2}\,p.
\]
Since $q$ stays away from $0$ and $1$, this is at least a positive constant times $d(p)d(q)$.

Second let $(p,q)\to(0,0)$.  By symmetry assume $p\le q$.  Then $a=\Phi^{-1}(p)\le b=\Phi^{-1}(q)<0$.  For $x\le a$,
\[
        \Phi\left(\frac{b-tx}{\sigma_t}\right)
        \ge \Phi\left(\frac{b-ta}{\sigma_t}\right)
        \ge \Phi\left(b\sqrt{\frac{1-t}{1+t}}\right).
\]
The last inequality uses $a\le b<0$.  Since $0<\sqrt{(1-t)/(1+t)}<1$, l'Hospital's rule gives
\[
        \frac{\Phi\left(b\sqrt{(1-t)/(1+t)}\right)}{\Phi(b)}\longrightarrow\infty
        \qquad\text{as } b\to-\infty .
\]
Thus, for $q$ small enough,
\[
        \Lambda_t(p,q)\ge 2pq,
\]
and hence $\Lambda_t(p,q)-pq\ge pq=d(p)d(q)$.

Third let $(p,q)\to(0,1)$.  Again let $a=\Phi^{-1}(p)$ and $b=\Phi^{-1}(q)$, so $a\to-\infty$ and $b\to\infty$.  Since
\[
        \Lambda_t(p,q)=p-\Pr(G\le a, H>b),
\]
we have
\[
        \Lambda_t(p,q)-pq
        =p(1-q)-\Pr(G\le a,H>b)  
        \ge p\left((1-q)-\Phi\left(\frac{ta-b}{\sigma_t}\right)\right).
\]
Here the inequality follows from the conditional estimate
\[
        \Pr(H>b\mid G=x)
        =\Phi\left(\frac{tx-b}{\sigma_t}\right)
        \le \Phi\left(\frac{ta-b}{\sigma_t}\right)
        \qquad (x\le a).
\]
Because $ta\le0$ and $\sigma_t<1$,
\[
        \Phi\left(\frac{ta-b}{\sigma_t}\right)
        \le \Phi\left(-\frac{b}{\sigma_t}\right).
\]
By l'Hospital's rule,
\[
        \frac{\Phi(-b/\sigma_t)}{1-q}
        =\frac{\Phi(-b/\sigma_t)}{\Phi(-b)}\longrightarrow0
        \qquad\text{as }b\to\infty .
\]
Consequently $\Lambda_t(p,q)-pq\ge \frac12 p(1-q)=\frac12d(p)d(q)$ near this mixed corner.  The corner $(1,0)$ follows by exchange symmetry, and the corner $(1,1)$ follows from the second case by complementing both variables.

We have shown that $R$ has a positive lower limit at every boundary point.  Since $R$ is positive and continuous in the interior, compactness gives $c_t>0$ such that
\[
        \Lambda_t(p,q)-pq\ge c_t d(p)d(q)
\]
for all $p,q\in[0,1]$.
Now set $r_1(x)=d(u(x))$.  The pointwise inequality gives
\[
        \E\bigl[\Lambda_t(u(G),u(H))-u(G)u(H)\bigr]
        \ge c_t\E[r_1(G)r_1(H)].
\]
Write $G=\sqrt t\,W+\sqrt{1-t}\,Z_1$ and $H=\sqrt t\,W+\sqrt{1-t}\,Z_2$, where $W,Z_1,Z_2$ are independent standard Gaussian vectors.  Conditioning on $W$ gives
\[
        \E[r_1(G)r_1(H)]
        =\E\left[\E[r_1(G)\mid W]^2\right]
        \ge \E\left[\E[r_1(G)\mid W]\right]^2
        = (\E r_1(G))^2.
\]
Combining the last two inequalities proves the desired estimate.

\emph{Proof of part \ref{it:epigraph}.}
This is the standard epigraph construction used to pass from set Borell to functional Borell theorem~\cite{Borell85,MosselNeeman15}.  The identity $\gamma_{R+1}(A_u)=\E u$ follows by conditioning on $x$.  The noise-stability identity follows by conditioning on $(G,H)$; the conditional probability that $Z\le\Phi^{-1}(u(G))$ and $W\le\Phi^{-1}(u(H))$ is exactly $\Lambda_t(u(G),u(H))$.

For the halfspace assertion, write $B$ as
\[
        B=\{(x,z): b+\langle a,x\rangle+cz\ge0\}.
\]
If $c=0$, then each vertical section has measure either $0$ or $1$, so $g$ is an indicator of a halfspace in $\bbR^R$.  If $c\ne0$, each vertical section is a half-line.  Its one-dimensional Gaussian measure is $\Phi(\alpha+\langle a',x\rangle)$ for suitable $\alpha,a'$.

For each fixed $x$, the vertical section of $A_u$ has Gaussian measure $u(x)$, while the vertical section of $B$ has Gaussian measure $g(x)$.  For any two measurable subsets $E,F\subseteq\bbR$,
\[
        |\gamma_1(E)-\gamma_1(F)|\le \gamma_1(E\triangle F).
\]
Applying this to the two vertical sections and integrating in $x$ gives
\[
        \E |u-g|\le\gamma_{R+1}(A_u\triangle B).
\]
If the vertical section of $B$ is also a lower half-line, then the two vertical sections are nested half-lines, so the preceding inequality is an equality for each $x$.  If $c=0$, the section of $B$ is either all of $\bbR$ or empty for each $x$, and equality again holds in each vertical fiber.

The centroid identities are Fubini's theorem.  For the first $R$ coordinates,
\[
        \E[\one((G,Z)\in A_u)\cdot G]
        =\E[u(G)G].
\]
For the last coordinate,
\[
        \int_{-\infty}^{\Phi^{-1}(u(x))} z\varphi(z)\,dz
        =-\varphi(\Phi^{-1}(u(x))),
\]
with the convention that the right side is $0$ when $u(x)\in\{0,1\}$.
\end{proof}

\begin{proof}[Proof of Lemma~\ref{lem:gaussian-main}]

By Theorem~\ref{thm:MN-functional-borell} and part \ref{it:Lambda-facts} of Lemma~\ref{lem:gaussian-facts}, the two quantities
\[
        \Delta_{\mathrm{Borell}}
        =\Gamma_t(m)-\E \Lambda_t(F(G),F(H))
\]
and
\[
        \Delta_{01}
        =\E\bigl[\Lambda_t(F(G),F(H))-F(G)F(H)\bigr]
\]
are nonnegative, and
\[
        \Delta=\Delta_{\mathrm{Borell}}+\Delta_{01}.
\]
By part \ref{it:pointwise-gap} of Lemma~\ref{lem:gaussian-facts},
\begin{equation}\label{eq:centroid-inverse-fractional}
        \E\min\{F,1-F\}\le C_t\Delta_{01}^{1/2}\le C_t\Delta^{1/2}.
\end{equation}

Apply part \ref{it:epigraph} of Lemma~\ref{lem:gaussian-facts} and set $A=A_F\subseteq\bbR^{R+1}$.  Then $\gamma_{R+1}(A)=m$, the set stability gap of $A$ is $\Delta_{\mathrm{Borell}}$, and the first $R$ coordinates of its centroid are $\widehat F^{(1)}$.

Let
\[
        m_I=\min_{a\in I}\varphi(\Phi^{-1}(a))>0.
\]
Every halfspace of Gaussian measure $m\in I$ has centroid norm at least $m_I$.  Since $I\Subset(0,1)$, we use the two robust Borell estimates below with constants chosen uniformly for $m\in I$.  Choose $\varepsilon_I>0$ so small that
\[
        C'\varepsilon_I\sqrt{\ln(e/\varepsilon_I)}\le m_I/2,
\]
where $C'$ is the universal constant obtained in the next paragraph from part \ref{it:small-density-centroid} of Lemma~\ref{lem:gaussian-facts}.  By Theorem~\ref{thm:Heilman-set-robust}, there is $\delta_I>0$ such that $\Delta_{\mathrm{Borell}}\le\delta_I$ implies that $A$ is within Gaussian symmetric difference $\varepsilon_I$ of some halfspace $H_0$ of measure $m$.  Therefore, if $\Delta_{\mathrm{Borell}}\le\delta_I$, the centroid $z_A$ of $A$ satisfies
\begin{equation}\label{eq:epigraph-centroid-lower}
        \|z_A\|\ge m_I/2.
\end{equation}
Indeed, let $z_{H_0}=\E[\one(G\in H_0)\cdot G]$ and write
\[
        E_0=A\setminus H_0,
        \qquad
        F_0=H_0\setminus A.
\]
Since $A$ and $H_0$ have the same measure, $\gamma(E_0)=\gamma(F_0)=\gamma(A\triangle H_0)/2\le\varepsilon_I/2$.  Hence
\[
        z_A-z_{H_0}=\E[\one(G\in E_0)\cdot G]-\E[\one(G\in F_0)G].
\]
Applying part \ref{it:small-density-centroid} of Lemma~\ref{lem:gaussian-facts} separately to indicator functions $\one_{E_0}$ and $\one_{F_0}$ gives
\[
        \|z_A-z_{H_0}\|
        \le C'\varepsilon_I\sqrt{\ln(e/\varepsilon_I)}
        \le m_I/2.
\]
Since $\|z_{H_0}\|\ge m_I$, the triangle inequality gives \eqref{eq:epigraph-centroid-lower}.

Assume first that $\Delta\le\delta_I$.  Then \eqref{eq:epigraph-centroid-lower} holds.  We first get the lower bound on the first Hermite vector.  Let
\[
        V(p)=\varphi(\Phi^{-1}(p)),
        \qquad d(p)=\min\{p,1-p\}.
\]
By Proposition 5.27 in~\cite{ODonnell14},
$V(p)\le C d(p)\sqrt{\ln(e/d(p))}$, with the right-hand side interpreted as $0$ when $d(p)=0$.  The function
$x\mapsto x\sqrt{\ln(e/x)}$ is increasing and concave on $[0,1]$.  Hence, using
\eqref{eq:centroid-inverse-fractional} and Jensen's inequality,
\[
        |(z_A)_{R+1}|
        =\E V(F(G))
        \le C\Delta^{1/2}\sqrt{\ln(e/\Delta)}
\]
for $0<\Delta\le\delta_I$; for $\Delta=0$ the left-hand side is $0$.  After decreasing
$\delta_I$, we may assume that $|(z_A)_{R+1}|\le m_I/4$.  Since the first $R$
coordinates of $z_A$ are $\widehat F^{(1)}$, \eqref{eq:epigraph-centroid-lower}
gives
\[
        \|\widehat F^{(1)}\|
        \ge \sqrt{(m_I/2)^2-(m_I/4)^2}
        \ge m_I/4.
\]
Set $c_0=m_I/4$.

Apply Theorem~\ref{thm:Heilman-penalty} to $A$ with $z_0=m_I/2$.  Let
$\widetilde B\subseteq\bbR^{R+1}$ be the halfspace whose centroid is a positive
multiple of $z_A$ and whose Gaussian measure is $m$.  By the last-coordinate
formula in part \ref{it:epigraph} of Lemma~\ref{lem:gaussian-facts}, the last
coordinate of $z_A$ is nonpositive; hence, if $\widetilde B$ has a nonzero
coefficient in the last coordinate, its vertical sections are lower half-lines.
Thus equality in the oriented part of part \ref{it:epigraph} of
Lemma~\ref{lem:gaussian-facts} holds in this case.  The proof below only uses the
weaker inequality, which holds for every orientation.  By
Theorem~\ref{thm:Heilman-penalty} and part \ref{it:penalty-symdiff} of
Lemma~\ref{lem:gaussian-facts},
\begin{equation}\label{eq:epigraph-close-centroid-halfspace}
        \gamma_{R+1}(A\triangle\widetilde B)
        \le C\Delta_{\mathrm{Borell}}^{1/2}
        \le C\Delta^{1/2}.
\end{equation}
Let
\[
        g(x)=\gamma_1\{z:(x,z)\in\widetilde B\}.
\]
By part \ref{it:epigraph} of Lemma~\ref{lem:gaussian-facts},
\begin{equation}\label{eq:u-close-section-g}
        \E|F-g|\le C\Delta^{1/2}.
\end{equation}
We now identify the direction of the threshold obtained from $g$.  Write the
normal vector of $\widetilde B$ as $(v,c)\in\bbR^R\times\bbR$.  Since the centroid of
$\widetilde B$ is a positive multiple of $z_A$, its first $R$ coordinates $v$ are
a positive multiple of $\widehat F^{(1)}$.  If $c<0$, the vertical sections of
$\widetilde B$ are lower half-lines and $g$ has the form
$\Phi(\alpha+\langle a,x\rangle)$ with $a$ a positive multiple of $v$.  If
$c>0$, the last coordinate of the centroid of $\widetilde B$ would be positive,
whereas the last coordinate of the centroid of $A_F$ is
\[
        -\E[\varphi(\Phi^{-1}(F(G)))]\le0.
\]
Thus $c>0$ cannot occur unless this last coordinate is zero; in that exceptional
case the first-level direction has $c=0$.  If $c=0$, then $g$ is already an
indicator of a Gaussian LTF with normal vector $v$.  Hence
\[
        B=\{x:g(x)\ge1/2\}
\]
is a Gaussian LTF whose normal vector is $\widehat F^{(1)}$.  Since
$x\mapsto\min\{x,1-x\}$ is $1$-Lipschitz on $[0,1]$,
\[
        \E |g-\one_B| =\E\min\{g,1-g\}  \le \E\min\{F,1-F\}+\E|F-g| 
        \le C\Delta^{1/2}
\]
by \eqref{eq:centroid-inverse-fractional} and \eqref{eq:u-close-section-g}.  Therefore
\[
        \E |F-\one_B|\le C\Delta^{1/2}.
\]
This proves the lemma, with $\delta_0$ chosen small enough.

\end{proof}

\subsection{Proof of Lemma~\ref{lem:smoothed-boolean-transfer}}

This section uses the sign cube $\Cube$.  We use an affine transformation $x_i=(1+\sigma_i)/2$ to go from $\Cube$ to $\ZeroCube$.  Linear threshold function $\theta+\sum_i w_i\sigma_i\ge0$ on $\Cube$ corresponds to linear threshold function
$\theta-\sum_i w_i+2\sum_i w_i x_i\ge0$ on $\ZeroCube$.

For $u:\Cube\to[0,1]$ and $\lambda\in(0,1)$, define
\[
        P_{u,\lambda}(x)=\sum_{S\subseteq[R]}\lambda^{|S|}\widehat u(S)\prod_{i\in S}x_i,
        \qquad x\in\bbR^R.
\]
That is, $P_{u,\lambda}$ has the same coefficients as $T_\lambda u$, but is
evaluated on Gaussian space.  For a function $P$, let $\overline P$ be the
truncation of $P$ to $[0,1]$.

\begin{theorem}[Smoothed invariance; \cite{MOO10} and Theorem 7.1 in~\cite{MosselNeeman15}]
\label{thm:MN-invariance}
For some constants $c > 0$ and $C>0$, the following holds.
Assume that for $u:\Cube\to[0,1]$. we have
$\max_i\Influence_i(u)\le\tau$.  Then
\[
        \E\left[\bigl(P_{u,\lambda}(G)-\overline{P_{u,\lambda}}(G)\bigr)^2\right]
        \le C \tau^{c \lambda}.
\]
\end{theorem}

\begin{claim}[Gaussian approximation for regular Boolean LTFs]\label{clm:MN-ltf-smoothing}
There is a universal constant $C<\infty$ such that the following holds.  Let
$L(\sigma)=\one\left(\langle w,\sigma\rangle\ge\theta\right)$ with
$\|w\|=1$. Let $\mu = |w|_\infty$ and
$H(x)=\one\left(\langle w,x\rangle\ge\theta\right)$.  For every fixed
$\lambda\in[1/2,1)$,
\[
        \E|P_{L,\lambda}(G)-H(G)|
        \le C\sqrt{1-\lambda}+o_{\mu\downarrow 0}(1),
\]
where the $o_{\mu\downarrow 0}(1)$ term may depend on $\lambda$.
\end{claim}

\begin{proof}
By Lemma~7.3 of Mossel--Neeman~\cite{MosselNeeman15}, applied with
\(\eta=\log\nicefrac{1}{\lambda}\), and by Cauchy--Schwarz,
\[
        \E\left|P_{L,\lambda}(G)-\one\left(\langle w,G\rangle\ge\theta\right)\right|
  \le C\sqrt{\log\nicefrac{1}{\lambda}+\mu^{c\log\nicefrac{1}{\lambda}}} .
\]
Since \(1-\lambda\le \log\nicefrac{1}{\lambda}\le 2(1-\lambda)\) for
\(\lambda\in[1/2,1)\), the upper bound is
\[
  O(\sqrt{1-\lambda})+o_{\mu\downarrow 0}(1)
\]
for every fixed \(\lambda<1\).
\end{proof}

\begin{theorem}[Anti-concentration for regular linear forms; see Theorem 30 in~\cite{MORS10}]
\label{thm:MORS-anti-concentration}
There is a universal constant \(C<\infty\) such that the following holds.  Let
\(\ell(X)=\sum_i a_iX_i\), where \(X\) is uniform on \(\Cube\).  Let
\(\sigma=\|a\|\) and \(\mu=\max_i |a_i|\).  Then, for every
\(\theta\in\bbR\) and every \(\alpha\ge\mu\),
\[
        \Pr(|\ell(X)-\theta|\le \alpha)\le C\alpha/\sigma .
\]
\end{theorem}

\begin{theorem}[Smoothed Lipschitz invariance;
see~\cite{MOO10} and also Corollary 11.69 in~\cite{ODonnell14}]
\label{thm:ODonnell-smoothed}
There is a universal constant $C<\infty$ such that the following holds.  Consider
$h:\Cube\to\bbR$ with $\Var(h)\le1$.  Let $0<\varepsilon\le1$ and
$0<\delta\le1/20$, and assume that, for every coordinate $i$,
\[
        \sum_{S: i\in S}(1-\delta)^{|S|-1}\widehat h(S)^2\le\varepsilon.
\]
Then, for every $c$-Lipschitz function $\psi:\bbR\to\bbR$,
\[
        \left|\E\psi(T_{1-\delta}h(X))-\E\psi(T_{1-\delta}h(G))\right|
        \le Cc\,\varepsilon^{\delta/3}.
\]
Here $T_{1-\delta}h(G)$ means that the multilinear expansion of $T_{1-\delta}h$
is evaluated at a standard Gaussian vector $G$.
\end{theorem}

The next lemma lists the standard consequences of these theorems that we will need
below.

\begin{lemma}[Boolean smoothing facts]\label{lem:boolean-smoothing-facts}
The following statements hold.
\begin{enumerate}[label=(\alph*),ref=(\alph*)]
\item\label{it:smoothed-l1-invariance}
Fix $\lambda\in[19/20,1)$ and $M<\infty$.  Let $q$ be a real-valued function on $\Cube$.  Assume that $\E q(X)^2\le M$ and $\max_i\Influence_i(q)\le\kappa$.  Then
\[
        \left|\E|T_\lambda q(X)|-\E|P_{q,\lambda}(G)|\right|=o_{\kappa\downarrow 0}(1),
\]
where the $o_{\kappa\downarrow 0}(1)$ term may depend on $\lambda$ and $M$.

\item\label{it:regular-ltf}
If $L(\sigma)=\one\left(\langle w,\sigma\rangle\ge\theta\right)$, $\|w\|^2=1$ and
$\|w\|_\infty\le\mu$, then
\[
        \max_i\Influence_i(L)=O(\mu),
        \qquad
        \Pr(\langle w,X\rangle=\theta)=O(\mu).
\]

\item\label{it:stability-smoothing}
Fix $0<t<1$.  For every $u:\Cube\to[0,1]$ and every $\lambda\in[1/2,1)$,
\[
        0\le \Stab_t(u)-\Stab_{t\lambda^2}(u)\le C_t(1-\lambda).
\]
\end{enumerate}
\end{lemma}

\begin{proof}
For part \ref{it:smoothed-l1-invariance}, assume \(M>0\); otherwise there is
nothing to prove.  Apply Theorem~\ref{thm:ODonnell-smoothed} to
\(h=M^{-1/2}q\) with \(\delta=1-\lambda\) and
\(\varepsilon=\min\{1,\kappa/M\}\).  Since \(\lambda\ge19/20\), the theorem
applies.  Taking \(\psi(x)=|x|\) and multiplying by \(M^{1/2}\) gives the claim.

For part \ref{it:regular-ltf}, assume first that \(\mu\le1/2\); otherwise the
bounds are trivial.  Fix \(i\) and write \(W_i=\sum_{j\ne i}w_jX_j\).  Coordinate
\(i\) can affect \(L\) only when \(|W_i-\theta|\le |w_i|\le\mu\).  The coefficients
of \(W_i\) are bounded by \(\mu\), and \(\operatorname{Var}(W_i)\ge3/4\).  Applying
Theorem~\ref{thm:MORS-anti-concentration} to \(W_i\) with interval length \(\mu\),
after normalizing \(W_i\), gives
\[
        \Pr(|W_i-\theta|\le |w_i|)
        \le \Pr(|W_i-\theta|\le \mu)=O(\mu).
\]
Thus \(\Influence_i(L)=O(\mu)\).  The same theorem applied to \(\langle w,X\rangle\)
and the interval \([\theta-\mu,\theta+\mu]\) gives
\(\Pr(\langle w,X\rangle=\theta)=O(\mu)\).

Finally, using the Fourier formula for $\Stab$, we get
\[
        \Stab_t(u)-\Stab_{t\lambda^2}(u)
        =\sum_{S\subseteq[R]} t^{|S|}(1-\lambda^{2|S|})\widehat u(S)^2.
\]
The summand for \(S=\varnothing\) is zero and the other summands are nonnegative.
For \(k\ge1\), \(t^k(1-\lambda^{2k})\le C_t(1-\lambda)\).  Since
\(\sum_S\widehat u(S)^2\le1\), part \ref{it:stability-smoothing} follows.
\end{proof}

\begin{proof}[Proof of Lemma~\ref{lem:smoothed-boolean-transfer}]
Let $a=\E u$ and $D=\pos{\Gamma_t(a)-\Stab_t(u)}$.  Let
$P_\lambda=P_{u,\lambda}$ and $\overline P_\lambda=\overline{P_{u,\lambda}}$.
By Theorem~\ref{thm:MN-invariance},
\[
        r_\lambda(\tau)=\|P_\lambda-\overline P_\lambda\|_{2}=o_{\tau\downarrow 0}(1).
\]
The mean of $P_\lambda$ is $a$, so $|\E\overline P_\lambda-a|\le r_\lambda(\tau)$.
Also the Gaussian stability of $P_\lambda$ at correlation $t$ is
$\Stab_{t\lambda^2}(u)$.  Using Lemma~\ref{lem:boolean-smoothing-facts}, the
Lipschitz bound for $\Gamma_t$, and Cauchy--Schwarz,
\[
        \Gamma_t(\E\overline P_\lambda)-
        \E[\overline P_\lambda(G)\overline P_\lambda(H)]
        \le D+C_t(1-\lambda)+o_{\tau\downarrow 0}(1),
        \qquad G\sim_t H.
\]
For fixed $\lambda$, since $r_\lambda(\tau)=o_{\tau\downarrow 0}(1)$ and $a\in[1/4,3/4]$, we
may assume that $\E\overline P_\lambda\in[1/5,4/5]$.  If the last displayed
quantity is larger than the constant $\delta_0$ from Lemma~\ref{lem:gaussian-main},
then the desired estimate is trivial after increasing $C_t$, since the left-hand
side is at most $1$; in this case we take $L_\lambda$ to be a constant LTF.  Thus we may assume that this quantity is smaller than
$\delta_0$.  Applying Lemma~\ref{lem:gaussian-main}, we get a
Gaussian LTF $H_\lambda=\one\left(\langle w_\lambda,x\rangle\ge\theta_\lambda\right)$
with $\|w_\lambda\|_2=1$ such that
\[
        \E|\overline P_\lambda-H_\lambda|
        \le C_t\sqrt D+C_t\sqrt{1-\lambda}+o_{\tau\downarrow 0}(1).
\]
Moreover, its normal vector is the normalized first Hermite vector of
$\overline P_\lambda$.  For each $i$,
$\E[P_\lambda(G)G_i]=\lambda\widehat u(\{i\})$ and
$|\widehat u(\{i\})|^2\le\Influence_i(u)\le\tau$.  Also,
\[
        |\E[\overline P_\lambda(G)G_i]-\E[P_\lambda(G)G_i]|
        \le \|\overline P_\lambda-P_\lambda\|_{2}=r_\lambda(\tau).
\]
Hence each coordinate of the first Hermite vector of $\overline P_\lambda$ is at most
$\sqrt\tau+r_\lambda(\tau)$ in absolute value.  Since this vector has norm bounded
below by Lemma~\ref{lem:gaussian-main},
$\|w_\lambda\|_\infty=o_{\tau\downarrow 0}(1)$.
Let $L_\lambda(\sigma)=\one\left(\langle w_\lambda,\sigma\rangle\ge\theta_\lambda\right)$.
By Claim~\ref{clm:MN-ltf-smoothing} and the preceding regularity bound,
\[
        \E|P_{L_\lambda,\lambda}(G)-H_\lambda(G)|
        \le C\sqrt{1-\lambda}+o_{\tau\downarrow 0}(1).
\]
Using the triangle inequality,
\[
\begin{aligned}
        \E|P_{u,\lambda}(G)-P_{L_\lambda,\lambda}(G)|
        &\le \E|P_{u,\lambda}(G)-\overline P_\lambda(G)|
        +\E|\overline P_\lambda(G)-H_\lambda(G)| \\
        &\qquad +\E|H_\lambda(G)-P_{L_\lambda,\lambda}(G)|.
\end{aligned}
\]
The first term is at most $r_\lambda(\tau)$, the second term is bounded by
Lemma~\ref{lem:gaussian-main}, and the third term is bounded by the last
displayed estimate.  Hence
\[
        \E|P_{u,\lambda}(G)-P_{L_\lambda,\lambda}(G)|
        \le C_t\sqrt D+C_t\sqrt{1-\lambda}+o_{\tau\downarrow 0}(1).
\]
It remains to move this estimate from Gaussian space back to the cube.  Let
$q=u-L_\lambda$.  Then
$T_\lambda q=T_\lambda u-T_\lambda L_\lambda$, and
$P_{u,\lambda}-P_{L_\lambda,\lambda}$ is the same multilinear polynomial
evaluated on Gaussian space.  Also $\E q^2\le1$.  By part \ref{it:regular-ltf} of
Lemma~\ref{lem:boolean-smoothing-facts},
\[
        \Influence_i(q)\le 2\Influence_i(u)+2\Influence_i(L_\lambda)=o_{\tau\downarrow 0}(1)
\]
for every $i$.  Part \ref{it:smoothed-l1-invariance} therefore gives
\[
        \E|T_\lambda u-T_\lambda L_\lambda|
        \le C_t\sqrt D+C_t\sqrt{1-\lambda}+o_{\tau\downarrow 0}(1).
\]
This concludes the proof.
\end{proof}

\subsection{Proof of Lemma~\ref{lem:smoothing-removal}}

We use the following estimate to remove the noise operator from the approximating LTF.

\begin{theorem}[Noise sensitivity of weighted majorities; see Theorem 1.1 in~\cite{Peres04}]
\label{thm:Peres}
Let $X$ be uniform on $\{-1,1\}^R$.  Let $N_\varepsilon(X)$ be obtained from $X$
by flipping each coordinate independently with probability $\varepsilon\le1/2$.
For nonzero weights $w_1,\ldots,w_R$ and a threshold $\theta\in\bbR$, define
$f(x)=\sgn(\sum_i w_i x_i-\theta)$, with $\sgn(0)=0$.  Then
\[
        \Pr\bigl(f(X)\ne f(N_\varepsilon(X))\bigr)=O(\sqrt\varepsilon).
\]
\end{theorem}

\begin{proof}[Proof of Lemma~\ref{lem:smoothing-removal}]
If $L$ is constant, then $T_\lambda L=L$ and the claim is immediate.  Assume now
that $L$ is nonconstant and has normalized coefficients with $\|w\|_\infty=\mu$.
Assume that $X$ and $Y$ are $\lambda$-correlated, then
corresponding bits of $X$ and $Y$ are equal with probability $\frac{1+\lambda}{2}$. Accordingly, $Y$ is distributed like $N_{\varepsilon}(X)$ with $\varepsilon = (1-\lambda)/2$.  Since $L$ is
$\{0,1\}$-valued,
$\|L(X)-T_\lambda L(X)\|_1=\Pr(L(X)\ne L(Y))$.

Discarding coordinates with $w_i=0$, which do not affect $L$, we  apply Theorem~\ref{thm:Peres} to
$F(x)=\sgn(\langle w,x\rangle-\theta)$, with $\sgn(0)=0$.  If
$L(X)\ne L(Y)$, then either $F(X)\ne F(Y)$ or one of $X,Y$ lies on the separating hyperplane.  By Theorem~\ref{thm:Peres} and part \ref{it:regular-ltf} of
Lemma~\ref{lem:boolean-smoothing-facts},
\[
        \|L-T_\lambda L\|_1\le C\sqrt{1-\lambda}+o_{\mu\downarrow 0}(1).
\]

It remains to compare $u$ and $T_\lambda u$.  Since $0\le u,T_\lambda u\le1$ and
$L$ is $\{0,1\}$-valued,
$|u-L|=u+L-2uL$ and $|T_\lambda u-L|=T_\lambda u+L-2(T_\lambda u)L$. 
Write,
\begin{align*}
\|u-L\|_1 & =\E[u]+\E[L]-2\E[uL]\\
\|T_\lambda u-L\|_1 &=\E[T_\lambda u] + \E[L]-2\E[(T_\lambda u)L]
\end{align*}
Observe that the first two terms are equal in both formulas, since the noise-operator preserves the expectation.
We get
\[
\left|\|u-L\|_1-\|T_\lambda u-L\|_1\right| \le 
2|\E[uL - (T_\lambda u)L]|=
2|\E[uL - u T_\lambda L]| =
2|\E[u(L - T_\lambda L)]| \leq
2\|L-T_\lambda L\|_1.
\]
Applying the triangle inequality
$
        \|T_\lambda u-L\|_1
        \le \|T_\lambda u-T_\lambda L\|_1+\|T_\lambda L-L\|_1
$, we get
\[
        \|u-L\|_1
        \le \|T_\lambda u-T_\lambda L\|_1+C\sqrt{1-\lambda}+o_{\mu\downarrow 0}(1).
\]
\end{proof}

\end{document}